\documentclass[11pt]{article}

\usepackage{xcolor}
\definecolor{ForestGreen}{rgb}{0.1333,0.5451,0.1333}
\definecolor{DarkRed}{rgb}{0.80,0,0}
\definecolor{Red}{rgb}{1,0,0}
\usepackage[linktocpage=true,
pagebackref=true,colorlinks,
linkcolor=DarkRed,citecolor=ForestGreen,
bookmarks,bookmarksopen,bookmarksnumbered]
{hyperref}

\usepackage{fullpage}
\usepackage[utf8]{inputenc}
\usepackage[american]{babel}
\usepackage[normalem]{ulem}
\usepackage{amsmath, amssymb, cases, amsthm}
\usepackage{thmtools}
\usepackage[shortlabels]{enumitem}
\usepackage{mdframed}
\usepackage{bbm}
\usepackage{bm}
\usepackage{microtype}
\usepackage{xcolor}
\usepackage{makecell}
\usepackage{mathtools}
\usepackage{float}
\usepackage{multirow}
\usepackage{soul}

\usepackage{easyReview} 
\usepackage{graphics}
\usepackage{mathrsfs}
\usepackage[ruled,vlined,linesnumbered]{algorithm2e}
\SetFuncSty{textsc}
\usepackage[capitalize,noabbrev]{cleveref}

\Crefname{claim}{Claim}{Claims}
\Crefname{condition}{Condition}{Conditions}
\Crefname{assumption}{Assumption}{Assumptions}
\Crefname{scenario}{Scenario}{Scenarios}

\usepackage{caption}
\usepackage{subcaption}

\usepackage{tcolorbox}

\declaretheorem[numberwithin=section,refname={Theorem,Theorems},Refname={Theorem,Theorems}]{theorem}

\declaretheorem[numberlike=theorem]{lemma}
\declaretheorem[numberlike=theorem,name=Lemma,refname={Lemma,Lemmas},Refname={Lemma,Lemmas}]{lem}

\declaretheorem[numberlike=theorem,name=Corollary,refname={Corollary,Corollaries},Refname={Corollary,Corollaries}]{cor}
\declaretheorem[numberlike=theorem,style=definition]{definition}
\declaretheorem[numberlike=theorem,style=definition,name=Definition,refname={Definition,Definitions},Refname={Definition,Definitions}]{defn}

\declaretheorem[numberlike=theorem,style=definition]{assumption}
\declaretheorem[numberlike=theorem]{claim}
\declaretheorem[numberlike=theorem,style=remark]{remark}

\declaretheorem[numberlike=theorem,refname={Fact,Facts},Refname={Fact,Facts},name={Fact}]{fact}

\declaretheorem[numberlike=theorem, refname={Observation,Observations},Refname={Observation,Observations},name={Observation}]{observation}

\newcommand{\vol}{\mathrm{vol}}
\renewcommand{\deg}{\mathrm{deg}}
\newcommand{\Bdeg}{\boldsymbol{\mathrm{deg}}}
\newcommand{\Bvol}{\boldsymbol{\mathrm{vol}}}

\newcommand{\polylog}{\mathrm{polylog}}
\newcommand{\poly}{\mathrm{poly}}
\newcommand{\defeq}{\stackrel{\mathrm{{\scriptscriptstyle def}}}{=}}

\renewcommand{\cong}{\mathrm{cong}}
\newcommand{\supp}{\mathrm{supp}}

\newcommand{\head}{\mathrm{head}}
\newcommand{\tail}{\mathrm{tail}}
\newcommand{\SCC}{\mathrm{SCC}}
\newcommand{\rev}[1]{\overleftarrow{#1}}
\newcommand{\forward}[1]{\overrightarrow{#1}}
\newcommand{\backward}[1]{\overleftarrow{#1}}
\newcommand{\alg}[2]{\textsc{#1(}#2\textsc{)}}

\newcommand{\Ba}{\boldsymbol{a}}
\newcommand{\Bb}{\boldsymbol{b}}
\newcommand{\Bc}{\boldsymbol{c}}
\newcommand{\Bw}{\boldsymbol{w}}
\newcommand{\Bd}{\boldsymbol{d}}
\newcommand{\Bf}{\boldsymbol{f}}

\newcommand{\Bx}{\boldsymbol{x}}

\newcommand{\BF}{\boldsymbol{F}}
\newcommand{\BM}{\boldsymbol{M}}

\newcommand{\Bsource}{\boldsymbol{\Delta}}
\newcommand{\Bsink}{\boldsymbol{\nabla}}

\newcommand{\Bzero}{\boldsymbol{0}}
\newcommand{\Bone}{\boldsymbol{1}}
\newcommand{\Bfout}{\boldsymbol{f}^{\mathrm{out}}}
\newcommand{\Bfin}{\boldsymbol{f}^{\mathrm{in}}}
\newcommand{\Btau}{\boldsymbol{\tau}}

\newcommand{\dist}{\mathrm{dist}}
\newcommand{\rand}{\mathrm{rand}}

\newcommand{\N}{\mathbb{N}}
\newcommand{\R}{\mathbb{R}}

\newcommand{\cC}{\mathcal{C}}

\newcommand{\cH}{\mathcal{H}}
\newcommand{\cI}{\mathcal{I}}

\newcommand{\tO}{\tilde{O}}

\global\long\def\level{\mathsf{level}}

\global\long\def\KKOV{\mathrm{KKOV}}
\global\long\def\KRV{\mathrm{KRV}}
\global\long\def\unfold{\mathrm{unfold}}

\newcommand{\otil}{\tilde{O}}
\newcommand{\card}[1]{|#1|}
\newcommand{\efirst}{E_{\mathrm{first}}}

\newcommand{\elast}{E_{\mathrm{last}}}

\newcommand{\Ecut}{E_{\mathrm{cut}}}

\def\cameraready{0}

\makeatletter
\newcommand\footnoteref[1]{\protected@xdef\@thefnmark{\ref{#1}}\@footnotemark}
\makeatother

\title{Combinatorial Maximum Flow\\ via Weighted Push-Relabel on Shortcut Graphs}

\author{Aaron Bernstein\thanks{
        New York University,
        \texttt{bernstei@gmail.com}. Supported by Sloan Fellowship, Google Research Fellowship,  NSF Grant 1942010, and Charles S. Baylis endowment at NYU.
    } \and Joakim Blikstad\thanks{
        CWI Amsterdam,
        \texttt{joakim@cwi.nl}.
    } \and Jason Li\thanks{
        Carnegie Mellon University,
        \texttt{jmli@cs.cmu.edu}.
    }  \and Thatchaphol Saranurak\thanks{
        University of Michigan,
        \texttt{thsa@umich.edu}.
        Supported by NSF Grant CCF-2238138. Partially funded by the Ministry of Education and Science of Bulgaria's support for INSAIT, Sofia University ``St.~Kliment Ohridski'' as part of the Bulgarian National Roadmap for Research Infrastructure.
    }  \and Ta-Wei Tu\thanks{
        Stanford University,
        \texttt{taweitu@stanford.edu}.
        Supported by a Stanford School of Engineering Fellowship, a Microsoft Research Faculty Fellowship, and NSF CAREER Award CCF-1844855.
    } }
\date{}

\begin{document}

\begin{titlepage}
  \maketitle \pagenumbering{gobble}
  \begin{abstract}
We give a combinatorial algorithm for computing exact maximum flows in directed graphs with $n$ vertices and edge capacities from $\{1,\dots,U\}$ in $\tilde{O}(n^{2}\log U)$ time, which is near-optimal on dense graphs. This shaves an $n^{o(1)}$ factor from the recent result of {[}Bernstein--Blikstad--Saranurak--Tu FOCS'24{]} and, more importantly, greatly simplifies their algorithm. We believe that ours is by a significant margin the simplest of all algorithms that go beyond $\otil(m\sqrt{n})$ time in general graphs. To highlight this relative simplicity, we provide a full implementation of the algorithm in C++.

The only randomized component of our work is the cut-matching game. 
Via existing tools, we show how to derandomize it for \emph{vertex-capacitated} max flow and obtain a deterministic $\otil(n^2)$ time algorithm. This marks the first deterministic near-linear time algorithm for this problem (or even for the special case of bipartite matching) in any density regime.

\end{abstract}

  \setcounter{tocdepth}{3}
  \newpage
  \tableofcontents
  \newpage
\end{titlepage}
\newpage
\pagenumbering{arabic}

\newpage

\section{Introduction}

Fast algorithms for computing maximum flows have been pivotal in algorithmic research, inspiring various paradigms such as graph sparsification, dynamic data structures, and the application of continuous optimization to combinatorial problems. These algorithms find numerous applications in problems like bipartite matching, minimum cuts, and Gomory--Hu trees~\cite{gomory1961multi,LiP20,LiNPSY21,Cen0NPSQ21,CenH0P23,Abboud0PS23}. Below, we summarize the development of fast maximum flow algorithms since Dantzig's introduction of the problem \cite{dantzig1951application}. The input for this problem consists of a directed graph with $n$ vertices and $m$ edges. For simplicity, we assume in this introduction that edge capacities range from $\{1,\dots,\text{poly}(n)\}$.

\paragraph{Augmenting Paths: The Impact of Clarity.}

Ford and Fulkerson \cite{ford1956maximal} introduced the \emph{augmenting path} framework, where algorithms iteratively identify an augmenting path in the residual graph and augment the flow along this path by a value equal to the bottleneck of the augmenting path. Using this framework, algorithms with a time complexity of $O(m\cdot\min\{m^{1/2},n^{2/3}\})$ were discovered since the 70s \cite{karzanov1973finding,EvenT75,GoldbergR98}; this time bound remained the best record for the next 40 years.

Despite the lack of progress on time complexity, the clarity of this framework has enabled impactful developments that extend far beyond the maximum flow problem. For example, the blocking flow \cite{karzanov1973finding,dinic1970algorithm,GoldbergR98} has become an important object for parallel, distributed, and local algorithms \cite{Vishkin92,Cohen95,AwerbuchKR07,OrecchiaZ14,ChangS20,HaeuplerHS23}. The push-relabel method \cite{GoldbergT88} generalizes to solve numerous problems, including approximate low-degree spanning trees \cite{ChaudhuriRRT09}, global minimum edge/vertex cuts \cite{HaoO94,HenzingerRG00}, parametric flow \cite{GalloGT89,GranotMQT12}, submodular function minimization \cite{Iwata03,FleischerI03}, submodular flow \cite{frank2012simple}, and packing/covering matroid \cite{Quanrud24}. Push-relabel algorithms have also been adapted to local, dynamic, and parallel settings \cite{HenzingerRW17,SaranurakW19,ChenMGS25}.

\paragraph{Interior Point Methods and Dynamic Graph Data Structures.}

In the 2000s, Daitch and Spielman \cite{DaitchS08} introduced a novel framework for computing maximum flow via \emph{interior point methods} (IPMs) by reducing the problem to solving $\tilde{O}(\sqrt{m})$ instances of electrical flow, which can be solved in near-linear time \cite{SpielmanT04}.\footnote{In this paper, we use $\tilde{O}(\cdot)$ to hide poly-logarithmic factors in $n$. As standard in modern graph algorithms, we use \emph{near-linear} time to refer to algorithms with a running time of $\tilde{O}(m)$, and \emph{almost-linear} time for those running in $m^{1+o(1)}$ time.}
This sparked further advances in IPMs \cite{Madry13,LeeS14,Madry16,LiuS20Energy,KathuriaLS20} that minimize the number of iterations required to solve electrical flow or related problems. Consequently, an $\tilde{O}(m\sqrt{n})$ time maximum flow algorithm \cite{LeeS14} and, for unit-capacity graphs, an $m^{4/3+o(1)}$ time algorithm \cite{KathuriaLS20} were developed, breaking the previous time record set by the augmenting-path approaches.

Since 2020, the emphasis has shifted from minimizing the number of iterations in IPMs to reducing the cost per iteration through \emph{dynamic graph data structures}. A series of impressive works \cite{BrandLNPSS0W20,BrandLLSS0W21,GaoLP21,BrandGJLLPS22} culminated in almost-optimal $m^{1+o(1)}$ time algorithms for maximum flow and min-cost flow \cite{ChenKLPGS22,Brand0PKLGSS23}, along with their extensions to dynamic settings \cite{ChenKLMG24,vdBCKLMPS24}.

\paragraph{In Pursuit of Clarity.}

These breakthrough algorithms with almost-optimal time bounds are, however, still not simple in two different aspects. Conceptually, the IPM framework updates the flow solutions without a clear combinatorial interpretation. And, technically, the required dynamic data structures are still highly involved.\footnote{There has been exciting progress in simplicity among IPM-based algorithms. The involved  dynamic min-ratio cycle data structure in \cite{ChenKLPGS22,Brand0PKLGSS23,ChenKLMG24} was recently replaced by a simpler dynamic min-ratio cut data structure \cite{vdBCKLMPS24} based on dynamic tree cut sparsifiers~\cite{GoranciRST21}.} The current state of the art motivated a new search for algorithms that are more ``combinatorial'', as a proxy for clarity.

With this motivation, Chuzhoy and Khanna \cite{ChuzhoyK24SODA,ChuzhoyK24STOC} presented an algorithm for maximum bipartite matching, a special case of maximum flow, that runs in $n^{2+o(1)}$ time without using IPM. Their algorithm is conceptually clearer; it repeatedly increases the flow value along paths listed by dynamic shortest-path data structures. These data structures, however, are still complex.

Recently, Bernstein, Blikstad, Saranurak, and Tu \cite{BernsteinBST24} introduced an augmenting-path-based maximum flow algorithm with a running time of $n^{2+o(1)}$, improving upon the $O(m\cdot\min\{m^{1/2},n^{2/3}\})$ bound within the augmenting path framework. They bypassed heavy dynamic data structures and replaced them with a simple variant of push-relabel called \emph{weighted push-relabel}. The weighted push-relabel algorithm itself is clear, combinatorial, and implementable. The authors of \cite{BernsteinBST24} showed that given edge weights derived from an \emph{expander hierarchy,} weighted push-relabel is guaranteed to efficiently return a flow with good quality. Unfortunately, their algorithm for constructing the expander hierarchy is very complex.
Furthermore, their analysis of why the returned flow has good quality, given the expander hierarchy, is also quite intricate.
Thus, the goal of obtaining much simpler and clearer algorithms remains highly desirable.

\paragraph{Our Result.}
We propose a new combinatorial maximum flow algorithm that is significantly simpler and also strictly speeds up all recent combinatorial algorithms \cite{ChuzhoyK24SODA,ChuzhoyK24STOC,BernsteinBST24}, improving the running time from $n^{2+o(1)}$ to $O(n^{2} \cdot \polylog(n))$.\footnote{While we work out an explicit number of $O(\log n)$ factors in the $\tilde{O}(\cdot)$ notation, we made no significant effort in minimizing it. There are places (e.g., the expander decomposition algorithm) where more efficient algorithms should exist following standard techniques, yet we chose slower ones that we believe are simpler for presentation, analysis, and implementation.}

\begin{restatable}{thm}{MaxFlow}

\label{thm:main}
There is a randomized algorithm that, given an $n$-vertex directed graph $G$ with edge capacities from $\{1,\dots,U\}$ and $s,t\in V(G)$, finds a maximum $(s,t)$-flow in $G$ with high probability in $O(n^{2}\log^{19}n\log U)$ time.
\end{restatable}

It is worth noting though that while our $\tilde{O}(n^2)$ time algorithm improves on the almost-linear time algorithm of \cite{ChenKLPGS22} on dense graphs, there also exists an IPM-based $\tilde{O}(m+n^{1.5})$ time algorithm of \cite{BrandLLSS0W21} that already runs in near-linear time on dense graphs.

Overall, we believe that our algorithm is by a significant margin the simplest among all exact maximum flow algorithms for general graphs that go beyond the $\otil(m\sqrt{n})$ bound of \cite{LeeS14}, which include all the recent progress using IPM plus dynamic graph data structures.\footnote{We are only comparing ourselves to other algorithms for general graphs (directed, capacitated). There is a separate line of work for directed graphs \emph{with small capacities}, culminating in $m^{4/3+o(1)}$ time algorithms \cite{Madry16,LiuS20Energy,LiuS20Divergence,Kathuria20}; these works use the IPM framework but do not require dynamic data structures. Finally, for the special case of \emph{undirected, uncapacitated, simple} graphs, there is a much simpler $\otil(n^2)$ time algorithm of Karger and Levine  \cite{KargerL15}.}
Our result is also relatively self-contained. The main outside tool that we need is weighted push-relabel from \cite{BernsteinBST24}, which is a simple extension of classic push-relabel; a fully self-contained description and analysis is given in \cite[Section 4]{BernsteinBST24}. Beyond that we use:
\begin{itemize}
\item Standard graph algorithms such as topological sort and finding strongly connected components.
\item Classic tools for capacitated flow such as link-cut trees and fractional-to-integral flow rounding.
\item The heaviest machinery that we use is the cut-matching game for expanders. We specifically use the ``non-stop" cut-matching game. Loosely speaking, whereas the standard cut-matching game either certifies that the graph is an expander or finds a sparse cut, the non-stop version either certifies almost-expansion or returns a \emph{balanced} sparse cut.  This has already become a relatively standard tool for undirected expanders \cite{RackeST14,SaranurakW19,LiRW25}. The algorithm for directed expanders is similar but needs a more refined analysis; we use as a black box the result of \cite{FleischmannLL25}.
\end{itemize}

\paragraph{Implementation.}

As evidence of its relative simplicity, we present the full implementation of \cref{thm:main} in C++ \cite{BBLST25implementation}.
To our knowledge, this is the first complete implementation of any maximum flow algorithm that is provably faster than $O(m\sqrt{n})$ time in any density regime. 
Previously, \cite{Kavi24} gave a partial implementation of the almost-linear time algorithm of \cite{ChenKLPGS22}, but  significant parts are missing.
The lack of implementation of the recent theoretically efficient maximum flow algorithms is likely due to the sheer complexity of the algorithm descriptions that results from their reliance on dynamic graph data structures. 

We emphasize that our implementation is not at all practical, particularly because of the $\log^{19} n$ factor in our running time. Our goal is merely to show that the entire algorithm can be captured from start to finish using code of reasonable length.

\paragraph{Derandomization.}
The only randomized component of \Cref{thm:main} is the KRV-style cut-matching game from \cite{FleischmannLL25} for computing expander decomposition. There exists a deterministic variant of the cut-matching game \cite{ChuzhoyGLNPS20,BernsteinGS20}, but it incurs an $n^{o(1)}$ overhead in the running time and the quality of expansion, while also being more complicated. Using \cite{ChuzhoyGLNPS20,BernsteinGS20}, it should be straightforward to make our algorithm deterministic, but with running time of $n^{2+o(1)}\log U$ instead of $\otil(n^2 \log U$).

In fact, for the special case of \emph{vertex-capacitated} graphs (instead of edge-capacitated), our algorithm naturally gives a \emph{deterministic} $\tilde{O}(n^2)$ time algorithm when combined with the appropriate implementation of the cut-matching game.
This is the first near-linear time deterministic algorithm for the problem
on dense graphs.
Even for the special case of bipartite matching, no near-linear time \emph{deterministic} algorithms were previously known in any density regime. In particular, the $\tilde{O}(m+n^{1.5})$ time algorithm by \cite{BrandLNPSS0W20,BrandLLSS0W21} is randomized. 

\begin{restatable}{thm}{VertexFlow}
  There is a deterministic algorithm that, given an $n$-vertex directed graph $G$ with vertex capacities from $\{1,\ldots,U\}$ and $s, t \in V(G)$, finds a maximum $(s,t)$-flow in $G$ in $O(n^2\log^{46} n\log U)$ time.
  \label{thm:vertex-flow}
\end{restatable}

\paragraph{Techniques.} 

The main technical difficulty faced by~\cite{BernsteinBST24} in their bottom-up expander hierarchy construction is enforcing the hierarchical structure in between levels. That is, the next level $i+1$ need to be a coarsening of the previous level $i$, in the sense that each cluster on level $i+1$ is the (disjoint) union of some clusters on level $i$. Unfortunately, directly building a coarsening at level $i+1$ (without changing levels $i$ and below) is infeasible, in that the multiplicative losses build up significantly over the levels of the hierarchy. Instead,~\cite{BernsteinBST24} need to re-cluster levels $i$ and below in order to build level $i+1$ with minimal losses. The final algorithm continuously and dynamically breaks down and repairs previous levels while building the next level, a process which is highly technical
and suffers an extra $n^{o(1)}$ factor.

Our starting point is a recent bottom-up \emph{weak} expander hierarchy construction for undirected graphs~\cite{LiRW25}.\footnote{A similar strategy was implemented by Bartal~\cite{bartal1996probabilistic} to construct a low-diameter decomposition hierarchy.} Their simple conceptual insight is that if one settles for weak expanders, then the back-and-forth rebuilding between levels is completely unnecessary. Instead, we can first build the next level $i+1$ without regard to the structure of the previous $i$ levels, so that the coarsening property is not ensured. To enforce coarsening, simply \emph{refine} each previous level with respect to level $i+1$. That is, each cluster on a previous level is further partitioned into all possible non-empty intersections between that cluster and some cluster on level $i+1$. It turns out that unlike (strong) expander decomposition, weak expander decompositions are closed under refinements: if each cluster in a weak expander decomposition is partitioned arbitrarily, then the new clustering is still a weak expander decomposition.

However, the original algorithm of~\cite{BernsteinBST24} commits to strong expander hierarchy for a good reason: their weighted push-relabel framework breaks down in subtle ways for weak expander hierarchy. Surprisingly, we show that adding \emph{shortcuts} to expanders, as similarly done on undirected graphs~\cite{HaeuplerH0RS24,LiRW25}, resolves this issue. More precisely, we run the weighted push-relabel flow algorithm of~\cite{BernsteinBST24}, where the edge weights are induced by an expander hierarchy, on the original input graph \emph{union with shortcuts}, where the shortcut edges are stars on top of each expander in the hierarchy. Thus, our algorithm follows the augmenting-path framework, except that the paths can use the shortcut edges.

This idea leads to major simplification in both the algorithm and analysis of \cite{BernsteinBST24}, as will be explained in detail in \Cref{sec:overview}. At a high level, the shortcut enables reasoning about the flow of edges at different levels of the hierarchy independently, whereas \cite{BernsteinBST24} needed an intricate analysis to handle the interactions of edges across levels.

\section{Preliminaries}\label{sec:prelim}

Let $\N\defeq\{1,2\ldots\}$ be the set of positive integers and $1/\N \defeq \{x: 1/x \in \N\}$.
We use boldface $\Ba$ to denote a vector, where $\Ba \leq \Bb$ means $\Ba$ is entry-wise upper-bounded by $\Bb$.
For a vector $\Bw$ we may write $\Bw(S)\defeq\sum_{e\in S}\Bw(e)$, and $\Bw|_S$ denote the restriction of $\Bw$ to $S$.
Let $\Bone_S$ denote the vector that has one on entries in $S$ and zero elsewhere.

\subsection{Basic Definitions}

\paragraph{Graphs.}
Without stated otherwise, in this paper we consider directed edge-capacitated graphs of the form $G = (V, E, \Bc_G)$, where $\Bc_G\in\R_{\geq 0}^{E}$.
When the context is clear, we may drop $\Bc_G$ from the definition of the capacitated graph and implicitly use the subscript $G$ to indicate its underlying graph.
By default, we will assume that the capacities are positive integers.
Via standard capacity scaling techniques (e.g., \cite[Appendix B]{BernsteinBST24}), we will also assume that the capacities are bounded by $n^2$, i.e., $\|\Bc_G\|_{\infty} \leq n^2$, by incurring an $O(\log U)$ overhead in the final running time.

For a directed edge $e = (u, v)$, let $\tail(e) \defeq u$ and $\head(e) \defeq v$.
Let $E_G(A,B)$ for disjoint $A,B\subseteq V$ denote the set of edges whose tails are in $A$ and heads in $B$.

\paragraph{Incident Edges, Degree, Volume, and Cuts.}
Consider an $F \subseteq E$.
Let $\delta_F(v)$ be the set of edges in $F$ incident to $v$, whose sum of capacities we denote by $\deg_F(v) \defeq \sum_{e \in \delta_F(v)}\Bc_G(e)$.
Let $\vol_F(S)$ for $S \subseteq V$ be $\vol_F(S) \defeq \sum_{v \in S}\deg_F(v)$.
Let $\deg_G(v) \defeq \deg_E(v)$ and $\vol_G(S) \defeq \vol_E(S)$.
We may use $\vol(v)$ and $\deg(v)$ interchangeably for a vertex $v$, and we may also use them as vectors on $V$, in which case we will use the notation $\Bdeg, \Bvol$ (with appropriate subscripts).
For $S \subseteq V$, let $\overline{S} \defeq V \setminus S$.

\paragraph{Flows.}
A \emph{(single-commodity) flow} $\Bf\in\mathbb{R}_{\ge0}^{E}$ assigns non-negative value to each edge.
The \emph{congestion} of a flow $\Bf$ in $G$ is $\cong_{G}(\Bf) \defeq \max_{e \in E}\Bf(e)/\Bc_G(e)$.
Let $\Bf^{\mathrm{out}}(u) = \sum_{(u,v)\in E} \Bf((u,v)) - \sum_{(v,u)\in E} \Bf((v,u))$ denote the net outgoing flow at vertex $u\in V$.

A \emph{(single-commodity) vertex-demand} is a pair $(\Bsource, \Bsink)$ with \emph{source} and \emph{sink} vector $\Bsource,\Bsink \in \R_{\geq 0}^{V}$.
A \emph{flow instance} is a tuple $\cI = (G, \Bsource, \Bsink)$, and we often consider \emph{diffusion instances} where $\|\Bsource\|_1 \leq \|\Bsink\|_1$.
Let $\Bsource_{\Bf}(u) \defeq \max\{0,\Bsource(u) - \Bsink(u) - \Bf^{\mathrm{out}}(u)\}$ and $\Bsink_{\Bf}(u) \defeq \max\{0, \Bsink(u) - \Bsource(u) + \Bf^{\mathrm{out}}(u)\}$ be the \emph{residual source} and \emph{sink} demands at $u$.
A flow $\Bf$ \emph{partially routes} $(\Bsource,\Bsink)$ if $\Bsource_{\Bf} \leq \Bsource$ and $\Bsink_{\Bf} \leq \Bsink$, and the \emph{value} of such an $\Bf$ is $|\Bf| \defeq \sum_{u\in V}(\Bsource(u) - \Bsource_{\Bf}(u))$, i.e., the amount of source (equivalently, sink) demand routed.
The flow $\Bf$ \emph{(completely) routes} $(\Bsource, \Bsink)$ if $|\Bf| = \|\Bsource\|_1$ (equivalently $\Bsource_{\Bf} = \Bzero$) and $\Bsink_{\Bf} \geq \Bzero$.
The demand $(\Bsource,\Bsink)$ is \emph{routable} (with congestion $\kappa$) if there exists a flow (with congestion $\kappa$) routing it.
The flow $\Bf$ is \emph{feasible} for $\cI$ if
$\Bsource_{\Bf} \le \Bsource$,
$\Bsink_{\Bf} \le \Bsink$, and
$\cong_{G}(\Bf)\le 1$.

Given a flow $\Bf$, the \emph{residual graph} $G_{\Bf}$ contains for each $e = (u, v) \in E$ a \emph{forward} edge $\overrightarrow{e} =(u,v)$ with capacity $\Bc_{\Bf}(\forward{e}) \defeq \Bc_{G}(e) - \Bf(e)$ if $\Bc_{\Bf}(\forward{e}) > 0$ and a \emph{backward} edge $\rev{e} = (v, u)$ with capacity $\Bc_{\Bf}(\backward{e}) \defeq \Bf(e)$ if $\Bc_{\Bf}(\backward{e}) > 0$.
An \emph{augmenting path} is a path in $G_{\Bf}$ consisting of edges with positive residual capacities from an unsaturated source $s$ (i.e., $\Bsource_{\Bf}(s)>0$) to an unsaturated sink $t$ (i.e., $\Bsink_{\Bf}(t)>0$).

A flow $\Bf$ is \emph{$1/z$-integral} where $z \in \N$ if each $\Bf(e)$ is an integral multiple of $z$, i.e., $\Bf \in (1/z) \cdot \N^E$.
Similarly, a demand $(\Bsource,\Bsink)$ is \emph{$1/z$-integral} if $\Bsource,\Bsink \in (1/z)\cdot \N^E$.

\begin{lemma}[Flow rounding~\cite{KangP15}]
  Given an integral flow $\Bf$ in an $m$-edge graph $G$ routing an integral demand $(\Bsource, \Bsink)$ with $\Bsource,\Bsink \in z\cdot \N^V$ where $z \in \N$, there is a deterministic $O(m\log m)$ time algorithm that computes an integral $\Bf^\prime$ with $\Bf^\prime \leq \lceil \Bf/z\rceil$ routing $(\Bsource/z, \Bsink/z)$.
  \label{lemma:rounding}
\end{lemma}

\subsection{Weighted Push-Relabel}

We will use the \emph{weighted push-relabel} algorithm from \cite{BernsteinBST24} as a black box (see \cite[Section 4]{BernsteinBST24}). Recall that the classic push-relabel algorithm up to height $h$ finds a flow $f$ such that the residual graph has no augmenting paths of length $h$; the running time is $\tilde{O}(mh)$ because an edge is processed every time it moves up a layer. The weighted push-relabel algorithm takes as additional input a weight function $\Bw$ and finds a flow such that every augmenting path $P$ has $\Bw(P) \leq 3h$. The key advantage of using weights is that an edge of weight $\Bw(e)$ only has to be processed every $\Bw(e)$ layers, so each edge contributes only $h/\Bw(e)$ to the running time.

  \begin{theorem}[Weighted Push-Relabel \cite{BernsteinBST24}] \label{thm:push-relabel-main-theorem}
  Suppose we have a maximum flow instance $\cI = (G,\Bc,\Bsource,\Bsink)$ consisting of an $n$-vertex $m$-edge directed graph $G = (V, E)$, integral edge capacities $\Bc\in \N^E$, and integral demand $(\Bsource,\Bsink)$. Additionally, suppose we have a weight function $\Bw \in \N^E$ and height parameter $h \in \N$. Then there is an algorithm \emph{\alg{WeightedPushRelabel}{$G, \Bc, \Bsource, \Bsink, \Bw, h$}} that in $O\left(\left(m + n + \sum_{e \in E}\frac{h}{\Bw(e)}\right)\log n \right)$ time finds a feasible integral flow $\Bf$ such that
  \begin{enumerate}[(i)]
    \item\label{item:push-relabel:invariant} the $\Bw$-distance in the residual graph $G_{\Bf}$ between any unsaturated source $s$ $(\Bsource_{\Bf}(s) > 0)$ and any unsaturated sink $t$ $(\Bsink_{\Bf}(t) > 0)$ is at least $\dist^{\Bw}_{G_{\Bf}}(s,t) > 3h$, and
    \item\label{item:push-relabel:short-flow} the average $\Bw$-length of the flow is $\frac{\Bw(\Bf)}{|\Bf|} \le 9h$.
    \item\label{item:push-relabel:approximation} $\Bf$ is a $\frac{1}{6}$-approximation of $\Bf^{*}_{\Bw,h}$---the optimal (not necessarily integral) flow with average $\Bw$-length $\frac{\Bw(\Bf^{\star}_{\Bw,h})}{|\Bf^{\star}_{\Bw,h}|}\le h$. %
  \end{enumerate}
\end{theorem}

\subsection{Crucial Concepts: Hierarchies, Expanders, and Shortcuts}

\paragraph{Levels and Hierarchy.}
Let $\level:E\rightarrow\{1,\dots,L\}$ be a \emph{level function} of $G$. Define $E_{i}\defeq\{e\mid\level(e)=i\}$ as a set of \emph{level-$i$ edges}. Let $E_{\le i},E_{\ge i},E_{<i}$ and $E_{>i}$ be defined in a natural way.

\begin{defn}
\label{dfn:hierarchy}
Given a level function $\level$ of $G$, a \emph{level-$i$ component} of $G$ is a strongly connected component in $G\setminus E_{>i}$. Let ${\cal C}_{i}$ contain all level-$i$ components of $G$. A \emph{hierarchy} of $G$ induced by $\level$ is ${\cal H}=\{{\cal C}_{0},\dots,{\cal C}_{L}\}$.
\end{defn}

Note that the vertex sets of components in the hierarchy ${\cal C}$ over all levels form a laminar family where ${\cal C}_{0}=\{\{v\}\}_{v\in V}$ consists of singletons and ${\cal C}_L$ consists of SCCs of $G$. When we refer to a hierarchy ${\cal H}$, we implicitly assume that the level function inducing ${\cal H}$ is given to us too. Next, we define a vertex ordering that respects a hierarchy.

\begin{defn}
\label{def:tau}
A vertex ordering $\Btau_{{\cal H}} \in [n]^V$ \emph{respects} a hierarchy ${\cal H}$ of $G$ if the following holds.
\begin{enumerate}
\item For every component $C\in{\cal H}$, the image $\Btau(C)$ of $C$ is contiguous in $[n]$.
\item If $u$ are $v$ are in different level-$i$ components but $u$ can reach $v$ in $G\setminus E_{>i}$, then $\Btau_{{\cal H}}(u)<\Btau_{{\cal H}}(v)$.%
\end{enumerate}
Note that such an ordering can be easily computed in $O(mL)$ time (see \cite[Lemma 5.3]{BernsteinBST24}).
\end{defn}

\paragraph{(Component-Constrained) Expansion.}

Rather than the usual definition of expanders in terms of cuts, we define them in terms of flow. Intuitively, a set of terminal edges $F \subseteq E$ is expanding if any unit demand on $F$ can be routed in $G$ with low congestion.

Formally, we say that a demand $(\Bsource,\Bsink)$ with $\Bsource(V) \leq \Bsink(V)$ \emph{respects} $\Bd\in\mathbb{R}_{\ge0}^{V}$ if $\Bsource(v),\Bsink(v)\le\Bd(v)$ for all vertices $v\in V$. Let ${\cal C}$ be a collection of disjoint vertex sets. We say that $(\Bsource,\Bsink)$ is \emph{${\cal {\cal C}}$-component-constrained} if $\Bsource(S)\le \Bsink(S)$ for all sets $S\in{\cal {\cal C}}$.\footnote{Ideally we want to route sources in $S$ to sinks in $S$, and a demand being $\mathcal{C}$-component-constrained means that this is not vacuously made impossible by simply not having enough sink capacity in $S$.}
\begin{defn}
\label{dfn:expander}
An edge set $F$ is \emph{$(\Bw, h)$-length $\phi$-expanding} in $G$ where $\Bw \in \R^{E(G)}_{\geq 0}$ are edge weights if every $\Bvol_{F}$-respecting demand is routable in $G$ via a $\Bf$ with congestion at most $1/\phi$ and average length $\frac{\Bw(\Bf)}{|\Bf|} \leq h$.\footnote{Equivalently, $\Bf$ can be decomposed into a collection of flow paths such that the (weighted) average length of these flow paths is bounded by $h$. Note that this is not the same as the \emph{length-constrained expansion} recently proposed in the literature (e.g.,~\cite{HaeuplerR022,HaeuplerHS23,HaeuplerH0RS24}).
For comparison, this notion if the average-length version of the \emph{router} defined in \cite{HaeuplerH0RS24}, which is a simpler-to-define object than \emph{length-constrained expanders} in \cite{HaeuplerH0RS24} and literature.
}
The edge set is \emph{$h$-hop $\phi$-expanding} if it is \emph{$(\Bone, h)$-length $\phi$-expanding}, and simply \emph{$\phi$-expanding} if the hop/length constraint is dropped.
An edge set $F$ is \emph{${\cal {\cal C}}$-component-constrained} \emph{$\phi$-expanding} if the same holds for $\Bvol_{F}$-respecting ${\cal {\cal C}}$-component-constrained demand.
\end{defn}

Intuitively, if $F$ is ${\cal C}$-component-constrained expanding, then any $\Bvol_{F}$-respecting demand between vertices inside the same component $S$ of ${\cal C}$ can be routed with low congestion. Note that the flow routing the demand inside each component $S$ can go out of $S$, but the total congestion for \emph{simultaneously} routing the demand of all components $S\in{\cal C}$ must be low.

\paragraph{Weak Expander Hierarchies.}

We use a weaker version of the expander hierarchy defined in \cite{BernsteinBST24}. We start with a brief intuition. If the entire edge set $E$ is $\phi$-expanding in $G$ then the hierarchy only has one level. If not, we want a smaller set of edges $E_2$ such that all the components of $G \setminus E_1$ are (weakly) $\phi$-expanding in $G$. We then want an even smaller set $E_3$ such that inside each component of $G \setminus E_3$, the edges of $E_2$ are $\phi$-expanding.
This is formalized as follows.

\begin{defn}
[Weak Expander Hierarchy]
\label{dfn:expander-hierarchy}
A hierarchy ${\cal H}$ of $G$ is \emph{$\phi$-expanding} if, for every level $i\in\{1,\dots,L\}$, the level-$i$ edge set $E_{i}$ is $[G\setminus E_{>i}]$-component-constrained $\phi$-expanding in $G$, where $[G\setminus E_{>i}]$ refers to the set of SCCs in $G \setminus E_{>i}$.
\label{def:hierarchy}
\end{defn}

In contrast to \cref{def:hierarchy}, the strong hierarchy of \cite{BernsteinBST24} instead required that each level-$i$ component $C$ is itself an expander.
That is, $E_i\cap C$ is $\phi$-expanding in $C$.\footnote{The hierarchy of \cite{BernsteinBST24} also included a set of ``DAG edges'' $D$. We do not need this terminology, but loosely speaking, the DAG edges in our hierarchy correspond to the edges in $E_i$ that cross the SCCs of $G\setminus E_{>i}$.}

\paragraph{Shortcut of Hierarchy.}

The key difference between our algorithm and that of \cite{BernsteinBST24} is the use of shortcuts. These are new edges that we add to the graph to speed up flow computations; see \cref{overview:sec:shortcuts} for more intuition.

\begin{defn}(Star Edges)
\label{dfn:star}
For any edge set $F\subseteq E$, the \emph{star $A_{F}$ on $F$} is a star whose root is a new Steiner vertex $r$ and leaves are the tail of all edges in $F$.\footnote{The choice of tail or head of edges is arbitrary.} The edges between the root and leaves are bidirectional. We say that $A_{F}$ has \emph{capacity scale} $\psi \in 1/\mathbb{N}$ if the edges between the root $r$ and the endpoint of $e\in F$ have capacity $\psi\cdot\Bc(e)$. If there are multiple edges in $F$ that have the same tail we avoid parallel edges in $A_F$ by simply combining those into single edges by summing up the capacities.
\end{defn}

\begin{defn}
[Shortcut of Hierarchy]
\label{def:shortcut_of_h}
Given a hierarchy ${\cal H}$ of $G$, the \emph{shortcut} $A$ induced by ${\cal H}$ is a collection of stars defined as follows. For each level-$i$ component $C\in{\cal C}_{i}$, let $A_{C}$ be the star on $E_{i}\cap C$. Let $A_{i}\defeq\bigcup_{C\in{\cal C}_{i}}A_{C}$ denote a \emph{level-$i$ shortcut} and $A\defeq\bigcup_{1\le i\le L}A_{i}$. We say that $A$ has capacity scale $\psi$ if every star $A_{C}$ has capacity scale $\psi$. 
A \emph{shortcut graph} $G_{A}\defeq G\cup A$ is obtained by augmenting the shortcut $A$ to $G$.
\end{defn}

\begin{defn}\label{def:weight-of-shortcut}
Given a shortcut~$A$ induced by a hierarchy ${\cal H}$ of some subgraph $G^\prime\subseteq G$ (with $V(G^\prime) = V(G)$), define the \emph{weight function} $\Bw$ of $G_{A}$ as follows. 
\begin{itemize}
\item For each original edge $e=(u,v)\in E$, set $\Bw_{\cH}(e)\defeq|\Btau_{{\cal H}}(u)-\Btau_{{\cal H}}(v)|$ where $\Btau_{{\cal H}}$ is a vertex ordering induced by ${\cal H}$ (recall \cref{def:tau}).
\item For each shortcut edge $e\in A_{C}$ in component $C$, set $\Bw_{\cH}(e)\defeq|V(C)|$.
\end{itemize}
\end{defn}

Note that to define the weight function $\Bw_{\cH}$, the hierarchy $\cH$ does not need to include all edges in $G$.
In this case, the weight of each $e \in G\setminus G^\prime$ is still well-defined since $\cH$ still induces an ordering $\Btau_{\cH}$ of $V(G)$.
Also observe that any level-$i$ edge $e = (u, v)$ that moves \emph{backward} in the vertex order $\Btau_{\cH}$ (i.e., $\Btau_{\cH}(u) > \Btau_{\cH}(v)$ is contained in some level-$i$ component.

\begin{observation}
  For any $(u,v) \in E_i$ with $\Btau_{\cH}(u) > \Btau_{\cH}(v)$ there exists a $C \in \cC_i$ such that $u, v \in C$.
  \label{obs:backward-has-star} 
\end{observation}

\section{Technical Overview}\label{sec:overview}

For simplicity of presentation, we assume during this overview that the input graph is unit-capacitated, in which case the congestion of $\Bf$ is simply $\cong(\Bf)\defeq\max_{e \in E}\Bf(e)$.
Note that it suffices to design a constant-approximate flow algorithm for directed graphs, as the exact algorithm then follows by repeating the approximate algorithm $O(\log n)$ times on the residual graph.

\subsection{High-Level Framework Borrowed from \texorpdfstring{\cite{BernsteinBST24}}{[BBST24]}} 
\label{overview:sec:previous}

We first describe the high-level framework of \cite{BernsteinBST24} that our algorithm also follows.

\paragraph{Weight Function.} Using \cref{thm:push-relabel-main-theorem} (weighted push-relabel from \cite{BernsteinBST24}), computing an approximate flow in $\otil(n^2)$ time reduces to computing a ``good" weight function $\Bw$ that satisfies: (1) there exists an approximate max flow $\Bf$ for which all flow paths $P$ have $\Bw(P) = \otil(n)$ (this allows us to set $h = \otil(n)$ in \cref{thm:push-relabel-main-theorem}) and (2) $\sum_e 1/\Bw(e) = \otil(n)$. 

Note that if $G$ is a DAG with topological order $\Btau$, then $\Bw(u,v) \defeq |\Btau(u) - \Btau(v)|$ is a good weight function: \emph{all} paths have $\Bw(P) \leq n$ and $\sum_e 1/\Bw(e) = \otil(n)$ because the weights incident to any vertex $v$ form a harmonic sum. Applying \cref{thm:push-relabel-main-theorem} yields a remarkably simple approximate max-flow algorithm for DAGs, as noted by \cite{BernsteinBST24}.

For general graphs, we follow \cite{BernsteinBST24} in defining $\Bw$ based on an expander hierarchy $\cH$: let $\Btau$  be any vertex ordering that respects $\cH$ (see \cref{def:tau}) and set $\Bw(u,v) \defeq |\Btau(u) - \Btau(v)|$. Observe that we again have $\sum_e 1/\Bw(e) = \otil(n)$ for the same reason as in a DAG.

\paragraph{The Three Key Steps.} At a high level, our algorithm follows three basic steps of \cite{BernsteinBST24}.
\begin{enumerate}
\item Prove that the weigh function $\Bw$ defined by an expander hierarchy $\cH$ is a good weight function. 
\item Show how to efficiently construct an $\cH$ using an expander decomposition subroutine.
\item Show how to compute expander decomposition with respect to edge set $E_i$ using weighted push-relabel with weight function defined by the partial hierarchy $E_1,\dots,E_i$ computed so far. 
\end{enumerate}
Below, we describe in \Cref{overview:sec:shortcuts} how we introduce \emph{shortcuts} into the framework of \cite{BernsteinBST24} and then explain how they significantly simplify (and improve) the three steps above in \Cref{overview:sec:weak-weight-function,overview:sec:hierarchy,overview:sec:decomp}, respectively.

\subsection{Our Key Contribution: Shortcutting Expanders}
\label{overview:sec:shortcuts}

Much of the technical difficulty of \cite{BernsteinBST24} comes from the need to keep track of the interactions between levels of the hierarchy. 
We show that it is possible to avoid much of this interaction by adding a set of \emph{new} edges (and a new Steiner vertex) to the graph, denoted \emph{shortcuts}, which simplify the graph topology and speed up flow computations. 

The starting point is the following observation: as long as we only add shortcuts to edge sets that are already expanding in $G$, they will not significantly change the flow structure of the graph.

\begin{observation}
\label{overview:obs:shortcuts}
Let $F$ be a set of terminal edges that is $\phi$-expanding in $G$. Then, adding a star $A_F$ on $F$ with capacity scale $\psi =\phi/\polylog(n)$ (\cref{dfn:star}) has minimal impact on the flow structure of $G$: formally, for any feasible flow $\Bf'$ in $G \cup A_F$, there exists a flow $\Bf$ in $G$ with congestion $1+1/\polylog(n)$.
\end{observation}

\begin{proof}
In $\Bf'$ the total flow on edges in $A_F$ is at most $\psi \card{A_F} = 2\psi\card{F}$. If we instead try to route this flow inside $G$ properly, the resulting flow instance has at most $\psi \deg_F(v)$ supply/demand on any vertex $v$. Since $F$ is expanding in $G$, this flow can be routed with a congestion of $\otil(\psi/\phi) = 1/\polylog(n)$.
\end{proof}

\paragraph{Shortcutting an Expander Hierarchy.}  In line with the above observation, we will add shortcuts to the expanders at every level of our hierarchy $\cH$. Formally, letting $\cH$ be a $\phi$-expander hierarchy (\cref{dfn:expander-hierarchy}), for every level $i$ and every SCC $C$ in $G \setminus E_{>i}$, we add a star on $G[C] \cap E_i$ with capacity scale $\psi = \phi/\polylog(n)$. For the rest of the overview, we refer to $A$ as the total set of shortcuts and define $G_A \defeq G \cup A$. Note that the final flow we compute will be in $G \cup A$; we discuss in Section \ref{overview:sec:unfolding} how to transform this into a flow on $G$.

\subsection{Improvements and Simplifications from Using Shortcuts} The addition of shortcuts ends up significantly simplifying many seemingly unrelated parts of the algorithm of \cite{BernsteinBST24}, as well as improving the running time from $n^{2+o(1)}$ to $\otil(n^2)$. At a high level, shortcuts on $G[C] \cap E_i$ provide a trivial way to send a moderate amount of flow between those edges without needing to interact with other levels of the hierarchy. In this section, we outline how we use shortcuts to solve the three key steps described in \cref{overview:sec:previous}.

\subsubsection{Weak Expansion and Shortcuts}
\label{overview:sec:weak-weight-function}
The authors of \cite{BernsteinBST24} show that given a \emph{strong} $\phi$-expander hierarchy $\cH$ (without shortcuts), the corresponding weight function $\Bw_{\cH}(u,v) \defeq |\Btau(u) - \Btau(v)|$ is a good one---that is, there exists an approximate max flow $\Bf$ such that all flow paths $P$ have $\Bw(P) = \otil(n/\phi)$. Unfortunately, this claim is false if $\cH$ is a weak expander hierarchy.

We briefly outline the proof of \cite{BernsteinBST24} to see why a strong hierarchy is needed. They show that to bound $\Bw_{\cH}(P)$ for a path $P$ in a flow $\Bf$---and hence to prove that $\Bw_{\cH}$ is a good weight function---it suffices to prove the following \emph{short-flow property}: for any level $i$ component $C$ (\cref{dfn:hierarchy}), $\Bf$ uses at most $\otil(1/\phi)$ edges in $C \cap E_i$. They prove the short-flow property by exploiting the fact that since $C \cap E_i$ is expanding in $C$ (because $\cH$ is an expander hierarchy), the aggregate flow through $C$ can be rerouted to use few edges in $E_i$. This rerouting needs to be done carefully to avoid blowing up congestion on lower levels. More importantly, the argument crucially requires the fact that rerouting a level $i$ component does not add any flow on $E_{>i}$, since otherwise those higher levels might once again violate the short-flow property. This is true in a strong expander hierarchy because $C$ is an expander in $G \setminus E_{>i}$. In a weak expander hierarchy, where the rerouting might use any edge in $G$, the whole argument breaks down. %

We now sketch the proof that adding shortcuts leads to a good weight function even for a weak hierarchy. The main advantage of this, as we will see in the next section, is that constructing a weak hierarchy is much easier. Our weight function $\Bw_{\cH}(u,v) \defeq |\Btau(u) - \Btau(v)|$ is the same as in \cite{BernsteinBST24} for edges in $G$; for every level $i$-component $C$, we set the star shortcuts on $C$ to have weight $\card{C}$. 

The argument that $\Bw$ satisfies the short-flow property is quite straightforward: unlike in \cite{BernsteinBST24}, we can handle each flow path $P$ separately rather than carefully rerouting aggregate instances. Consider any flow path $P$ in the current flow and any level $i$ component $C$; for simplicity assume $P$ has one unit of flow. Let $e_1, ..., e_q$ be the edges of $P$ in $C \cap E_i$. If $q < 2/\psi = \otil(1/\phi)$ then $P$ already satisfies the short-flow property. Otherwise, let $\efirst \defeq \{e_1, ..., e_{1/\psi}\}$ and $\elast \defeq \{e_{q-1/\psi+1}, ..., e_q$\}. We ``leak" flow from $\efirst$ to $\elast$ through the star edges in the natural way: reduce the flow on all edges in $P[e_1,e_q]$ by $\psi$ and instead add $\psi$ flow on the star edge from $e_1$ to $r$ and $r$ to $e_q$; then reduce the flow on all edges in $P[e_2,e_{q-1}]$ by $\psi$ and add $\psi$ flow on the star edge from $e_2$ to $r$ and $r$ to $e_{q-1}$. Repeating this for $1/\psi$ steps, the only edges with non-zero flow on $P \cap C \cap E_i$ will be the $2/\psi = \otil(1/\phi)$ edges of $\efirst$ and $\elast$. It is not hard to check the above procedure obeys the capacities of star edges. Note that the flow of all other edges only decreases, so there is no risk of violating the short-flow property for other levels; this is why a weak expander hierarchy suffices once we add shortcuts.\footnote{In fact, observe that our argument does not use the expansion properties of $\cH$ at all. The only reason we need expansion is to guarantee that adding shortcuts does not significantly alter the flow structure of the graph (\cref{overview:obs:shortcuts})}

\subsubsection{Bottom-Up Construction of Weak Expander Hierarchy}
\label{overview:sec:hierarchy}
As mentioned in the three key steps of \cref{overview:sec:previous}, both our algorithm and that of \cite{BernsteinBST24} construct the hierarchy via the following expander decomposition subroutine: given a partial hierarchy $E_1, ..., E_i$, either certify that $E_i$ is expanding (in which case we have a complete hierarchy) or compute a set of edges $\Ecut$ such that 
$E_i\setminus \Ecut$ is expanding in $G\setminus \Ecut$.\footnote{The \emph{strong} expander decomposition of \cite{BernsteinBST24} requires $E_i\setminus \Ecut$ to be expanding in $G\setminus \Ecut$, while our \emph{weak} one, which is simpler to compute, only needs it to be expanding in $G$.} In this section, we assume black-box access to such a subroutine and show that computing a strong expander hierarchy is still extremely challenging, while computing a weak one is trivial. The main advantage of shortcuts is thus that they allow us settle for a weak hierarchy.

Say that we have already computed a partial hierarchy $E_1, ..., E_i$, where $E_i$ might not be expanding in $G$. It is tempting to run the expander decomposition subroutine, and simply set $E_{i+1} \gets \Ecut$. If we were in the lucky case where $E_{i+1} \subseteq E_i$, then this would indeed work (we would then set $E_i \leftarrow E_i \setminus \Ecut$). Unfortunately, one can construct examples where the expander decomposition \emph{must} pick $\Ecut$ that uses edges from $E_{<i}$, which causes serious complications.
Consider some level $j < i$ and some level-$j$ SCC $C$ inside which $E_j$ is expanding. Say that $\Ecut$ contains some edges in $C$. If we set $E_{i+1} \gets \Ecut$, the edges of $\Ecut$ are added to $E_{>j}$, so $C$ might split into SCCs $C_1, C_2$. But even though $E_j$ was expanding in $C$, there is no guarantee that is it expanding in $C_1$ or $C_2$: routing a flow-instance in $C_1$ might require using edges from $C_2$, and vice versa.

To summarize, when constructing a strong expander hierarchy in a bottom-up fashion, fixing expansion properties at level $i$ can actually ruin the properties at level $j < i$. As a result,  the algorithm of \cite{BernsteinBST24} could not proceed in a direct bottom-up fashion, and instead needed to repeatedly move up and down between levels to fix issues caused in the previous step. 
Ensuring a small number of total moves required significant algorithmic adjustment and was the most technically involved aspect of the entire paper (see \cite[Section 7]{BernsteinBST24}). It was also the reason behind the $n^{o(1)}$ factor in their running time: the total number of moves was exponential in the number of layers in the hierarchy, so they could only afford a hierarchy with $o(\log n)$ layers.

By contrast, if we are willing to settle for a weak expander hierarchy, then the natural bottom-up construction easily goes through: we simply set $E_{i+1} \gets \Ecut$ and remove all edges in $E_{\leq i} \cap \Ecut$ from their respective sets. No further changes to the lower levels are needed and we can proceed to finding the next set $E_{i+2}$. In particular, consider the component $C$ from the previous paragraph: although $E_j$ might not be expanding in $G[C_1]$ or $G[C_2]$ individually, we know by definition of expander hierarchy that any demand respecting $G[C] \cap E_j$ can be routed in the whole graph $G$, and hence the same goes for any demands respecting $G[C_1] \cap E_j$ and $G[C_2] \cap E_j$.

\subsubsection{Computing a Single Level of the Hierarchy}
\label{overview:sec:pruning}
\label{overview:sec:decomp}
We now briefly discuss our approach to the expander decomposition subroutine from the previous section. Our high-level approach mimics that of \cite{BernsteinBST24}. By standard cut-matching games, computing an expander decomposition boils down to repeatedly solving the following: we are given some supply/demand on $E_j$ and want to either find a large flow or a small cut. We do so by running weighted push-relabel (\cref{thm:push-relabel-main-theorem}) with the weight function $\Bw$ defined by the lower levels of the hierarchy $E_1, \ldots, E_{j-1}$. If the flow $\Bf$ computed by weighted push-relabel is large, then we are done.
If the flow $\Bf$ is small, we need to return a small cut.

As in regular push-relabel, we compute a small cut by exploiting the property that the residual graph with respect to flow $\Bf$ contains no augmenting paths of small $\Bw$-weight, so the vertices can be partitioned into many layers $L_k$. We then argue that one of these layers has few edges crossing it in the residual graph and hence yields a small cut. We analyze each hierarchy-level $i$ separately and argue that most layers $L_k$ have few crossing edges from $E_i$. Observe that by definition of expander hierarchy, all edges in $E_i$ are in some SCC $C$ of $G \setminus E_{>i}$. We prove two key properties:
\begin{enumerate}
\item The vertices of each $E_i \cap C$ have small diameter---that is, for any $(u,u'),(v,v') \in E_i \cap C$, $\dist_{\Bw}(u,v)$ is small---and hence occupy only a small number of layers $L_k$.%
\item\label{item:path-reversal} Observe that push-relabel augments down the computed flow $\Bf$ and so reverses a small number of paths going through $C$. We thus also need to prove that there exists a small ``pruned set" $P_C \subseteq E_i$ such that the remaining vertices in each $(E_i \cap C) \setminus P_C$ have a small diameter even after the path reversals.\footnote{Note that we only only need to prove the \emph{existence} of $P_C$.}
\end{enumerate}

Our proofs of the two properties above diverge significantly from those of \cite{BernsteinBST24}. In \cite{BernsteinBST24}, although it is easy to argue that there exists a $(u, v)$-path with few edges in $E_i$ (because $E_i$ is expanding in $C$), bounding the weight of lower-level edges requires a much more careful analysis. The path reversals of Property \labelcref{item:path-reversal} pose an even bigger challenge. Intuitively, since $C$ is an expander, one can prove the existence of $P_C$ by using expander pruning. But even though the number of reversed flow paths is small, these paths can have many edges, so standard expander pruning will not give sufficiently strong bounds. Thus \cite{BernsteinBST24} needed to develop an entirely new expander pruning subroutine that can handle path reversals and is also sensitive to distances under $\Bw_{\cH}$.

By contrast, the presence of shortcuts leads to a much simpler analysis. Recall that every component $C$ now comes with a star root $r_C$ and star edges incident to every edge $e \in E_i \cap C$.
The diameter of  $E_i \cap C$ is  thus trivially small because one can use the path $(u,r_C)\circ(r_C,v)$. To handle path reversals, we observe that any simple path will use only two star edges, so only a small number of star edges will be saturated by the flow paths.
As such, except for a small set $P_C$ all other edges in $E_i \cap C$ will still be incident to an (unsaturated) bidirectional star edge in the residual graph and thus still have small diameter. This proof sketch again highlights the way in which star edges allow us to largely avoid analyzing the interaction between different levels of the hierarchy.

\subsection{Unfolding the Shortcuts} 
\label{overview:sec:unfolding}
While greatly simplifying most of the algorithm, shortcuts do also add a new element of complexity, as the algorithm computes a flow $\Bf'$ in $G_A$ instead of the original graph $G$. By an extension of Observation \ref{overview:obs:shortcuts}, we can show that there must exist a flow $\Bf$ of almost the same value, but we still need to compute this flow. Fortunately, this only needs to be done at the very end of the algorithm, by which point we have computed a lot of useful information that allows us to compute $\Bf$ using straightforward techniques. Loosely speaking, given an SCC of $G \setminus E_{>i}$, the flow on the shortcut edges of the star on $G[C] \cap E_i$ defines a flow instance respecting $E_i$ in $C$. Since each component $C$ is a (weak) expander with respect to $E_i$, all of these flow instances can be routed in $G$ with low congestion. 
This routing is easy to obtain, again, via the weighted push-relabel algorithm.
This leverages the fact that we have already proven when computing the expander hierarchy that these demands can be routed not just with low congestion, but also with short paths.
Iterating this for all levels in the hierarchy top-down,  
we map the flow in $G_A$ back to a flow in $G$ with small congestion blow-up.
Finally, we run a standard flow rounding to convert it into the desired feasible approximate flow.

\section{Weighted Push-Relabel on Shortcut Graphs}
\label{sec:push_relabel}

\newcommand{\volfc}{\Bvol_{F}}

Below, in \cref{cor:approx max flow in shortcut graph}, we show that by running the weighted push-relabel algorithm of \cite[Algorithm~1]{BernsteinBST24} we can in $\tO(n^2)$ time find a constant-approximate maximum flow (and approximate minimum cut) on the shortcut graph $G_A$ with capacity scale $\frac{1}{\polylog(n)}$.
To this end, we
prove a more general version, \cref{lem:push-relabel-on-shortcut}, that will later prove useful in the construction of the expander hierarchy. In this version, we have an additional edge set $F$ not part of the hierarchy.

This should be compared to \cite[Theorem~6.1]{BernsteinBST24}.
However, note that our proof is significantly simpler due to the addition of shortcuts.
In fact, in contrast to \cite[Theorem~6.1]{BernsteinBST24} where the given hierarchy is required to be expanding, such a condition is not even needed for our \cref{lem:push-relabel-on-shortcut} to work.
Instead, \cref{lem:push-relabel-on-shortcut} gives us a flow in $G_A$ not $G$, and the expansion of the hierarchy is used only when we are ``unfolding'' the flow back to $G$, which happens after running the weighted push-relabel algorithm.

\begin{lem}
[Weighted Push-Relabel on Shortcut Graphs]\label{lem:push-relabel-on-shortcut}
Let $F\subseteq E$ and ${\cal H}$ be the hierarchy of $G\setminus F$ with $L=O(\log n)$ levels. Let $A$ be the shortcut induced by ${\cal H}$ with capacity scale $\psi \in 1/\N$, and $G_{A}\defeq G\cup A$ be the shortcut graph. Given a diffusion instance ${\cal I}=(\Bsource,\Bsink)$ with $\psi$-integral demand and parameter $\kappa\in \N$, there is a deterministic algorithm that runs on $G_{A}$ in $O(n^{2}\log^3 n \cdot (\kappa+1/\psi))$ time and finds a $\psi$-integral flow $\Bf$ in $G_{A}$ with congestion $\kappa$ and average weight $\frac{\Bw_{\cH}(\Bf)}{|\Bf|} < O(n\log n \cdot (\kappa+1/\psi))$.

Additionally, if $|\Bf| < \|\Bsource\|_{1}$, then the algorithm also returns a cut $\emptyset\neq S\subsetneq V(G_A)$ with $\Bsink_{\Bf}(S)=0$ and $\Bsource_{\Bf}(\overline{S})=0$ of size
\[
\Bc(E_{G_{A}}(S,\overline{S}))\le\frac{41|\Bf|+\min\{\volfc(S),\volfc(\bar{S})\}}{\kappa}.
\]
\end{lem}

Recall from \cref{def:weight-of-shortcut} that even though the hierarchy $\cH$ is only on $G\setminus F$, the weights $\Bw_{\cH}(e)$ for $e\in F$ are still well-defined from the vertex ordering induced by $\cH$.
Before proving the more general version \cref{lem:push-relabel-on-shortcut}, we show how it directly implies \cref{cor:approx max flow in shortcut graph} in the case when $F = \emptyset$ and $\kappa = 1$.

\begin{cor}
\label{cor:approx max flow in shortcut graph}Let ${\cal H}$ be a hierarchy of $G$. Let $A$ be the shortcut induced by ${\cal H}$ with capacity scale $\psi\in 1/\N$, and $G_{A}=G\cup A$ be the shortcut graph. There is a deterministic algorithm that, given $G_{A}$ and $s,t\in V(G)$, in $O(n^{2}\log^{3}n / \psi)$ time finds an $O(1)$-approximate maximum $(s,t)$-flow $\Bf$ and an $O(1)$-approximate minimum $(s,t)$-cut $S$ in $G_{A}$.
\end{cor}

\begin{proof}
This follows from \Cref{lem:push-relabel-on-shortcut} when $\kappa=1$, $F=\emptyset$, and ${\cal I}=(\Bsource,\Bsink)$ is a flow instance for routing a $(s,t)$-flow. More precisely, $\Bsource(s)=n^{3}U$ where $U$ is a maximum capacity of $G$ and $\Bsource(v)=0$ for all $v\neq s$. Similarly, $\Bsink(t)=n^{3}U$ and $\Bsink(v)=0$ for all $v\neq t$. Since this demand cannot be fully routed, \cref{lem:push-relabel-on-shortcut} must return a flow $\Bf$ with congestion $\kappa$ and a cut $S$, such that
$\Bc(E_{G_{A}}(S,\overline{S}))\le 41|\Bf|$.
Moreover, $s\in S$ and $t\not\in S$, since the residual demands satisfy $\Bsink_{\Bf}(S) = 0 =\Bsource_{\Bf}(\overline{S})$ and $\Bsource_{\Bf}(s), \Bsink_{\Bf}(t)>0$.
Hence $\Bf$ together with the cut $(S,\overline{S})$ must be $41$-approximate maximum flow respectively $41$-approximate minimum cut in $G_A$. The running time follows from \cref{lem:push-relabel-on-shortcut}.
\end{proof}

In the remainder of this section, we prove \cref{lem:push-relabel-on-shortcut}.

\paragraph{Weighted Push-Relabel of \cite{BernsteinBST24}.}
Our approach to prove \cref{lem:push-relabel-on-shortcut} is similar to the corresponding sparse-cut algorithm of \cite[Theorem~6.1, Algorithm~2]{BernsteinBST24}. Because we run on the shortcut graph $G_A$, instead of the original graph $G$ with an directed expander hierarchy as in \cite{BernsteinBST24}, our proofs here will be significantly simpler,\footnote{This section replaces most of the analysis in \cite[Section~5 and~6]{BernsteinBST24}. We still rely on the rather short weighted push-relabel algorithm \cite[Section~4]{BernsteinBST24}.} and have less overhead in the running time.

\paragraph{The Algorithm.}
Let $M \defeq \sum_{e\in F} \Bc(e)$ be the total capacity in $F$.
We set
\begin{equation}
h \defeq \left\lceil n\cdot \left(\frac{6L}{\psi}+100\kappa\log M\right) \right\rceil
\label{eq:h}
\end{equation}
as the target height, and scale up all the capacities (but not the demands) by a factor of $\kappa$.
We then run the weighted push-relabel algorithm (\cref{thm:push-relabel-main-theorem}) with weights $\Bw_{\cH}$ on $G_A$ to compute the flow $\Bf$, which by \cref{thm:push-relabel-main-theorem}\labelcref{item:push-relabel:short-flow} satisfies $\frac{\Bw_{\cH}(\Bf)}{|\Bf|} \leq 9h$ as required.\footnote{Since \cref{thm:push-relabel-main-theorem} requires the capacities to be integer, and the short-cut edges have capacity scale $\psi$, we scale up all capacities and demands before calling the weighted push-relabel by $\frac{1}{\psi}\in \N$, and then scale down the returned flow. The scaled down flow will have congestion $\kappa$ with respect to the original capacities $\Bc$, and might no longer be integral.}
In case we did not route all of the demand we also need to output a cut. Similar to \cite{BernsteinBST24}, we modify the edge lengths slightly to $\Bw_{\Bf}$ by setting many of the edge weights in the residual graph to 0.
In particular, 
we set $\Bw_{\Bf}(\forward{e}) = 0$ for
any edge $e = (u,v)\in E\setminus F$ (notably this excludes edges in $F$, short-cut edges in $A$, and reversed edges formed in the residual graph) that moves forward in the vertex order associated with $\cH$ (i.e., $\Btau_{\cH}(u) < \Btau_{\cH}(v)$)
. In the residual graph (ignoring edges with zero residual capacity), we calculate the $\Bw_{\Bf}$-distance for each vertex from its closest unsaturated source, and let $S_{i}$ be the corresponding distance layers. The main part of the analysis is to show that one of these distance-layer cuts $(S_{\le i}, \overline{S_{\le i}})$ must satisfy the output-guarantee of \cref{lem:push-relabel-on-shortcut}.

\begin{algorithm}[!ht]
  \caption{\alg{SparseCut}{$\cI = (G, F, \cH, \Bc, \Bsource, \Bsink), \kappa, F, \cH$}} \label{alg:sparse-cut}
  
  \SetEndCharOfAlgoLine{}
  \SetKwInput{KwData}{Input}
  \SetKwInput{KwResult}{Output}
  \SetKwProg{KwProc}{function}{}{}
  \SetKwFunction{Relabel}{Relabel}
   \SetKwFor{Loop}{main loop}{}{}

  Let  
  $h \defeq \left\lceil n\cdot \left(\frac{6L}{\psi}+100\kappa\log M\right) \right\rceil$, and $\Bc^{\kappa} \defeq \kappa \cdot \Bc$.

  Run \alg{WeightedPushRelabel}{$G_A,\frac{1}{\psi}\cdot \Bc^{\kappa},\frac{1}{\psi}\cdot\Bsource, \frac{1}{\psi}\cdot \Bsink,\Bw,h$} (\cref{thm:push-relabel-main-theorem}) to get a flow $\Bf'$.\;
  Set $\Bf \defeq \psi \cdot \Bf'$. \tcp*{$\Bf$ is $\psi$-integral}
  \lIf{$|\Bf| = \|\Bsource\|_1$}{
      \Return $\Bf$
  }
  \Else{
    
    Let $\Bw_{\Bf}$ be $\Bw$ extended to $(G_A)_{\Bf}$, except set $\Bw_{\Bf}(\forward{e}) \defeq 0$ for all edges $e\in E(G)\setminus F$ moving forward in the vertex order $\Btau_{\cH}$ (i.e., $e = (u, v)$ such that $\Btau_{\cH}(u) < \Btau_{\cH}(v)$).\; 
  
    Let $S_0 = \{s\in V : \Bsource_{\Bf}(s)>0\}$.\;
    Compute $\Bw_{\Bf}$-distance levels $S_i \defeq \left\{v \in V: \dist_{G_{\Bf}}^{\Bw_{\Bf}}(S_0, v) = i\right\}$ in the residual graph $G_{\Bf}$.
    
    \Return $\Bf$ and the cut $(S_{\leq i}, \overline{S_{\leq i}})$ minimizing $\Bc^{\kappa}_{\Bf}(E_{G_{\Bf}}(S_{\leq i}, \overline{S_{\leq i}})) - \min\{\volfc(S_{\leq i}),\volfc(\overline{S_{\leq i}})\}$.
}
\end{algorithm}

\paragraph{Running Time.}
Note that $G_A$ has at most $n+nL$ vertices and $m+nL$
edges, since there are $L$ layers in the hierarchy $\cH$, and in each layer there is at most one added star edge per vertex.
The weight function $\Bw_\cH$ can be computed in $O(mL)$ time, and the distance layers can be computed in $O((m+nL)\log n)$ time by Dijkstra's shortest-paths algorithm.
By \cref{thm:push-relabel-main-theorem}, the running time of the $\alg{WeightedPushRelabel}{}$ call is:
\begin{align*}
O\left(\left(m + n + \sum_{e\in E(G_A)} \frac{h}{\Bw_{\cH}(e)}\right)\log n
\right)&=  O\left(h\log n \sum_{e\in E(G_A)} \frac{1}{\Bw_{\cH}(e)}\right)
\\
&=  O\left(h\log n \left(
\sum_{e\in A} \frac{1}{\Bw_\cH(e)}
+ \sum_{e\in E(G)} \frac{1}{\Bw_\cH(e)}
\right) \right)
\end{align*}
Since each edge $e\in A$ has $\Bw(e)\ge 1$, we have that
$\sum_{e\in A} \frac{1}{\Bw(e)} = O(nL)$, as there are at most $nL$ star edges.
The term with $E(G)$ can be bounded similarly as \cite[Claim~5.5]{BernsteinBST24} by harmonic sums, since their weights are induced by a permutation $\Btau_{\cH}$:
    \begin{align*}
    \sum_{e\in E(G)} \frac{1}{\Bw(e)}
    &\le \sum_{\substack{\Btau_\cH(u), \Btau_\cH(v) \in [n]\\\Btau_\cH(u)\neq \Btau_\cH(v)}}
    \frac{1}{|\Btau_\cH(u)-\Btau_\cH(v)|}
    = 2\sum_{i=1}^{n}\sum_{j=i+1}^{n}\frac{1}{j-i} = O(n\log n).
\end{align*}

We thus see that the total running time of the algorithm is $O((m+nL)\log n + hn(L+\log n)\log n)$.
Plugging in the value of $h$ from \labelcref{eq:h}, the running time can be bounded by
\begin{align*}
O\left(
n^2\cdot \left(\frac{L}{\psi}+\kappa\log M\right)
\cdot (L + \log n)\cdot \log n
\right).
\end{align*}
Since $L = O(\log n)$ and $M \leq n^{4}$, the running time becomes
$O(n^2\log^3 n \cdot (\kappa+1/\psi))$.

\paragraph{Finding a Sparse Cut.}
In case $|\Bf| = ||\Bsource||_1$, we are done, so let us focus on the case of $|\Bf| < ||\Bsource||_1$, where we need to show that the returned cut satisfies the output guarantee of \cref{lem:push-relabel-on-shortcut}.

\Cref{thm:push-relabel-main-theorem} guarantees that the $\Bw_\cH$-length of any remaining augmenting path is at least $3h$ in the residual graph.
The following lemma, which can be proved in exactly the same way as \cite[Lemma~6.3]{BernsteinBST24}, certifies that even after setting some edge weights to $0$ to obtain $\Bw_{\Bf}$, the $\Bw_{\Bf}$-length of any augmenting path is still at least $h$.
\begin{lemma}
  \label{lemma:many-levels}
  If $\dist_{(G_A)_{\Bf}}^{\Bw_\cH}(S_0, v) > 3h$, then $\dist_{(G_A)_{\Bf}}^{\Bw_{\Bf}}(S_0, v) > h$.
\end{lemma}
\begin{proof}[Proof Sketch.]
On any path $P$, all edges $(u,v)$ going forward in the vertex order $\Btau_{\cH}$ (recall new weight is $\Bw_{\Bf}(u,v) = 0$; old weight was $\Bw_\cH(u,v) = |\Btau(u)-\Btau(v)|$) can be charged to increasing the location in the vertex order, whose total increase is at most $n$ plus the amount we move backwards in the vertex order (and all edges $(u,v)$ moving backwards have weight $\Bw_{\Bf}(u,v) = \Bw_\cH(u,v) \ge |\Btau(u)-\Btau(v)|$). Hence $\Bw_\cH(P)\ge \Bw_{\Bf}(P)\ge \frac{\Bw_\cH(P)-n}{2}$, and the lemma follows since $h > n$.
\end{proof}
By the output guarantee of \cref{thm:push-relabel-main-theorem} together with \cref{lemma:many-levels}, we thus know that among distance layers $S_0, S_1, \ldots S_h$ there is no unsaturated sink. That is, we know that
$\Bsource_{\Bf}(\overline{S_{0}}) = 0$
and for all $0\le i\le h$ that $\Bsink_{\Bf}(S_{i}) = 0$.

\begin{defn}[Good Cuts]
For $0\le i < h$, the distance layer cut $(S_{\le i}, \overline{S_{\le i}})$ is \emph{good}
if the capacity of this cut, in the residual graph,\footnote{Recall that we set $\Bc^{\kappa} = \kappa \cdot \Bc$. and that in the residual graph $(G_A)_{\Bf}$ each edge $e = (u,v)$ has a corresponding forward edge $\forward{e} = (u,v)$ with capacity 
$\Bc^{\kappa}_{\Bf}(\forward{e}) \defeq \Bc^{\kappa}(e) - \Bf(e)$ and a backward edge $\rev{e} = (v, u)$ with capacity $\Bc^{\kappa}_{\Bf}(\backward{e}) \defeq \Bf(e)$.} disregarding edges in $F$, is at most $40|\Bf|$, i.e., if 
\begin{align*}
\Bc^{\kappa}_{\Bf}(E_{(G_A)_{\Bf}}(S_{\le i}, \overline{S_{\le i}}) \setminus F) \le 40 |\Bf|.
\end{align*}
\end{defn}
We will prove that there are at least $h/4$ good cuts. If $F = \emptyset$, this means that the flow $\Bf$ together with such a good cut are both constant-approximate. When $F \neq \emptyset$, assuming we have at least $h/4$ good cuts, one can instead finish the proof with a standard ball growing argument.

We begin by showing that (large) components at a certain level at the hierarchy remain mostly intact and close together in the residual graph.
The following lemma, which has a very short proof due to the construction of the shortcuts, is to be contrasted with the corresponding \cite[Lemma~6.5]{BernsteinBST24} that required a much longer analysis.

\begin{restatable}{lemma}{LowDiameterExpanderNew}
  Consider a level $\ell$ of the hierarchy $\cH$ and a level-$\ell$ component $C$.
  Let $E_C \defeq E_\ell \cap C$, and let $r$ be the Steiner vertex of the star $A_{E_C}$.
  There exists a subset $P_C \subseteq E_C$ such that
  \begin{enumerate}[(1)]
    \item\label{low-diameter-expander-new:item1} for each $e \in E_C \setminus P_C$, we have
    $
    \dist_{(G_A)_{\Bf}}^{\Bw_{\Bf}}(r, \tail(e))
    \leq |C|$ and
    $\dist_{(G_A)_{\Bf}}^{\Bw_{\Bf}}(\tail(e),r)
    \leq |C|$;
    and
    \item\label{low-diameter-expander-new:item2}  $\Bc^{\kappa}(P_C) \leq \frac{2|\Bf|}{\psi}$.
  \end{enumerate}
  \label{lemma:low-diameter-expander-new} 
\end{restatable}
\begin{proof}
Recall that in $G_A$ there is a bidirectional star $A_{E_C}$ added between the tails of edges in $E_C$ to $r$, where all these star edges have a $\psi$-fraction of the capacity of the corresponding edges in $E_C$. Let $P' \subseteq A_{E_C}$ be the subset of the (directed) edges of this star that are saturated by the flow~$\Bf$, i.e., edges $e$ where $\Bf(e) = \Bc^\kappa(e)$. The total capacity of $P'$ is at most $2|\Bf|$, since every flow path passes through the star at most once.\footnote{This uses the fact that the weighted push-relabel of \cite{BernsteinBST24} works by repeatedly pushing flow along simple augmenting paths. If an augmenting path increases the flow value by $a$, then this augmenting path can be responsible for at most $2a$ reduction in the capacity of edges in $A_{E_C}$, as the path passes through the center $r$ at most once.}
Let $P_C$ be all edges in $E_C$ whose tail is incident to an edge in $P'$. In particular, $\Bc^{\kappa}(P_C)\le \frac{2|\Bf|}{\psi}$.
Moreover, all tails of edges $e \in E_C\setminus P_C$ are of distance at most $|C|$ from the star center $r$ in the residual graph, in both directions, since the corresponding star edges are not saturated and have $\Bw$-length $|C|$ (hence $\Bw_{\Bf}$-length $|C|)$.
\end{proof}

We are now ready to show that at least a constant fraction of our distance layer cuts are good.
\begin{lemma}
\label{lemma:good-level-cuts}
There are at least $h/4$ good cuts.
\end{lemma}
\begin{proof}
Consider one of our candidate level cuts $(S_{\le i}, \overline{S_{\le i}})$ where $0\le i < h$, and what edges can cross it in the residual graph $(G_A)_{\Bf}$.
\begin{enumerate}
  \item \ul{Backward edges $\backward{e}$}. These have residual capacities $\Bc^{\kappa}_{\Bf}(\backward{e}) = \Bf(e)$. The contribution of these (across all good level cuts) can be bounded by $\sum_{\backward{e}\in \backward{E}} \Bc^{\kappa}_{\Bf}(\backward{e}) \Bw_{\Bf}(e) = \sum_{e\in E} \Bf(e) \Bw_{\Bf}(e) = \Bw_{\Bf}(\Bf) \le 9h \cdot |\Bf|$ by \cref{thm:push-relabel-main-theorem}\labelcref{item:push-relabel:short-flow} and that $\Bw_{\Bf} \leq \Bw$.
  \item \ul{Forward edges $\forward{e}$, where $e = (u, v)$ is a level-$\ell$ edge with $\Btau_{\cH}(u) > \Btau_{\cH}(v)$ that is contained in a level-$\ell$ component $C$ of $\cH$}, 
  and the corresponding \ul{star edges} (which exist by \cref{obs:backward-has-star}).
   We will use \cref{lemma:low-diameter-expander-new} to show that most level cuts only have few of these edges.
\end{enumerate}
The remaining edges are either in $F$ (which we disregard for in the definition of good cuts), or are edges $e$ where we set $\Bw_{\Bf}(e) = 0$, as they move forward in $\Btau_{\cH}$ (and hence cannot cross our distance layer cuts). 

For each level $\ell$ of $\cH$, and level-$\ell$ component $C$, let $P_C$ be the pruned set in \cref{lemma:low-diameter-expander-new}, and let $E'_{C} \defeq (E_{\ell}\cap C)\setminus P_C$.
Let $A_{E'}\subseteq A_{(E_\ell\cap C)}$ be the shortcut edges (in the star of the level-$\ell$ component~$C$) that are not incident to $P_C$.
Call a cut $(S_{\le i}, \overline{S_{\le i}})$ \emph{bad} if any edge 
$e\in E'_C \cup A_{E'_C}$
crosses it, for any level $\ell$ and level-$\ell$ component $C$.
First we argue that there are at most $h/2$ bad cuts.
In particular, we argue that some level-$\ell$ component $C$ can only be responsible for $3|C|$ bad cuts (where $|C|$ counts the number of vertices in this component). Indeed by \cref{lemma:low-diameter-expander-new}, each (tail) of an edge in $E'_C$ is of distance at most $|C|$ from the Steiner vertex $r$ at the center of the shortcut star on $E_{\ell}\cap C$. Let $i_r$ be the layer such that $S_{i_r} \ni r$. The earliest layer any edge in $E'_C$ can start in is in layer $i_r-|C|$, and the latest layer any edge in $E'_C$ can end is in layer $i_r + |C| + |C|$ (the tail of the edge can be at layer $i_r + |C|$, and then the length of this edge can be another $|C|$ layers). Thus all edges in $E'_C$, and their corresponding star edges, must be within $3|C|$ consecutive distance layers.

For a fixed level $\ell$, all level-$\ell$ components sum up to $n$ in size, and hence they can be responsible for at most $3n$ bad cuts in total. Summing over the $L$ layers, there are at most $3nL \le h/2$ bad cuts.

Now, there are at least $h/2$ remaining non-bad cuts. We will argue that half of them are good. The total capacity of edges over these cuts is at most $9h \cdot |\Bf|$ from the backwards edges.
Each edge in $P_C$ has length at most $|C|$. Also, as $\psi\le 1$, the corresponding star edges connected to $P_C$ contribute at most the same amount (in terms of capacity) as the edges in $P_C$.
Hence, if we add the contribution from the edges in $P_C$ (and the corresponding star edges) for all levels $\ell$ and level-$\ell$ components $C$, we have, by
  \cref{lemma:low-diameter-expander-new}:
\begin{align*}
\sum_{\substack{0\le i\le h \text{ such that}\\\text{$(S_{\le i},\overline{S_{\le i}})$ is not bad}}}
&\Bc^{\kappa}_{\Bf}(E_{(G_A)_{\Bf}}(S_{\le i},\overline{S_{\le i}}))
\\
&\le 9h|\Bf| +  
\sum_{\ell = 1, \ldots, L}\ \sum_{\text{level-$\ell$ component $C$}} |C|\cdot 2\Bc^{\kappa}(P_C) 
\\
&\le 9h|\Bf| + n \cdot L \cdot \frac{4|\Bf|}{\psi}
\\
&\le 10h|\Bf|.
\end{align*}
The last inequality comes from our choice of $h$ in \labelcref{eq:h}, $h\ge \frac{4nL}{\psi}$. By an averaging argument, at most half of the non-good cuts have capacity more than $2\cdot \frac{10h|\Bf|}{h/2} = 40|\Bf|$.
Hence, at least $h/4$ cuts have capacity, in the residual graph and disregarding $F$,  at most $40|\Bf|$, i.e., they are good cuts.
\end{proof}

In case $F = \emptyset$ (as in \cref{cor:approx max flow in shortcut graph}), we would already be done. When $F \neq \emptyset$, what remains is to argue the contribution of these edges via a  standard ball-growing argument stating that at least one of our $h/4$ good cuts $(S_{\le i}, \overline{S_{\le i}})$ has capacity, in $F$, at most $\min\{\volfc(S_{\le i}), \volfc(\overline{S_{\le i}})\}$. We do this in \cref{lemma:exists-sparse-cut} below, and note that our main \cref{lem:push-relabel-on-shortcut} follows almost directly from this.

\begin{restatable}{lemma}{ExistsSparseCut}
  \label{lemma:exists-sparse-cut}
  There exists a level cut $S_{\leq i}$ with $0 \leq i < h$ such that
  \begin{equation}
  \Bc^{\kappa}_{\Bf}(E_{(G_A)_{\Bf}}(S_{\leq i}, \overline{S_{\leq i}})) \leq 40|\Bf| + \min\{\volfc(S_{\leq i}), \volfc(\overline{S_{\leq i}})\}.
  \label{eq:sparse-cut-cf}
  \end{equation}
\end{restatable}
\begin{proof}
  Let $S_{\leq i_1}, \ldots, S_{\leq i_g}$ be the $g\ge h/4$ many good level cuts given by \cref{lemma:good-level-cuts}.
  Let $U_{j} \defeq S_{\leq i_{j \cdot n}} \setminus S_{\leq i_{(j-1) \cdot n}}$ for each $1 \leq j \leq \left\lfloor \frac{g}{n}\right\rfloor$.
  Let $k \defeq \left\lfloor \frac{g}{n}\right\rfloor \geq \frac{g}{2n}$.
  That is, we first split the distance levels into $g$ blocks at $i_1, \ldots, i_{g}$, and then merge every $n$ consecutive blocks to form the $U_j$'s.
  Observe that since $i_{j \cdot n} \geq i_{(j-1)\cdot n} + n$, we must have
  \begin{equation}
    \dist_{(G_A)_{\Bf}}^{\Bw_{\Bf}}(U_j, U_{j+2}) > n
    \label{eq:dist-gap}
  \end{equation}
  for every $j$.
  We will now only consider level cuts that are between some $U_j$ and $U_{j+1}$ and bound the contribution of edges from $F$ to them using a ball-growing argument.
Note that if $\volfc(U_{\leq 1}) = 0$ or $\volfc(\overline{U_{\leq k}}) = 0$, then the lemma is vacuously true, and therefore we assume otherwise.
  We show that there exists a $1 \leq j \leq k$ such that
  \begin{equation}
    \Bc^{\kappa}_{\Bf}\left(E_{(G_{A})_{\Bf}}(U_{\leq j}, \overline{U_{\leq j}}) \cap F\right) \leq \min\{\volfc(U_{\leq j}), \volfc(\overline{U_{\leq j}})\},
    \label{eq:sparse-cut}
  \end{equation}
  which proves the lemma.
  Assume for contradiction that none of the $U_j$ satisfies \eqref{eq:sparse-cut}.
  Because of \eqref{eq:dist-gap} and that the weight of any edge is bounded by $n$
  we know that all edges in $E_{(G_{A})_{\Bf}}(U_{\leq j}, \overline{U_{\leq j}})$ with positive capacities must be in $E_{(G_{A})_{\Bf}}(U_j, U_{j+1})$.
  Let $\volfc^\kappa(S) \defeq \kappa \cdot \volfc(S)$.
  If $\volfc(U_{\leq k/2}) \leq \volfc(\overline{U_{\leq k/2}})$ then we have
  \[ \volfc^\kappa(U_{\leq j}) \geq \volfc^\kappa(U_{\leq j-1}) + \Bc^{\kappa}_{\Bf}\left(E_{(G_{A})_{\Bf}}(U_{\leq j}, \overline{U_{\leq j}}) \cap F\right) \]
  for $1 \leq j \leq k/2$ which with the assumption of \eqref{eq:sparse-cut}, and that $\volfc^\kappa(U_{\le 1})= \kappa \cdot \volfc(U_{\leq 1}) \ge \kappa$, implies that
  \[
    \volfc^\kappa(U_{\leq j}) \geq \left(1+\frac{1}{\kappa}\right)\volfc^\kappa(U_{\leq j-1}) \implies
    \volfc^\kappa(U_{\leq k/2}) \geq \left(1 + \frac{1}{\kappa}\right)^{k/2-1}\kappa >
    2^{(k/2-1)/\kappa}\kappa \ge 2\kappa M.
  \]
  The last inequality above uses $k/2 - 1 \geq k/4 \geq \frac{h}{16n}$ (since by \cref{lemma:good-level-cuts} we have $g \geq h/4$) and, by \labelcref{eq:h}, that $h \geq 100 n\kappa \log M$,
  where $M$ is the total capacity of edges in the input graph.
  This is a contradiction because $\volfc^\kappa(S)$ of any $S$ should always be bounded by $2\kappa M$.
  
  Similarly, if $\volfc(U_{\leq k/2}) > \volfc(\overline{U_{\leq k/2}})$, a symmetric argument gives
  \[
    \volfc^\kappa(\overline{U_{\leq j}}) \geq \left(1+\frac{1}{\kappa}\right)\volfc^\kappa(\overline{U_{\leq j+1}}) \implies
    \volfc^\kappa(\overline{U_{\leq k/2+1}}) \geq \left(1 + \frac{1}{\kappa}\right)^{k/2-1}\kappa
  \]
  In both cases we have arrived at a contradiction, proving the lemma.
\end{proof}

We end the section by finishing \cref{lem:push-relabel-on-shortcut}.

\begin{proof}[Proof of \cref{lem:push-relabel-on-shortcut}]
Note that the capacity, in the residual graph $(G_{A})_{\Bf}$, of any cut can be at most $|\Bf|$ smaller than the capacity of the same cut in the original graph, i.e.,
$\Bc^{\kappa}(E_{(G_A)}(S_{\leq i}, \overline{S_{\leq i}}))
\le \Bc^{\kappa}_{\Bf}(E_{(G_A)_{\Bf}}(S_{\leq i}, \overline{S_{\leq i}})) + |\Bf|$. Hence, we get
$\Bc(E_{G_{A}}(S,\overline{S}))\le\frac{41|\Bf|+\min\{\volfc(S),\volfc(\bar{S})\}}{\kappa}$ when combined with \cref{lemma:exists-sparse-cut}, which concludes the proof of \cref{lem:push-relabel-on-shortcut}.
\end{proof}

\section{Expander Hierarchies: Single-Pass Bottom-Up Construction}
\label{sec:hierarchy}

The goal of this section is to construct in $\tilde{O}(n^{2})$ time a weak expander hierarchy ${\cal H}$ (see \cref{dfn:expander-hierarchy}) of $G$ that is $\tilde{\Omega}(1)$-expanding.
In fact, we show something stronger:
we give a data structure that can, for any flow in the shortcut graph $G_A$, computationally ``unfold'' it to a corresponding flow in $G$ of approximately the same value. This is useful in the final maximum flow algorithm in \cref{sec:maxflow algo}.
This section proves the following \cref{lem:hierarchy and unfolding structuring}, with $\phi_{\rand} = \Omega(1/\log^{7} n)$ being the universal parameter defined later in \cref{lem:exp decomp}.

\begin{lem}
\label{lem:hierarchy and unfolding structuring}There is a randomized algorithm that, given an $n$-vertex graph $G$, in $O(n^{2}\log^{18}n)$ time computes 
\begin{itemize}
\item a hierarchy ${\cal H}$ of $G$ with $L = O(\log n)$ levels,
\item a shortcut $A$ induced by ${\cal H}$ with capacity scale $\psi_{\rand}=\Omega(\phi_{\rand}/\log n)$, and 
\item a ``flow-unfolding'' data structure ${\cal D}_{\unfold}$ such that, with high probability, given any $\psi_{\rand}$-integral feasible flow $\Bf_{A}$ in $G_{A}\defeq G\cup A$, ${\cal D}_{\unfold}$ returns a flow $\Bf\defeq{\cal D}_{\unfold}(\Bf_{A})$ in $G$ routing the same vertex-demand with congestion $3$ in $O(n^2\log^{15} n)$ time.
\end{itemize}
\end{lem}

To formally see why the flow-unfolding data structure of \cref{lem:hierarchy and unfolding structuring} implies that $\cH$ is $\tilde{\Omega}(1)$-expanding, consider any level $i$ and the level-$i$ edge set $E_i$ in $\cH$.
Any $[G\setminus E_{i+1}]$-component-constrained $\Bvol_{E_{i}}$-respecting  demand $(\Bsource,\Bsink)$ can be routed in $G_A$ with congestion $1/\psi_{\rand} = \tilde{O}(1)$ by trivially routing through the shortcut edges. Denote this flow by $\Bf_A$. Then, \Cref{lem:hierarchy and unfolding structuring} can map $\Bf_A$ to a flow $\Bf$ in $G$ with congestion $\tilde{O}(1)$.
This shows that $E_{i}$ is $[G\setminus E_{i+1}]$-component-constrained $\tilde{\Omega}(1)$-expanding in $G$ for every $i$. Thus, ${\cal H}$ is indeed $\tilde{\Omega}(1)$-expanding according to \cref{def:hierarchy}.\footnote{Meticulous readers might notice that we lose an extra $O(\log n)$ factor in congestion here. This is in fact unnecessary because our \cref{lem:exp decomp} only have a congestion blow-up of $1/\phi_{\rand}$. That said, what is described here is an easy way to see how the flow-unfolding data structure is a stronger construct.}

We now describe on a high level how the weak expander hierarchy is constructed.

\paragraph{Dynamics of Levels and Hierarchies.}
We will construct the expander hierarchy in $L=\lceil\log_{10/9}\Bc(E)\rceil+1=O(\log n)$ rounds. In each round $r$, we have a level function $\level^{(r)}:E\rightarrow\{1,\dots,r\}$. We define $E_{i}^{(r)}\defeq \{e: \level^{(r)}(e)=i\}$ as the set of round-$r$ level-$i$ edges. In particular, for round $1$, we initialize $\level^{(1)}(e)\defeq 1$ for all edges $e$ and therefore $E=E_{1}^{(1)}$. 

In round $r\geq1$, we will choose an edge set $E_{r+1}^{(r+1)}$ where $\Bc(E_{r+1}^{(r+1)})\le 0.9 \cdot \Bc(E_{r}^{(r)})$, and then set their level to $r+1$. That is, $\level^{(r+1)}(e)\defeq r+1$ if $e\in E_{r+1}^{(r+1)}$ and the level of other edges remains the same: $\level^{(r+1)}(e) \defeq \level^{(r)}(e)$ if $e\notin E_{r+1}^{(r+1)}$. In particular, we have $E_i^{(r+1)} = E_i^{(r)} \setminus E_{r+1}^{(r+1)}$ for all $i \leq r$.

We call $E_{r}^{(r)}$ the set of \emph{top-level edges }in round $r$. We will build objects for the graph induced by non-top-level edges. Let ${\cal H}^{(r)}=\{{\cal C}_{0}^{(r)},\dots,{\cal C}_{r-1}^{(r)}\}$ be the hierarchy of $G\setminus E_{r}^{(r)}$ induced by $\level^{(r)}$ (restricted to non-top-level edges $E\setminus E_{r}^{(r)}$). Let $A^{(r)}$ be a shortcut induced by ${\cal H}^{(r)}$ with capacity scale $\phi_{\rand}/L$. Note that $A^{(r)}$ does not include the stars on the top-level edges $E_{r}^{(r)}$. For example, $A^{(1)}=\emptyset$ and ${\cal H}^{(1)} = \{{\cal C}^{(1)}_{0}\}$, where ${\cal C}^{(1)}_{0}$ are singleton components. Denote the round-$r$ shortcut graph by $G_{A^{(r)}}\defeq G\cup A^{(r)}$. 

Observe that, for all $i<r$, $E_{i}^{(r)}\subseteq E_{i}^{(r-1)}$ and the family ${\cal C}_{i}^{(r)}$ of level-$i$ components of ${\cal H}^{(r)}$ is a refinement of ${\cal C}_{i}^{(r-1)}$ because $E_{>i}^{(r)}\supseteq E_{>i}^{(r-1)}$. In words, as the number of rounds $r$ increases, the levels of edges only move up. For a fixed level $i<r$, the family ${\cal C}_{i}^{(r)}$ of level-$i$ components only undergoes refinements. Hence, the stars in the level-$i$ shortcut $A_{i}^{(r)}$ only keep splitting.\footnote{In fact, our algorithm does not actually split the stars $A_i^{(r)}$. Instead, we maintain $L$ different versions of shortcut graphs that are fixed once constructed, and delegate each level of routing to its corresponding shortcut graph.}
At the end, we will return ${\cal H}^{(L+1)}$ and guarantee that $E_{L+1} = \emptyset$ which makes $\mathcal{H}^{(L+1)}$ a hierarchy of $G$.

\paragraph{Algorithm: Expander Decomposition Once per Level. }

For each $r\ge1$, once we choose $E_{r+1}^{(r+1)}$, this determines how $\level^{(r+1)},{\cal H}^{(r+1)},G_{A^{(r+1)}}$ can be obtained from $\level^{(r)},{\cal H}^{(r)},G_{A^{(r)}}$. We construct $E_{r+1}^{(r+1)}$ using the following variant of (weak) expander decomposition, which we prove in \cref{sec:decomp}.

\begin{restatable}
[Weak Expander Decomposition]{lemma}{WeakExpanderDecomposition}\label{lem:exp decomp}
There is a universal $\phi_{\rand} = \Theta(1/\log^{7} n)$ where $\phi_{\rand} \in 1/\N$ such that there exists a randomized algorithm that,
given the round-$r$ level function $\level^{(r)}$ and the induced shortcut graph $G_{A^{(r)}}$ with capacity scale $\psi$ of an $n$-vertex graph $G$, computes in $O(n^2\log^9 n \cdot (\log^3 n + 1/\psi))$ time the \emph{round-$(r+1)$ top-level edge set} $E_{r+1}^{(r+1)}$ such that
\begin{enumerate}
\item $\Bc(E_{r+1}^{(r+1)})\le 0.9 \cdot \Bc(E_{r}^{(r)})$ and 
\item\label{item:random-expanding}  $E_{r}^{(r+1)}$ is $[G\setminus E_{r+1}^{(r+1)}]$-component-constrained\footnote{Recall that this is defined to be $\mathcal{C}$-component-constrained for $\mathcal{C} \defeq \SCC(G \setminus E_{r+1}^{(r+1)})$, which is not necessarily the same as that of $G_{A^{(r)}} \setminus E_{r+1}^{(r+1)}$.}
$\Omega(1/\log^6 n)$-expanding in $G_{A^{(r)}}$, with high probability.
\end{enumerate}
Furthermore, conditioned on Item~\ref{item:random-expanding} succeeding, given a $[G\setminus E_{r+1}^{(r+1)}]$-component-constrained $(\Bvol_{E_r^{(r+1)}}\cdot \kappa/z)$-respecting $1/z$-integral demand $(\Bsource,\Bsink)$ for some $\kappa,z\in \N$, there is an $O(n^2\log^{6} n \cdot (\log^3 n + 1/\psi))$ time algorithm that returns a flow $\Bf$ routing $(\Bsource,\Bsink)$ in $G_{A^{(r)}}$ such that $\Bf(e) \leq \frac{\lceil \Bc(e)\cdot\kappa/\phi_{\rand}\rceil}{z}$.
\end{restatable}

Note crucially that the value of $\phi_{\rand}$ does not depend on the scale of the hierarchy, and thus there is no circular dependence; see its proof in \cref{sec:decomp}.

\begin{remark}
The above \cref{lem:exp decomp} is a form of expander decomposition. For example, in the first round when $r = 1$, we have $A^{(r)} = \emptyset$ and $E_{1}^{(1)} = E$.
The lemma then computes a set $R \defeq E^{(2)}_{2}$, of capacity $\Bc(R)\le 0.9 \cdot \Bc(E)$, such that
$E^{(2)}_{1} \defeq E\setminus R$ is $[G\setminus R]$-component-constrained expanding in $G_{A^{(1)}} = G$. That is, the algorithm identifies a not-too-large set of edges $R$, after whose removal all strongly connected components of $G$ are expanders.
\end{remark}

Note that as in our new weighted push-relabel algorithm on shortcut graphs, we do not need the expansion of the hierarchy ${\cal H}^{(r)}$ to guarantee that $E_{r}^{(r+1)}$ is expanding in $G_{A^{(r)}}$.
The expansion of $\mathcal{H}^{(r)}$ is used later when guaranteeing that $E_r^{(r+1)}$ is expanding in the \emph{original graph} $G$ or, equivalently, when proving the flow-unfolding data structure.

Given \Cref{lem:exp decomp}, the algorithm for computing our final expander hierarchy is very simple.
See \cref{alg:hierarchy}.

\begin{algorithm}
\SetEndCharOfAlgoLine{}
Initialize $\level^{(1)}(e)\defeq 1$ for all edge $e\in E$.\;
Set $L\defeq\lceil\log_{10/9}\Bc(E)\rceil+1$.\;
\For{$r = 1, \ldots, L$} {
Let ${\cal H}^{(r)}$ be the hierarchy of $G\setminus E_{r}^{(r)}$ induced by $\level^{(r)}$.\;
Let $A^{(r)}$ be the shortcut induced by ${\cal H}^{(r)}$ with capacity scale $\psi_{\rand} \defeq \phi_{\rand}/L$. Let $G_{A^{(r)}}\defeq G\cup A^{(r)}$.\;
Run \Cref{lem:exp decomp} on $G_{A^{(r)}}$ and compute the round-$(r+1)$ top-level edge set $E_{r+1}^{(r+1)}$.\;
Set $\level^{(r+1)}(e)\defeq\begin{cases}
r+1 & e\in E_{r+1}^{(r+1)},\\
\level^{(r)}(e) & \text{otherwise.}
\end{cases}$\;
}
Return the last hierarchy ${\cal H}^{(L+1)}$ and the shortcut $A^{(L+1)}$.\;
\caption{A bottom-up construction of an expander hierarchy.\label{alg:hierarchy}}
\end{algorithm}

\paragraph{Analysis.} We now analyze \cref{alg:hierarchy}.
Let $z \defeq 200L/\psi_{\rand} \in \N$ be the integrality of the flows that we will work with.
Note that we are being deliberately explicit about the integrality of flows as our push-relabel algorithm works with integral demands and flows (which up to scaling by a $\poly(n)$ factor works for $1/\poly(n)$-integral demands and flows).
We first show that we can do a single level of unfolding by incurring a small congestion.
For intuition one should think of the following lemma as incurring a multiplicative $(1+1/L)$ congestion (while the extra $0.1/L$ factor stems from rounding to preserve integrality) in each level of unfolding, and thus the final congestion in $G$ is $(1+1/L)^L = O(1)$.

\begin{lemma}
  Given a $1/z$-integral flow $\Bf$ in $G_{A^{(r+1)}}$ routing a demand $(\Bsource,\Bsink)$ supported on $V$ with congestion $\kappa/z$ for integer $\kappa \geq z$, we can compute in $O(n\log^{14} n)$ time a $1/z$-integral flow $\Bf^\prime$ in $G_{A^{(r)}}$ routing the same demand with congestion $\kappa^\prime/z$ where $\kappa^\prime \in \N$ satisfies $\kappa^\prime \leq \left(1+\frac{1.01}{L}\right)\kappa$.
  \label{lemma:single-level-unfolding}
\end{lemma}

\begin{proof}
Let $A^{(r+1)}_r$ be the set of level-$r$ stars of the shortcut $A^{(r+1)}$.
For each strongly connected component $C$ of $G\setminus E_{r+1}^{(r+1)}$, there is a star $A_C$ in $A^{(r+1)}_r$ whose leaves are vertices in the component $C$, with a bidirectional edge to each leaf $v$ of capacity $\deg_{E_r^{(r+1)}\cap C}(v)\cdot\phi_{\rand}/L=\deg_{E_r^{(r+1)}\cap C}(v)\cdot\psi_{\rand}$. Let $\Bf^C$ be the flow $\Bf$ restricted to the (bidirectional) edges of the star $A_C$. Since $\Bf$ has congestion $\kappa/z$, the vertex-demand $(\Bsource^C,\Bsink^C)$ of $\Bf^C$ satisfies
\[ \Bsource^C(v),\Bsink^C(v)\le \frac{\psi_{\rand}\kappa}{z} \cdot \deg_{E_r^{(r+1)}\cap C}(v) \leq \frac{\tilde{\kappa}}{z} \cdot \deg_{E_r^{(r+1)} \cap C}(v), \]
where $\tilde{\kappa} \defeq \lceil \psi_{\rand} \kappa \rceil$,
for all vertices $v$ in the component $C$. Moreover, for the center $u$ of the star $A_C$, we have $\Bsource^C(u)=\Bsink^C(u)=\Bf^{\mathrm{out}}(u) = 0$ since $u\notin V$ and the original demand $(\Bsource,\Bsink)$ is supported on $V$.
It follows that $\Bsource^C(C) = \Bsink^C(C)$.

Consider the $1/z$-integral vertex-demand $(\Bsource^+,\Bsink^{+})\defeq \left(\sum_C\Bsource^C,\sum_{C}\Bsink^{C}\right)$, which is both $[G\setminus E_{r+1}^{(r+1)}]$-component-constrained
and $(\Bvol_{E_{r}^{(r+1)}} \cdot \tilde{\kappa} /z)$-respecting by the above discussion.
Invoking \cref{lem:exp decomp}, we can compute in $O(n\log^{14} n)$ time a $1/z$-integral flow $\Bf^{+}$ routing $(\Bsource^{+}, \Bsink^{+})$ in $G_{A^{(r)}}$ such that
\[ \Bf^{+}(e) \leq \frac{\lceil \tilde{\kappa} \cdot \Bc(e)/\phi_{\rand}\rceil}{z} \leq \frac{\tilde{\kappa}/\phi_{\rand} + 1/\psi_{\rand}}{z} \cdot \Bc(e), \]
where we used that $\Bc(e) \geq \psi_{\rand}$.

Note that vertex-demand $(\Bsource^{+}, \Bsink^{+})$ is also routed by the sum of flows $\sum_C\Bf^C$.
 It follows that $\Bf^+$ routes the same vertex-demand as $\sum_C\Bf^C$. Therefore, the flow $\Bf'\defeq\Bf-\sum_C\Bf^C+\Bf^+$ routes the same vertex-demand $(\Bsource,\Bsink)$ as $\Bf$ in $G^\prime \defeq G \cup A^{(r)} \cup (A^{(r+1)} \setminus A_r^{(r+1)})$, is $1/z$-integral, and satisfies
 \begin{equation}
 \begin{split}
 \Bf^\prime(e) \leq \frac{\kappa+\tilde{\kappa}/\phi_{\rand} + 1/\psi_{\rand}}{z} \cdot \Bc_{G^\prime}(e) &\leq \frac{\kappa +  \kappa/L + 1/\phi_{\rand} + 1/\psi_{\rand}}{z} \cdot \Bc_{G^\prime}(e) \\ &\leq \frac{(1+1.01/L)\kappa}{z} \cdot \Bc_{G^\prime}(e).
 \label{eq:cong}
 \end{split}
 \end{equation}

The last inequality in \eqref{eq:cong} follows from
\[ \frac{1}{\phi_{\rand}} + \frac{1}{\psi_{\rand}} \leq \frac{2}{\psi_{\rand}} \leq \frac{0.01}{L} \cdot \kappa \]
since $\kappa \geq z = 200L/\psi_{\rand}$.
What remains is to transform the flow to be supported solely on $G_{A^{(r)}} \defeq G \cup A^{(r)}$.
Observe that for each $1\le i<r$, the partition ${\cal C}^{(r+1)}_i$ is a refinement of the partition ${\cal C}_i^{(r)}$ since the only differences in $\level^{(r)}$ and $\level^{(r+1)}$ is that $\level^{(r+1)}$ lifts some edges to level $r+1$. Therefore, we can naturally map each edge in level-$i$ stars of the shortcut $A^{(r+1)}$ to a corresponding edge in a level-$i$ star of the shortcut $A^{(r)}$. 
Under this mapping, the new flow $\Bf''$ is supported on $G_{A^{(r)}}$, and its vertex-demand is still supported on $V$.
  \end{proof}

Now we can prove \cref{lem:hierarchy and unfolding structuring}.

\begin{proof}[Proof of~\cref{lem:hierarchy and unfolding structuring}]
 We run \cref{alg:hierarchy} to construct the hierarchy $\cH\defeq \cH^{(L+1)}$ and shortcut $A\defeq A^{(L+1)}$.
 By \Cref{lem:exp decomp}, we have $\Bc(E_{r+1}^{(r+1)})\le 0.9 \cdot \Bc(E_{r}^{(r)})$ for each $r\ge1$, and since $L=\lceil\log_{9/10}\Bc(E)\rceil + 1$, we have $\Bc(E^{(L+1)}_{L+1})\le 0.9^{\lceil\log_{10/9}\Bc(E)\rceil}\cdot\Bc(E)<1$, which implies $E^{(L+1)}_{L+1}=\emptyset$. Then, ${\cal H}\defeq{\cal H}^{(L+1)}$ is the hierarchy of $G\setminus E_{L+1}^{(L+1)}=G$, and $A\defeq A^{(L+1)}$ has capacity scale $\phi_{\rand}/L$ by construction.
 The running time is $O(n^2\log^{17} n)$ per level by \cref{lem:exp decomp}, hence $O(n^2\log^{18} n)$ overall.
  For the flow-unfolding data structure $D_{\mathrm{unfold}}$, we proceed by repeatedly unfolding a level-$(r+1)$ flow $\Bf^{(r+1)}$ to a level-$r$ flow $\Bf^{(r)}$ in $G_{A^{(r)}}$ with \cref{lemma:single-level-unfolding}, starting from $\Bf^{(L+1)}\defeq\Bf$ in $G_{A^{(L+1)}}$.
 Initially, $\Bf^{(L+1)}$ has congestion $1=z/z$, and inductively the final flow $\Bf^{(1)}$ routes $(\Bsource,\Bsink)$ in $G_{A^{(1)}}=G$ and has congestion $(1+1.01/L)^L \leq e^{1.01} \leq 3$.
  The running time of the flow-unfolding data structure is $O(n^2\log^{14} n \cdot L) = O(n^2\log^{15} n)$.
  This completes the proof.
\end{proof}

\section{The Maximum Flow Algorithm}
\label{sec:maxflow algo}

Given the algorithm for computing an approximate maximum flow in a shortcut graph from \Cref{sec:push_relabel} and the algorithm for mapping the flow on the shortcut graph to the original graph without blowing up congestion too much from \Cref{sec:hierarchy}, we obtain an approximate maximum flow algorithm summarized as follows.
\begin{lem}
\label{lem:apx flow}
There is a randomized algorithm that, given a graph $G$ and $s,t\in V(G)$ with polynomially bounded integral edge capacities $\Bc \in \N^E$, in $O(n^{2}\log^{18}n)$ time finds an \textbf{integral} $O(1)$-approximate maximum $(s,t)$-flow $\Bzero \leq \Bf \leq \Bc$ and minimum $(s,t)$-cut $S$ in $G$ with high probability.
\end{lem}

Note that \Cref{lem:apx flow} immediately implies our main \Cref{thm:main} by repeatedly solving the problem in the residual graph up to $O(\log n)$ times (recall that we have assumed the capacities are bounded by $\poly(n)$ via capacity scaling). We state the algorithm for \Cref{lem:apx flow} in \Cref{alg:apx flow} and prove its correctness below.

\ifnum\cameraready=0
\begin{algorithm}[H]
\else
\begin{algorithm}
\fi
\SetEndCharOfAlgoLine{}

Using \Cref{lem:hierarchy and unfolding structuring}, compute a hierarchy ${\cal H}$, the shortcut $A$ induced by ${\cal H}$ with capacity scale $\psi$, and the data structure ${\cal D}_{\unfold}$.\;
Given the shortcut graph $G_{A}\defeq G\cup A$ and $s,t\in V(G)$ as input to \Cref{cor:approx max flow in shortcut graph}, compute a $\psi$-integral $O(1)$-approximate maximum $(s,t)$-flow $\Bf_{A}$ in $G_{A}$ and $O(1)$-approximate minimum $(s,t)$-cut $S_{A}$ in $G_{A}$.\;
Set $\Bf\gets{\cal D}_{\unfold}(\Bf_{A})$, scale the flow value on each edge down by a factor of $1/3$, and round $\Bf$ to integral using \cite{KangP15}.\;
Return $(\Bf,S_{A} \cap V)$.\;
\caption{An $O(1)$-approximate maximum $(s,t)$-flow and minimum $(s,t)$-cut algorithm.\label{alg:apx flow}}
\end{algorithm}

\paragraph{Correctness. }

We will prove that $\Bf$ is a \emph{feasible} $(s,t)$-flow in
$G$ and $S_{A} \cap V$ is a $(s,t)$-cut in
$G$ where $\Bc(E_{G}(S_{A} \cap V,\overline{S_A \cap V}))\le O(|\Bf|)$. We have $\Bc(E_{G}(S_{A},\overline{S_A}))\le\Bc(E_{G_{A}}(S_{A},\overline{S_A}))$ because $G\subseteq G_{A}$. Next, $\Bc(E_{G_{A}}(S_{A},\overline{S_A}))\le O(|\Bf_{A}|)$ because, by \Cref{cor:approx max flow in shortcut graph}, $\Bf_{A}$ is a $O(1)$-approximate feasible maximum flow in $G_{A}$. 

It remains to prove that $\Bf$ is a feasible $(s,t)$-flow in $G$ and has value $|\Bf|=\Omega(|\Bf_{A}|)$. Indeed, by \Cref{lem:hierarchy and unfolding structuring}, ${\cal D}_{\unfold}(\Bf_{A})$ is an $(s,t)$-flow in $G$ with congestion at most $3$ and $|{\cal D}_{\unfold}(\Bf_{A})|=|\Bf_{A}|$. After scaling the flow value on each edge by $1/3$ and rounding, $\Bf$ becomes integral and is a feasible $(s,t)$-flow in $G$ of value $|\Bf|=\lceil |\Bf_{A}|/3 \rceil$.

\paragraph{Running Time.}

\Cref{lem:hierarchy and unfolding structuring} takes $O(n^{2}\log^{18}n)$ time. \Cref{cor:approx max flow in shortcut graph} takes $O(n^{2}\log^{3}n/\psi=O(n^{2}\log^{11}n)$ because $\psi=\Omega(1/\log^{8}n)$. The data structure call ${\cal D}_{\unfold}(\Bf_{A})$ takes $O(n^2\log^{15} n)$ time. Finally, flow rounding takes $O(m\log n) = O(n^2\log n)$ time. Overall, it takes $O(n^{2}\log^{18}n)$ time to compute an $O(1)$-approximate maximum flow.
This proves \cref{lem:apx flow}.

Lastly, recursively applying \cref{lem:apx flow} $O(\log n)$ times we solve maximum flow in $O(n^2\log^{19}n)$ time, proving \cref{thm:main}.

\MaxFlow*

\section{Weak Expander Decomposition: Proof of \texorpdfstring{\Cref{lem:exp decomp}}{Lemma 5.2}}
\label{sec:decomp}

In this section we give an algorithm for weak expander decomposition, proving \cref{lem:exp decomp}.

\subsection{Cut-Matching Game}

To compute a weak expander decomposition, we use the now-standard procedures of cut-matching games.
The cut-matching game constructs an expander by the interaction between a cut player and a matching player: Starting from an empty graph, in each iteration, the cut player chooses a cut of the vertices, and the matching player outputs a matching between the two sides of the cut, which is then added to the graph.
The guarantee of the cut-matching games is that there exist good strategies for the cut players so that after a small number of iterations, regardless of the matchings chosen by the matching player, the final graph (which equals the union of these matchings) becomes an expander.

To use this procedure to certify whether a given graph $G$ is an expander or not (or whether a vertex weight is expanding in it), we can implement the matching player by solving a flow problem on the graph that finds a matching between the vertices embeddable with low congestion into $G$ via flow paths.
Thus, after the cut-matching game, the expander constructed by the game is embeddable into $G$ with low congestion, thus proving the expansion of $G$.

The standard cut-matching game requires the matchings we choose to be perfect.
We use a slight variant that only requires \emph{almost} perfect matchings, at the cost of only guaranteeing the expansion of \emph{almost} all vertices.
This procedure has become fairly standard for undirected graphs, and its directed generalization can be proven in various ways.
We describe two ways of achieving this.

\paragraph{(Standard) Cut-Matching Game with Pruning.}
For instance, one can run the standard directed \emph{cut-matching game} (e.g., \cite{Louis10}).
Since it requires the matchings to be perfect, we add \emph{fake} edges to our matchings to make them artificially perfect.
At the end of the cut-matching game, we get an expander that contains few fake edges.
We then run the \emph{expander pruning} algorithm (e.g., \cite{SulserP25}) to extract an expander subgraph supported on a large fraction of vertices and that consists of only real edges (which proves the expansion of these vertices).

\paragraph{Non-Stop Cut-Matching Game.}
Alternatively, one can use a \emph{non-stop} version of the cut-matching game that allows removing small subsets of vertices during the interaction between the cut and the matching players~\cite{SaranurakW19,LiRW25,FleischmannLL25}.
That is, the matching player does not need to return a perfect matching, and vertices that are unmatched in some iterations get removed from the game.
The non-stop cut-matching game is able to guarantee the expansion of all non-removed vertices at the end. Note that this approach only returns weak expanders (in the sense that the expansion of non-removed vertices is potentially certified through removed ones), while the aforementioned one based on pruning returns strong expanders.

While the first version might be more modularized for presentation, we believe the \emph{non-stop cut-matching game} is the simpler approach (both in implementation and analysis) as it avoids the pruning step. Thus we opt for using it for our algorithm\footnote{In \cref{sec:deterministic-b-matching}---where we show how to derandomize our algorithm for the special case of vertex-capacitated graphs---we instead use the former approach of standard cut-matching games plus pruning (as the non-stop version of the cut-matching games are randomized).} (see \cref{lemma:non-stop-cut-matching} below).

\paragraph{Matching Player.} In either version of the cut-matching game, we need to implement the matching player in the game by providing the cut player a (directed and capacitated) matching from some vertex set $A$ to $B$.
We define the following notion of \emph{perfectness} for these matchings.

\begin{definition}
  Consider a partition $(P,Q)$ of $V$.
  A capacitated graph $M$ on $V$ (with capacities $\Bc_M$) is a \emph{$(P,Q)$-matching} if (1) $\deg_M^{-}(v)=0$ for all $v\in P$ and (2) $\deg_{M}^{+}(v)=0$ for all $v\in Q$.
  It is $\Bd$-perfect for some measure $\Bd$ satisfying $\Bd(P)=\Bd(Q)$ if it also holds that (i) $\deg_{M}^{+}(v) = \Bd(v)$ for all $v\in P$ and (ii) $\deg_{M}^{-}(v) = \Bd(v)$ for all $v\in Q$.
  It is \emph{$1/z$-integral} if the capacities $\Bc_M$ are $1/z$-integral.
\end{definition}

Consider a $(P,Q)$-matching $M$.
Let $\Bd_M \in \R^V$ be defined by $\Bd_M(v) \defeq \deg_{M}^{+}(v)$ for $v\in P$ and $\Bd_M(v) \defeq \deg_{M}^{-}(v)$ for $v\in Q$.
That is, $\Bd_M$ is the unique measure such that $M$ is $\Bd_M$-perfect.

We encapsulate what we consider a valid output of the matching player in the non-stop cut-matching game as follows (note that $\Bd^\prime$ is supposed to be a subdemand of $\Bd$ for which we are trying to certify almost expansion).

\begin{definition}
  Let $(P, Q)$ be a partition of $V$ and $\Bd^\prime$ be a vertex measure on $V$ with $\Bd^\prime(P) = \Bd^\prime(Q)$.
  A \emph{$(P, Q, \Bd^\prime, \Delta, m)$-bidirectional-matching} is a pair of matchings $(\overrightarrow{M}, \overleftarrow{M})$ consisting of $m$ edges such that $\Bd_{\overrightarrow{M}}, \Bd_{\overleftarrow{M}} \leq \Bd^\prime$ are submeasure of $\Bd^\prime$ that satisfy $\Bd_{\overrightarrow{M}}(V), \Bd_{\overleftarrow{M}}(V)\geq \Bd^\prime(V)-\Delta$.
  The bidirectional-matching is \emph{$1/z$-integral} if both $\overrightarrow{M}$ and $\overleftarrow{M}$ are.
  \label{def:matching}
\end{definition}

We use the following directed, capacitated, non-stop cut-matching game of \cite{FleischmannLL25} as a black box.
The cut-matching game is formulated not exactly this way in \cite{FleischmannLL25}, and we discuss how our \cref{lemma:non-stop-cut-matching} can be extracted from \cite{FleischmannLL25} in \cref{appendix:cmg}.

\begin{restatable}[Non-Stop CMG \cite{FleischmannLL25}]{lemma}{NonStopCutMatching}
  Consider an integral vertex measure $\Bd$ on $n$ vertices $V$ where $\|\Bd\|_{\infty}\leq\poly(n)$.
  Let $z\in \N$.
  Suppose there is an oracle that on input partition $(P,Q)$ and $1/z$-integral $\Bd^\prime\leq \Bd$ outputs a $1/z$-integral $(P,Q,\Bd^\prime,\Delta,m)$-bidirectional-matching for $\Delta \defeq \delta_{\KRV} \cdot \Bd(V)$ where $\delta_{\KRV} = O(1/\log^2 n)$ is sufficiently small.
  
  Then, there is a randomized algorithm that invokes the oracle $T_{\KRV} = O(\log^2 n)$ times and takes $O(m\log^2 n)$ additional running time to compute a submeasure $\tilde{\Bd} \leq \Bd$ with $\|\tilde{\Bd}\|_1 \geq \frac{7}{8}\|\Bd\|_1$ such that with high probability  $\tilde{\Bd}$ is $T_{\KRV}$-hop $\Omega(1)$-expanding in
  $W \defeq \bigcup_{i \in [k]}\overrightarrow{M}_i \cup \overleftarrow{M}_i$, where $\overrightarrow{M}_i \cup \overleftarrow{M}_i$ is the output of the oracle in the $i$-th invocation.

  \label{lemma:non-stop-cut-matching}
\end{restatable}

\subsection{Implementing the Matching Player}

Note that given a partition $(P, Q)$ of $V$ and vertex measure $\Bd^\prime$, a flow $\Bf$ that routes the demand $(\Bsource, \Bsink)$ defined by $\Bsource \defeq \Bone_P \cdot \Bd^\prime$ and $\Bsink \defeq \Bone_Q \cdot \Bd^\prime$ gives an $\Bd^\prime$-perfect $(P,Q)$-matching via path decomposition of flows:
For any flow $\Bf$ it is standard that we can write it as $\Bf = C_{\Bf} + P_{\Bf}$, where $C_f$ is a circulation\footnote{A \emph{circulation} is a flow that routes the trivial demand $(\Bzero, \Bzero)$.} and the support of $P_{\Bf}$ is acyclic.
The part $P_{\Bf}$ can be further decomposed into $P_{\Bf} = \sum_{i \in [k]}\lambda_i \Bf_{P_i}$ where  $\lambda_i > 0$, and each $\Bf_{P_i}$ is the flow that sends one unit of demand along a simple $(s_i, t_i)$-path $P_i$.
We call such $(C_{\Bf}, \{(\lambda_i, P_i)\}_{i \in [k]})$ a \emph{path decomposition} of $\Bf$.

Now, for any path decomposition $(C_{\Bf}, \{(\lambda_i, P_i)\}_{i \in [k]})$ of a flow $\Bf$ routing the above demand $(\Bsource,\Bsink)$,
the capacitated matching $M$ that contains for each $(u,v)$-path $P_i$ an edge $(u, v)$ with capacity $\Bc_M(u, v) \defeq \lambda_i$ is a $\Bd^\prime$-perfect $(P,Q)$-matching.
This flow perspective will be used frequently.
In particular, it is easy to see that our flow algorithm \cref{lem:push-relabel-on-shortcut} can be used as the oracle required by \cref{lemma:non-stop-cut-matching} when $\Bd = \Bvol_F$.
To achieve a stronger property that we need later, we further require that the flow we found is decomposable into short paths under the weight $\Bw_{\cH}$.
We call a path decomposition $(C_{\Bf}, \{(\lambda_i, P_i)\}_{i \in [k]})$ \emph{$(\Bw_{\cH}, h)$-short} if each path $P_i$ satisfies $\Bw_{\cH}(P_i) \leq h$.\footnote{While our \cref{lem:push-relabel-on-shortcut} already guarantees that the \emph{average} weight of any path decomposition is small, we want the \emph{worst case} weight is small as well (as required by the definition of being $(\Bw_{\cH}, h)$-short). The proof of \cref{lemma:flow-on-shortcut-short-path} thus essentially consists of running \cref{lem:push-relabel-on-shortcut}, taking the paths that are short, and recursing on the demand that was not routed by these short paths.}
The path decomposition is \emph{$1/z$-integral} if each $\lambda_i$ is.
We encapsulate this into the following lemma.
Since it can be quite straightforwardly obtained by running \cref{lem:push-relabel-on-shortcut} multiple times and performing path decomposition to extract short paths, we defer the proof to \cref{appendix:omitted}.
Note that to circumvent the path decomposition barrier in capacitated graphs, instead of outputting the paths $\{P_i\}_{i \in [k]}$, we only return the \emph{representation} $\{(\lambda_i, s_i, t_i)\}_{i \in [k]}$ (i.e., the endpoints of the paths) of the decomposition which suffices for constructing the matchings.

\begin{restatable}{lemma}{FlowOnShortcut}\label{lemma:flow-on-shortcut-short-path}
Let $F\subseteq E$ and ${\cal H}$ be the hierarchy of $G\setminus F$. Let $A$ be the shortcut induced by ${\cal H}$ with capacity scale $\psi \in 1/\N$, and $G_{A}=G\cup A$ be the shortcut graph. Given a diffusion instance ${\cal I}=(\Bsource,\Bsink)$ with $\psi$-integral demand and parameter $\kappa\in \N$, there is a deterministic algorithm that runs on $G_{A}$ in $O(n^{2}\log^4 n \cdot (\kappa+1/\psi))$ time and finds a $\psi$-integral flow $\Bf$ in $G_{A}$ with congestion $O(\kappa\log n)$ and
the representation of a $\psi$-integral $(\Bw_{\cH}, h)$-short path decomposition of $\Bf$ for some $h = O(n\log n \cdot (\kappa + 1/\psi))$.

Additionally, if $|\Bf| < \|\Bsource\|_{1}$, then the algorithm also returns a cut $\emptyset\neq S\subsetneq V$ with 
$\Bsink_{\Bf}(S)=0$ and $\Bsource_{\Bf}(\overline{S})=0$ of size
\[
\Bc(E_{G_{A}}(S,\overline{S}))\le\frac{41\min\{\Bsource(S),\Bsink(\bar{S})\}+\min\{\volfc(S),\volfc(\bar{S})\}}{\kappa}.
\]
\end{restatable}

We now use the flow subroutine of \cref{lemma:flow-on-shortcut-short-path} to construct our matching player.

\begin{lemma}[Matching Player]
  Let $F \subseteq E$ and $\cH$ be a hierarchy of $G \setminus F$.
  Let $A$ be the shortcut induced by $\cH$ with capacity scale $\psi \in 1/\N$ and $G_A \defeq G \cup A$ be the shortcut graph that contains $m^\prime$ edges.
  Then, for any $\phi < 1$ and $\delta < 1$, given a partition $(P, Q)$ of $V$ and $\psi$-integral $\Bd^\prime \leq \Bd$ for $\Bd \defeq \Bvol_F$, there is a deterministic algorithm that runs in $O(n^2\log^4 n \cdot (1/\phi+1/\psi))$ time and computes either
  \begin{itemize}
    \item a balanced $\phi$-sparse cut $S$ such that $\vol_F(S) \in \left[\tfrac{\delta}{2}\vol_F(V), \tfrac{1}{2}\vol_F(V)\right]$, or
    \item a $\psi$-integral $(P,Q,\Bd^\prime,\Delta,m^\prime)$-bidirectional-matching $(\overrightarrow{M},\overleftarrow{M})$ where $\Delta \defeq \delta \cdot \vol_F(V)$.
  \end{itemize}

  In the second case, the matchings are such that, given any $\psi$-integral flow $\Bf_M$ on $M \defeq \overrightarrow{M} \cup \overleftarrow{M}$ of congestion $\kappa$,
  there exists a $\psi$-integral flow $\Bf_{G_A}$ routing the same demand as $\Bf_M$ does in $G_A$ with congestion $O(\kappa \cdot \log n/\phi)$ such that
  $\Bw_{\cH}(\Bf_{G_A}) \leq \|\Bf_M\|_1 \cdot h$
  for some $h = O(n\log n\cdot (1/\phi+1/\psi))$.
  
  \label{lemma:matching-player}
\end{lemma}

\begin{proof}
  Given the partition $(P, Q)$ and $\Bd^\prime$, to find $\overrightarrow{\Bd}$ and $\overrightarrow{M}$ we solve the flow problem defined by $\Bsource \defeq \Bone_P \cdot \Bd^\prime$ and $\Bsink \defeq \Bone_Q \cdot \Bd^\prime$ in $G_{A}$ (which is $\psi$-integral by the input guarantee).
  We invoke \cref{lemma:flow-on-shortcut-short-path} on with $\kappa \defeq \lceil 50/\phi \rceil$ to compute a flow $\forward{\Bf}$ in $G_{A}$ with congestion $O(\log n/\phi)$ routing $(\Bsource,\Bsink)$.
  If $|\forward{\Bf}| < \|\Bsource\|_1$, then \cref{lemma:flow-on-shortcut-short-path} additionally gives us a cut $S$ that contains all the vertices with excess, meaning that (a subset of) $\forward{\Bf}$ routes the demand restricted to $\bar{S}$ (note that this subset of $\Bf$ may still use edges outside of $G_{A}[\bar{S}]$.)
  We first show that $S$, if non-empty, is indeed a $\phi$-sparse cut with respect to $F$.
  By \cref{lemma:flow-on-shortcut-short-path}, we have
  \begin{align*}
    \Bc(E_{G_{A}}(S,\bar{S}))
    &\leq \frac{\phi}{50} \cdot \left(41 \min\{\Bsource(S), \Bsink(\bar{S})\} + \min\{\vol_{F}(S), \vol_{F}(\bar{S})\right) \\
    &< \phi \cdot \min\{\vol_{F}(S), \vol_{F}(\bar{S})\}
  \end{align*}
  using $\Bsource(S) \leq \vol_{F}(S)$ and $\Bsink(\bar{S}) \leq \vol_{F}(\bar{S})$.
  If $\min\{\vol_F(S),\vol_F(\overline{S})\}\geq \delta/2 \cdot \vol_F(V)$, then we terminate the algorithm and return the balanced sparse cut. 
  Otherwise, given the flow $\forward{\Bf}$, we extract the matching $\overrightarrow{M}$
  from the representation $\{(\lambda_i, s_i, t_i)\}_{i \in [k]}$ of a $(\Bw_{\cH}, h)$-short decomposition of $\forward{\Bf}$, where $h = O(n\log n \cdot (\kappa+1/\psi)) = O(n\log n \cdot \max\{1/\phi+1/\psi\})$ is the parameter from \cref{lemma:flow-on-shortcut-short-path}.
  That is, for each $i \in [k]$, we add an edge $(s_i, t_i)$ to $\forward{M}$ with capacity $\lambda_i$, where we note that $\forward{M}$ is $\psi$-integral as the path decomposition is.
  Note that $\Bd_{\overrightarrow{M}}$ is equal to the source (for $P$) and sink (for $Q$) demand routed by $\forward{\Bf}$, and thus $\Bd_{\overrightarrow{M}}(V)$ is equal to $2|\forward{\Bf}|$.
  By \cref{lemma:flow-on-shortcut-short-path}, we have $\Bsink_{\Bf}(S)=0$ and $\Bsource_{\Bf}(\overline{S})=0$, and using that $\Bsource_{\Bf}(S)=\Bsink_{\Bf}(\overline{S})$ we obtain
  \[
    \Bd_{\overrightarrow{M}}(V)=2|\forward{\Bf}|=\Bd^\prime(V)-2\Bsource_{\forward{\Bf}}(S)\geq\Bd^\prime(V) - 2\min\{\vol_F(S), \vol_F(\overline{S})\}\geq \Bd^\prime(V)-\delta \Bd(V).
  \]
  This shows that $\overrightarrow{M}$ is a valid part of a $(P,Q,\Bd^\prime,\Delta,m^\prime)$-bidirectional-matching.
  To construct $\overleftarrow{M}$, we run the same algorithm in the reversed graph $\rev{G_A}$ with the same demand vector (and reverse the direction of the obtained flow and matching), noting that $\rev{G_A}$ is indeed the shortcut graph of $\rev{G}$.
  If neither of the executions found a balanced sparse cut, then we output $(\overrightarrow{M},\overleftarrow{M})$;
  Otherwise, we output one of the sparse cuts.
  The algorithm runs in $O(n^2 \cdot (1/\phi+1/\psi)\cdot\log^4 n)$ time by \cref{lemma:flow-on-shortcut-short-path}.

  We now proceed to the second part of the lemma, i.e., any demand on $M \defeq \forward{M} \cup \backward{M}$ can be routed with low congestion and short paths.
  Recall that our matchings $\overrightarrow{M}$ and $\overleftarrow{M}$ are constructed via \cref{lemma:flow-on-shortcut-short-path}, with the guarantee that
  the path decompositions of $\forward{\Bf}$ and $\backward{\Bf}$ are $(\Bw_{\cH}, h)$-short.
  We prove the claim for $\forward{M}$ since $\backward{M}$ is analogous, and the overall statement follows by summing up the two flows.
  Consider a given $\psi$-integral flow $\Bf_{\forward{M}}$ on $\forward{M}$ of congestion $\kappa$.
  Let $(C_{\forward{\Bf}}, \{(\lambda_i, P_i)\}_{i\in [k]})$ be the underlying $(\Bw_{\cH}, h)$-short path decomposition of $\forward{\Bf}$.
  We can write
  \[ \forward{\Bf} - C_{\forward{\Bf}} = \sum_{u,v} \sum_{\text{$(u,v)$-path $P_i$ with $\Bw_{\cH}(P_i) \leq h$}}\lambda_i \Bf_{P_i}, \]
  where the $\lambda_{P_i}$'s for each $(u, v)$ sum to the total capacity of the $(u, v)$ edge in $\forward{M}$.
  We now construct a flow $\Bf^\prime$ in $G_A$ routing the same demand as $\Bf_{\forward{M}}$ by simply scaling up $\forward{\Bf}$ appropriately for each $(u, v)$.
  In particular, we set
  \begin{equation}
  \Bf^\prime \defeq \sum_{u,v} \left(\frac{\Bf_{\forward{M}}(u,v)}{\Bc_{\forward{M}}(u,v)}\cdot \sum_{\text{$(u,v)$-path $P_i$ with $\Bw_{\cH}(P_i) \leq h$}}\lambda_i \Bf_{P_i}\right).
  \label{eq:flow-in-shortcut}
  \end{equation}
  It is easy to see that $\Bf^\prime$ routes the same demand as $\Bf_{\forward{M}}$.
  We also have $\Bf^\prime(e) \leq \forward{\Bf}(e) \cdot \max_{u,v}\frac{\Bf_{\forward{M}}(u,v)}{\Bc_{\forward{M}}(u,v)} \leq \kappa \forward{\Bf}(e)$ for each $e$, meaning that $\Bf^\prime$ is of congestion $O(\kappa \cdot \log n/\phi)$ in $G_A$.
  Finally, we have $\Bw_{\cH}(\Bf^\prime) \leq \|\Bf_{\forward{M}}\|_1 \cdot h$, since each $P_i$ in \eqref{eq:flow-in-shortcut} has weight $\Bw_{\cH}(P_i) \leq h$, and by definition those $\lambda_i$'s sum up to $\Bc_{\forward{M}}(u, v)$.
  Combining the two $\Bf^\prime$ for $\forward{M}$ and $\backward{M}$ as $\Bf_{G_A}$, this concludes the proof.
\end{proof}

\subsection{Decomposing the Graph into Weak Expanders}
The rest of the proof of our weak expander decomposition \cref{lem:exp decomp} is now mostly standard:
We run the non-stop cut-matching game (\cref{lemma:non-stop-cut-matching}) to attempt to certify the expansion of a large fraction of edges, using our matching player (\cref{lemma:matching-player}) based on weighted push-relabel (\cref{lemma:flow-on-shortcut-short-path}).
In each iteration, either our matching player finds a balanced cut---in which case we simply split the graph at the cut and recurse on both sides---or the cut is unbalanced in which case we have a matching that we feed to \cref{lemma:non-stop-cut-matching} and continue the algorithm.
Thus, we either get a low recursion depth due to the balanced first case, or we successfully finish running the cut-matching game \cref{lemma:non-stop-cut-matching} and certify the expansion of a large fraction of the edges.

The crucial idea that makes our construction different (and simpler) from previous expander decomposition algorithms lies in how we leverage ``weak'' expander decomposition to handle the second case.
To get a ``strong'' expander decomposition, one needs to ensure that each component is expanding \emph{within} itself, meaning that the underlying witness cannot use any edges outside of the component.
Previous work such as \cite{SaranurakW19,SulserP25} thus needs to perform a \emph{trimming} step that extracts from a ``near'' expander to an actual ``strong'' expander.
In contrast, since weak expansion suffice for our purposes, we can allow the witnesses to be embedded into edges that we ``cut'' from the graph, potentially even outside of the components that they are in (as long as the \emph{overall} congestion can be bounded).
This nullifies the need of the trimming step, as we can simply remove the cut edges, and consequently greatly simplifies our algorithm.\footnote{The state-of-the-art trimming algorithms for \emph{directed} graphs \cite{BernsteinGS20,SulserP25} are currently significantly more intricate than the undirected counterpart~\cite{SaranurakW19}.}
We first claim that we can compute a large subset of expanding edges from an expanding vertex measure.

\begin{claim}
  Given an edge set $F$ and a vertex measure $\Bd^\prime \leq \Bd$ with $\|\Bd^\prime\|_1 \geq 0.8\|\Bd\|_1$ where $\Bd \defeq \Bvol_F$, there is an $O(|F|)$ time algorithm that computes an edge subset $F^\prime \subseteq F$ such that $\Bvol_{F^\prime} \leq 2\Bd^\prime$ and $\Bc(F^\prime) \geq 0.2 \cdot \Bc(F)$.
  \label{claim:from-vertex-weight-to-edges}
\end{claim}

\begin{proof}
  Let $U \defeq \{v \in V: \Bd^\prime(v) \leq \frac{1}{2}\Bd(v)\}$.
  We have
  \[
    \Bd(U) \leq 2\left(\Bd(U) - \Bd^\prime(U)\right) \leq 0.4 \cdot \Bd(V).
  \]
  Thus, we simply let $F^\prime$ be the set of edges in $F$ not incident to $U$ which can be easily computed in linear time. We have $\Bc(F^\prime) \geq 0.2 \cdot \Bc(F)$ since $\Bd(V) = 2\Bc(F)$.
  Also, we have $\Bvol_{F^\prime}|_{V\setminus U} \leq \Bd|_{V \setminus U} \leq 2\Bd^\prime|_{V \setminus U}$ and that $\Bvol_{F^\prime}|_{U} \leq 2\Bd^\prime|_U$ since each $u \in U$ has $\deg_{F^\prime}(u) = 0$.
\end{proof}

\paragraph{Splitting the Shortcut Graph.}
Our algorithm for \cref{lem:exp decomp} will involve ``splitting'' a shortcut graph $G_A$ induced by $\cH$ along a cut in $G_A$ (in particular when we find a sparse cut in $G_A$ when doing expander decomposition).
Formally, we overload notation slightly and denote by $(G[S])_A$ the \emph{induced} shortcut graph which is defined to be the shortcut graph of $G[S]$ induced by the hierarchy $\cH[S]$.
Note that while technically $(G[S])_A \cup (G[\bar{S}])_A$ is \emph{not} a subgraph of $G_A$, it behaves essentially like one in the sense that there is a trivial one-way correspondence from edges in $(G[S])_A \cup (G[\bar{S}])_A$ to $G_A$.
We thus have the following useful observation.

\begin{observation}
  Given a flow $\Bf$ in $(G[S])_A \cup (G[\bar{S}])_A$ routing a demand $(\Bsource, \Bsink)$ supported on $V(G)$, there is an $O(m)$ time algorithm that computes a flow $\Bf_G$ routing $(\Bsource,\Bsink)$ in $G_A$.
  \label{obs:trivial}
\end{observation}

Finally, once the expander decomposition is computed, to support the second part of \cref{lem:exp decomp}, we will prove that $E_r^{(r+1)}$ is not only expanding but also $(\Bw_\cH, h)$-length expanding for some reasonably small $h$.
This allows us to efficiently route the demand using the following lemma.

\begin{lemma}
  Let $F \subseteq E$ and $\cH$ be a hierarchy of $G \setminus F$.
  Let $A$ be the shortcut induced by $\cH$ with capacity scale $\psi \in 1/\N$ and $G_A \defeq G \cup A$ be the shortcut graph.
  If $F^\prime \subseteq F$ is $(\Bw_{\cH}, h)$-length $\phi$-expanding in $G_A$ for some $h \in \N$, then there is a deterministic algorithm that, given any $1/z$-integral $(\Bvol_{F^\prime} \cdot \kappa/z)$-respecting demand $(\Bsource,\Bsink)$ for some $\kappa,z \in \N$, computes a $1/z$-flow $\Bf$ routing $(\Bsource,\Bsink)$ in $G_A$ with $\Bf(e) \leq \frac{\lceil c_{\ref{lemma:routing-data-structure}} \cdot \kappa \log n/\phi \cdot \Bc_{G_A}(e) \rceil}{z}$ in $O(m\log n+hn\log^3 n)$ time for some universal constant $c_{\ref{lemma:routing-data-structure}} > 0$.
  
  \label{lemma:routing-data-structure}
\end{lemma}

\begin{proof}
  Consider the given $1/z$-integral $(\Bvol_{F^\prime} \cdot \kappa/z)$-respecting demand $(\Bsource,\Bsink)$.
  Note that $z \cdot (\Bsource,\Bsink)$ is an integral $(\Bvol_{F^\prime}\cdot \kappa)$-respecting demand.
  By definition of $F^\prime$ being $(\Bw_{\cH}, h)$-length $\phi$-expanding in $G_A$, there exists a flow $\Bf^*$ routing $z \cdot (\Bsource,\Bsink)$ in $G_A$ with congestion $\kappa/\phi$ such that $\frac{\Bw_{\cH}(\Bf)}{|\Bf|} \leq h$.
  Running \cref{thm:push-relabel-main-theorem} on $G_A$ with demand $z \cdot (\Bsource,\Bsink)$, we can compute a $\psi$-integral flow $\Bf^\prime$ routing $1/6$-fraction of the demand in $G_A$ with congestion $\kappa/\phi$ in $O(m + hn\log^2 n)$ time.\footnote{This bound was established in the proof of \cref{lem:push-relabel-on-shortcut} and for brevity we omit repeating the calculation here.}
  Repeating this $O(\log n)$ times and summing all the flows, we get a $\psi$-integral flow $\Bf_{\psi}$ routing $z \cdot (\Bsource,\Bsink)$ in $G_A$ with congestion $O(\kappa\log n/\phi)$ in $O(m\log n+hn\log^3 n)$ time.
  We set $c_{\ref{lemma:routing-data-structure}}$ to be the universal constant hidden in the above $O(\cdot)$ expression.
  Using \cref{lemma:rounding}, we can round $1/\psi \cdot \Bf_{\psi}$ to an integral flow $\Bf_z$ routing $z \cdot (\Bsource,\Bsink)$ such that $\Bf_z(e) \leq \lceil c_{\ref{lemma:routing-data-structure}} \cdot \kappa \log n/\phi \cdot \Bc_{G_A}(e) \rceil$.\footnote{More specifically, $1/\psi \cdot \Bf_{\psi}$ routes $z/\psi \cdot (\Bsource,\Bsink)$, and thus \cref{lemma:rounding} gives us an $\Bf_z \leq \lceil \Bf_{\psi} \rceil$ routing $z\cdot(\Bsource,\Bsink)$.}
  Setting $\Bf \defeq \Bf_z/z$ completes the proof.
\end{proof}

We are now ready to prove \cref{lem:exp decomp}, which we restate below.

\WeakExpanderDecomposition*

\begin{proof}
  Let $\phi_{\mathrm{exp}} \defeq \Theta(1/\log^3 n)$ be sufficiently small.
  Let $F \defeq E_{r}^{(r)}$ be the terminal edge set.
  We apply the cut-matching game of \cref{lemma:non-stop-cut-matching} on the graph $G_{A^{(r)}}$ and vertex measure $\Bd \defeq \Bvol_{F}$, using the output of \cref{lemma:matching-player} with parameters $\phi_{\mathrm{exp}}$ and $\delta_{\KRV}$ as the oracle.
  This takes $O(\log^2 n) \cdot O(n^2\log^{4}n \cdot (\log^3 n + 1/\psi))$ time.
  There are the following two cases.
  \begin{enumerate}
    \item\label{item:case1} If in any of the $T_{\KRV}=O(\log^2 n)$ iterations, \cref{lemma:matching-player} outputs a balanced sparse cut $S$ (i.e., $\vol_F(S) > \frac{\delta_{\KRV}}{2} \cdot \vol_F(V)$), then we terminate \cref{lemma:non-stop-cut-matching} and instead simply recurse on both sides of the cuts.
    In particular, let $X$ be the edge set among $E_{G_{A^{(r)}}}(S,\bar{S})$ and $E_{G_{A^{(r)}}}(\bar{S},S)$ that has a smaller total capacities.
    Since $S$ is $\phi$-sparse, we have $\Bc_{G_{A^{(r)}}}(X) < \phi_{\mathrm{exp}} \cdot \vol_F(S)$.
    We add $X \cap E(G)$ into the set $E_{r+1}^{(r+1)}$ and split $G$ into two parts $(G[S])_{A^{(r)}}$ and $(G[\bar{S}])_{A^{(r)}}$, recursing on both of them (with terminal edges $F \cap G[S]$ and $F \cap G[\bar{S}]$).

    \item\label{item:case2} On the other hand, if \cref{lemma:matching-player} outputs a valid $\psi$-integral bidirectional matching for all the $T_{\KRV}$ iterations, then by \cref{lemma:non-stop-cut-matching} we get a vertex measure $\tilde{\Bd}\leq\Bd$ with $\|\tilde{\Bd}\|_1\geq \frac{7}{8}\|\Bd\|_1$ that is $O(\log^2 n)$-hop $\Omega(1/\log^2 n)$-expanding in the graph $W$ defined in \cref{lemma:non-stop-cut-matching}, with high probability.
    We terminate the recursion here and conclude that most of $F$ is expanding in $G_{A^{(r)}}$ and obtain the appropriate routing data structure required.
    Indeed, by \cref{claim:from-vertex-weight-to-edges}, we can compute in $O(m)$ time an edge set $F^\prime \subseteq F$ such that $\Bvol_{F^\prime}\leq 2\tilde{\Bd}$ and $\Bc(F^\prime) \geq 0.2 \cdot \Bc(F)$, and then add the remaining edges $F\setminus F^\prime$ to $E_{r+1}^{(r+1)}$.
    That $\Bvol_{F^\prime}\leq 2\tilde{\Bd}$ means that $F^\prime$ is also $O(\log^2 n)$-hop $\Omega(1)$-expanding in $W$.

    \begin{claim}
      The edge set $F^\prime$ is $(\Bw_\cH, O(h\log^2 n))$-length $\Omega(\phi_{\mathrm{exp}}/\log^3n)$-expanding in $G_A$ for some $h = O(n\log n \cdot (\log^3 n + 1/\psi))$.
      \label{claim:length-expanding}
    \end{claim}

    \begin{proof}
      For any $\Bvol_{F^\prime}$-respecting demand $(\Bsource,\Bsink)$, we construct a flow $\Bf$ routing in $G_A$ with low congestion and average weight under $\Bw_\cH$ as follows.
      Since $F^\prime$ is $O(\log^2 n)$-hop $\Omega(1)$-expanding in $W$, there exists a flow $\Bf_W$ routing $(\Bsource,\Bsink)$ in $W$ with congestion $O(1)$ such that $\frac{\|\Bf_W\|_1}{|\Bf_W|} \leq O(\log^2 n)$.
      The flow $\Bf_W$ induces a demand on the matchings $M_i$'s that form $W$: 
      For each $M_i$, by \cref{lemma:matching-player}, the vertex demand induced by $\Bf_W|_{M_i}$ can be routed in $G_A$ by a flow $\Bf_i$ with congestion $O(\log n/\phi_{\mathrm{exp}})$ such that
      $\Bw_{\cH}(\Bf_i) \leq \|\Bf_W|_{M_i}\|_1 \cdot h$
      for $h = O(n\log n \cdot (1/\phi_{\mathrm{exp}}+1/\psi)) = O(n\log n \cdot (\log^3 n + 1/\psi))$.
      Setting $\Bf$ to be the sum of the $\Bf_i$'s, we have that $\Bf$ is of congestion $O(\log^3 n/\phi_{\mathrm{exp}})$ in $G_A$ and
      \[
        \Bw_{\cH}(\Bf) = \sum_i \Bw_{\cH}(\Bf_i) \leq h \cdot \sum_i \|\Bf_{W}|_{M_i}\|_1 \leq O(h \cdot \log^2 n) \cdot |\Bf_W|.
      \]
      Using the fact that $\Bf$ routes the same demand as $\Bf_W$ does, this completes the proof.
    \end{proof}

  \end{enumerate}

  Overall, algorithm recursively decomposes the graph $G_{A^{(r)}}$ until no balanced sparse cut is found by \cref{lemma:matching-player}.
  The recursion depth is bounded by $O(\log^3 n)$ as the volume of the graph decreases by a factor of $1 - \Omega(1/\log^2 n)$ in each recursion.
  Since the graphs on each level are edge-disjoint, the total capacities of edges added to $E_{r+1}^{(r+1)}$ in Case~\labelcref{item:case1} is bounded by $O(\log^3 n) \cdot \phi_{\mathrm{exp}} \cdot \Bc(E_{r}^{(r)}) \leq 0.1\cdot \Bc(E_{r}^{(r)})$ for $\phi_{\mathrm{exp}} = O(1/\log^3 n)$ sufficiently small.
  For edges added in Case~\ref{item:case2} (i.e., for each leaf nodes in the recursion tree), since the leaves are edge-disjoint and for each of them the capacities of edges we added to $E_{r+1}^{(r+1)}$ is at most $0.8 \cdot \Bc(F)$ for the respective set $F$ of the leaf node, in total we added at most $0.8 \cdot \Bc(E_{r}^{(r)})$ units of capacity to $E_{r+1}^{(r+1)}$ in Case~\ref{item:case2}.
  Therefore, we can bound the final total capacities of $E_{r+1}^{(r+1)}$ by $0.9 \cdot \Bc(E_r^{(r)})$, as required.
  The running time for constructing $E_{r+1}^{(r+1)}$ is $O(n^2\log^{9} n \cdot (\log^3 n + 1/\psi))$.
  This is because at each level we spend $O(n^2\log^{6} n \cdot (\log^3 n + 1/\psi))$ time certifying the expansion of the current edge set $E_r^{(r)}$ and constructing the witness (at the leaf nodes), and there are $O(\log^3 n)$ levels.

  To show that $E_r^{(r+1)}$, which by definition is just $E_r^{(r)} \setminus E_{r+1}^{(r+1)}$, is $[G \setminus E_{r+1}^{(r+1)}]$-component-constrained $\Omega(1/\log^6 n)$-expanding, we first note that the components of $G \setminus E_{r+1}^{(r+1)}$ is a refinement of the components corresponding to the leaf nodes in the recursion.
  By construction, the set $E_r^{(r+1)}$ is the disjoint union of the $F^\prime_C$'s, the edge set $F^\prime$ for each leaf node $C$ that we have proven expansion in \cref{claim:length-expanding} in the corresponding graph $(G[C])_A$.
  Combined with \cref{obs:trivial}, this implies $E_r^{(r+1)}$ is $\Omega(1/\log^6 n)$-expanding because any flow routing $[G\setminus E_{r+1}^{(r+1)}]$-component-constrained $\Bvol_{E_{r}^{(r+1)}}$-respecting demand in $\bigcup_G G[C]_{A^{(r)}}$ (which exists by the expansion guarantee of each $F^\prime_C$) can be converted into an equivalent one in $G_{A^{(r)}}$.
  Using the same argument, for any given $[G\setminus E_{r+1}^{(r+1)}]$-component-constrained $(\Bvol_{E_r^{(r+1)}} \cdot \kappa/z)$-respecting $1/z$-integral demand, we can route it in $G_{A^{(r)}}$ by first splitting split the demand to each component.
  For each component, we apply \cref{lemma:routing-data-structure} to route it within the component and then combine the flows using \cref{obs:trivial}, leveraging \cref{claim:length-expanding} that $F^\prime$ is $(\Bw_\cH, O(n\log^{3}n \cdot (\log^3 n + 1/\psi)))$-length $\Omega(1/\log^6 n)$-expanding.
  The resulting flow $\Bf$ satisfies $\Bf(e) \leq \frac{\lceil \Bc(e) \cdot \kappa/\phi_{\rand}\rceil}{z}$ if we properly define $\phi_{\rand} \defeq \Theta(1/\log^7 n)$ (note that the $\phi$ value in \cref{lemma:routing-data-structure} is $\Theta(1/\log^6 n)$ per \cref{claim:length-expanding}).
  The time to compute such a flow is $O(n^2 \log^{6}n \cdot (\log^3 n + 1/\psi))$ by \cref{lemma:routing-data-structure} with $h = O(n\log^{3}n \cdot (\log^3 n + 1/\psi))$, noting that the leaf components of the recursion tree are disjoint.
  This completes the proof.
\end{proof}

\section*{Acknowledgements}
We thank Aaron Sidford for helpful discussions on deterministic flow algorithms.

\bibliography{reference}

\newcommand{\etalchar}[1]{$^{#1}$}
\begin{thebibliography}{vdBCK{\etalchar{+}}24}

\bibitem[AKR07]{AwerbuchKR07}
Baruch Awerbuch, Rohit Khandekar, and Satish Rao.
\newblock Distributed algorithms for multicommodity flow problems via
  approximate steepest descent framework.
\newblock In {\em Proceedings of the Eighteenth Annual {ACM-SIAM} Symposium on
  Discrete Algorithms, {SODA} 2007}, pages 949--957. {SIAM}, 2007.

\bibitem[ALPS23]{Abboud0PS23}
Amir Abboud, Jason Li, Debmalya Panigrahi, and Thatchaphol Saranurak.
\newblock All-pairs max-flow is no harder than single-pair max-flow: Gomory-hu
  trees in almost-linear time.
\newblock In {\em 64th {IEEE} Annual Symposium on Foundations of Computer
  Science, {FOCS} 2023}, pages 2204--2212. {IEEE}, 2023.

\bibitem[Bar96]{bartal1996probabilistic}
Yair Bartal.
\newblock Probabilistic approximation of metric spaces and its algorithmic
  applications.
\newblock In {\em Proceedings of 37th Conference on Foundations of Computer
  Science}, pages 184--193. IEEE, 1996.

\bibitem[BBST24]{BernsteinBST24}
Aaron Bernstein, Joakim Blikstad, Thatchaphol Saranurak, and Ta{-}Wei Tu.
\newblock Maximum flow by augmenting paths in $n^{2+o(1)}$ time.
\newblock In {\em 65th {IEEE} Annual Symposium on Foundations of Computer
  Science, {FOCS} 2024}, pages 2056--2077. {IEEE}, 2024.

\bibitem[BPS20]{BernsteinGS20}
Aaron Bernstein, Maximilian {Probst Gutenberg}, and Thatchaphol Saranurak.
\newblock Deterministic decremental reachability, {SCC}, and shortest paths via
  directed expanders and congestion balancing.
\newblock In {\em 61st {IEEE} Annual Symposium on Foundations of Computer
  Science, {FOCS} 2020}, pages 1123--1134. {IEEE}, 2020.

\bibitem[BT25]{BBLST25implementation}
Joakim Blikstad and Ta-Wei Tu.
\newblock Implementation of weighted push-relabel maximum flow algorithm on
  shortcut.
\newblock \url{https://github.com/TaWeiTu/combinatorial-max-flow}, 2025.

\bibitem[CGL{\etalchar{+}}20]{ChuzhoyGLNPS20}
Julia Chuzhoy, Yu~Gao, Jason Li, Danupon Nanongkai, Richard Peng, and
  Thatchaphol Saranurak.
\newblock A deterministic algorithm for balanced cut with applications to
  dynamic connectivity, flows, and beyond.
\newblock In {\em 61st {IEEE} Annual Symposium on Foundations of Computer
  Science, {FOCS} 2020}, pages 1158--1167. {IEEE}, 2020.

\bibitem[CHLP23]{CenH0P23}
Ruoxu Cen, William He, Jason Li, and Debmalya Panigrahi.
\newblock Steiner connectivity augmentation and splitting-off in
  poly-logarithmic maximum flows.
\newblock In {\em Proceedings of the 2023 {ACM-SIAM} Symposium on Discrete
  Algorithms, {SODA} 2023}, pages 2449--2488. {SIAM}, 2023.

\bibitem[CK24a]{ChuzhoyK24SODA}
Julia Chuzhoy and Sanjeev Khanna.
\newblock A faster combinatorial algorithm for maximum bipartite matching.
\newblock In {\em Proceedings of the 2024 {ACM-SIAM} Symposium on Discrete
  Algorithms, {SODA} 2024}, pages 2185--2235. {SIAM}, 2024.

\bibitem[CK24b]{ChuzhoyK24STOC}
Julia Chuzhoy and Sanjeev Khanna.
\newblock Maximum bipartite matching in $n^{2+o(1)}$ time via a combinatorial
  algorithm.
\newblock In {\em Proceedings of the 56th Annual {ACM} {SIGACT} Symposium on
  Theory of Computing, {STOC} 2024}. {ACM}, 2024.

\bibitem[CKL{\etalchar{+}}22]{ChenKLPGS22}
Li~Chen, Rasmus Kyng, Yang~P. Liu, Richard Peng, Maximilian {Probst Gutenberg},
  and Sushant Sachdeva.
\newblock Maximum flow and minimum-cost flow in almost-linear time.
\newblock In {\em 63rd {IEEE} Annual Symposium on Foundations of Computer
  Science, {FOCS} 2022}, pages 612--623. {IEEE}, 2022.

\bibitem[CKL{\etalchar{+}}24]{ChenKLMG24}
Li~Chen, Rasmus Kyng, Yang~P. Liu, Simon Meierhans, and Maximilian~Probst
  Gutenberg.
\newblock Almost-linear time algorithms for incremental graphs: Cycle
  detection, sccs, s-t shortest path, and minimum-cost flow.
\newblock In {\em Proceedings of the 56th Annual {ACM} Symposium on Theory of
  Computing, {STOC} 2024}, pages 1165--1173. {ACM}, 2024.

\bibitem[CLN{\etalchar{+}}21]{Cen0NPSQ21}
Ruoxu Cen, Jason Li, Danupon Nanongkai, Debmalya Panigrahi, Thatchaphol
  Saranurak, and Kent Quanrud.
\newblock Minimum cuts in directed graphs via partial sparsification.
\newblock In {\em 62nd {IEEE} Annual Symposium on Foundations of Computer
  Science, {FOCS} 2021}, pages 1147--1158. {IEEE}, 2021.

\bibitem[CMGS25]{ChenMGS25}
Daoyuan Chen, Simon Meierhans, Maximilian~Probst Gutenberg, and Thatchaphol
  Saranurak.
\newblock Parallel and distributed expander decomposition: Simple, fast, and
  near-optimal.
\newblock In {\em Proceedings of the 2025 Annual {ACM-SIAM} Symposium on
  Discrete Algorithms, {SODA} 2025}, pages 1705--1719. {SIAM}, 2025.

\bibitem[Coh95]{Cohen95}
Edith Cohen.
\newblock Approximate max-flow on small depth networks.
\newblock {\em {SIAM} J. Comput.}, 24(3):579--597, 1995.

\bibitem[CRRT09]{ChaudhuriRRT09}
Kamalika Chaudhuri, Satish Rao, Samantha~J. Riesenfeld, and Kunal Talwar.
\newblock A push-relabel approximation algorithm for approximating the
  minimum-degree {MST} problem and its generalization to matroids.
\newblock {\em Theor. Comput. Sci.}, 410(44):4489--4503, 2009.

\bibitem[CS20]{ChangS20}
Yi{-}Jun Chang and Thatchaphol Saranurak.
\newblock Deterministic distributed expander decomposition and routing with
  applications in distributed derandomization.
\newblock In {\em 61st {IEEE} Annual Symposium on Foundations of Computer
  Science, {FOCS} 2020}, pages 377--388. {IEEE}, 2020.

\bibitem[Dan51]{dantzig1951application}
George~B Dantzig.
\newblock Application of the simplex method to a transportation problem.
\newblock {\em Activity analysis and production and allocation}, 1951.

\bibitem[Din70]{dinic1970algorithm}
Efim~A Dinic.
\newblock Algorithm for solution of a problem of maximum flow in networks with
  power estimation.
\newblock In {\em Soviet Math. Doklady}, volume~11, pages 1277--1280, 1970.

\bibitem[DS08]{DaitchS08}
Samuel~I. Daitch and Daniel~A. Spielman.
\newblock Faster approximate lossy generalized flow via interior point
  algorithms.
\newblock In {\em Proceedings of the 40th Annual {ACM} Symposium on Theory of
  Computing}, pages 451--460. {ACM}, 2008.

\bibitem[ET75]{EvenT75}
Shimon Even and Robert~Endre Tarjan.
\newblock Network flow and testing graph connectivity.
\newblock {\em {SIAM} J. Comput.}, 4(4):507--518, 1975.

\bibitem[FF56]{ford1956maximal}
Lester~Randolph Ford and Delbert~R Fulkerson.
\newblock Maximal flow through a network.
\newblock {\em Canadian journal of Mathematics}, 8:399--404, 1956.

\bibitem[FI03]{FleischerI03}
Lisa Fleischer and Satoru Iwata.
\newblock A push-relabel framework for submodular function minimization and
  applications to parametric optimization.
\newblock {\em Discret. Appl. Math.}, 131(2):311--322, 2003.

\bibitem[FLL25]{FleischmannLL25}
Henry Fleischmann, George~Z. Li, and Jason Li.
\newblock Improved directed expander decompositions.
\newblock {\em CoRR}, abs/2507.09729, 2025.

\bibitem[FM12]{frank2012simple}
Andr{\'a}s Frank and Zolt{\'a}n Mikl{\'o}s.
\newblock Simple push-relabel algorithms for matroids and submodular flows.
\newblock {\em Japan journal of industrial and applied mathematics},
  29:419--439, 2012.

\bibitem[GGT89]{GalloGT89}
Giorgio Gallo, Michael~D. Grigoriadis, and Robert~Endre Tarjan.
\newblock A fast parametric maximum flow algorithm and applications.
\newblock {\em {SIAM} J. Comput.}, 18(1):30--55, 1989.

\bibitem[GH61]{gomory1961multi}
Ralph~E Gomory and Tien~Chung Hu.
\newblock Multi-terminal network flows.
\newblock {\em Journal of the Society for Industrial and Applied Mathematics},
  9(4):551--570, 1961.

\bibitem[GLP21]{GaoLP21}
Yu~Gao, Yang~P. Liu, and Richard Peng.
\newblock Fully dynamic electrical flows: Sparse maxflow faster than
  {G}oldberg-{R}ao.
\newblock In {\em 62nd {IEEE} Annual Symposium on Foundations of Computer
  Science, {FOCS} 2021}, pages 516--527. {IEEE}, 2021.

\bibitem[GMQT12]{GranotMQT12}
Frieda Granot, S.~Thomas McCormick, Maurice Queyranne, and Fabio Tardella.
\newblock Structural and algorithmic properties for parametric minimum cuts.
\newblock {\em Math. Program.}, 135(1-2):337--367, 2012.

\bibitem[GR98]{GoldbergR98}
Andrew~V. Goldberg and Satish Rao.
\newblock Beyond the flow decomposition barrier.
\newblock {\em J. {ACM}}, 45(5):783--797, 1998.

\bibitem[GRST21]{GoranciRST21}
Gramoz Goranci, Harald R{\"{a}}cke, Thatchaphol Saranurak, and Zihan Tan.
\newblock The expander hierarchy and its applications to dynamic graph
  algorithms.
\newblock In {\em Proceedings of the 2021 {ACM-SIAM} Symposium on Discrete
  Algorithms, {SODA} 2021}, pages 2212--2228. {SIAM}, 2021.

\bibitem[GT88]{GoldbergT88}
Andrew~V. Goldberg and Robert~Endre Tarjan.
\newblock A new approach to the maximum-flow problem.
\newblock {\em J. {ACM}}, 35(4):921--940, 1988.

\bibitem[HHL{\etalchar{+}}24]{HaeuplerH0RS24}
Bernhard Haeupler, D.~Ellis Hershkowitz, Jason Li, Antti Roeyskoe, and
  Thatchaphol Saranurak.
\newblock Low-step multi-commodity flow emulators.
\newblock In {\em Proceedings of the 56th Annual {ACM} Symposium on Theory of
  Computing, {STOC} 2024}, pages 71--82. {ACM}, 2024.

\bibitem[HHS23]{HaeuplerHS23}
Bernhard Haeupler, D.~Ellis Hershkowitz, and Thatchaphol Saranurak.
\newblock Maximum length-constrained flows and disjoint paths: Distributed,
  deterministic, and fast.
\newblock In {\em Proceedings of the 55th Annual {ACM} Symposium on Theory of
  Computing, {STOC} 2023}, pages 1371--1383. {ACM}, 2023.

\bibitem[HO94]{HaoO94}
Jianxiu Hao and James~B. Orlin.
\newblock A faster algorithm for finding the minimum cut in a directed graph.
\newblock {\em J. Algorithms}, 17(3):424--446, 1994.

\bibitem[HRG00]{HenzingerRG00}
Monika~Rauch Henzinger, Satish Rao, and Harold~N. Gabow.
\newblock Computing vertex connectivity: New bounds from old techniques.
\newblock {\em J. Algorithms}, 34(2):222--250, 2000.

\bibitem[HRG22]{HaeuplerR022}
Bernhard Haeupler, Harald R{\"{a}}cke, and Mohsen Ghaffari.
\newblock Hop-constrained expander decompositions, oblivious routing, and
  distributed universal optimality.
\newblock In {\em {STOC} '22: 54th Annual {ACM} {SIGACT} Symposium on Theory of
  Computing}, pages 1325--1338. {ACM}, 2022.

\bibitem[HRW17]{HenzingerRW17}
Monika Henzinger, Satish Rao, and Di~Wang.
\newblock Local flow partitioning for faster edge connectivity.
\newblock In {\em Proceedings of the Twenty-Eighth Annual {ACM-SIAM} Symposium
  on Discrete Algorithms, {SODA} 2017}, pages 1919--1938. {SIAM}, 2017.

\bibitem[Iwa03]{Iwata03}
Satoru Iwata.
\newblock A faster scaling algorithm for minimizing submodular functions.
\newblock {\em {SIAM} J. Comput.}, 32(4):833--840, 2003.

\bibitem[Kar73]{karzanov1973finding}
Alexander~V Karzanov.
\newblock On finding maximum flows in networks with special structure and some
  applications.
\newblock {\em Matematicheskie Voprosy Upravleniya Proizvodstvom}, 5:81--94,
  1973.

\bibitem[Kat20]{Kathuria20}
Tarun Kathuria.
\newblock A potential reduction inspired algorithm for exact max flow in almost
  $\widetilde{O}(m^{4/3})$ time.
\newblock {\em CoRR}, abs/2009.03260, 2020.

\bibitem[Kav24]{Kavi24}
Nithin Kavi.
\newblock Partial implementation of max flow and min cost flow in almost-linear
  time.
\newblock {\em CoRR}, abs/2407.10034, 2024.

\bibitem[KKOV07]{KhandekarKOV07}
Rohit Khandekar, Subhash~A. Khot, Lorenzo Orecchia, and Nisheeth~K. Vishnoi.
\newblock On a cut-matching game for the sparsest cut problem.
\newblock Technical Report UCB/EECS-2007-177, EECS Department, University of
  California, Berkeley, Dec 2007.

\bibitem[KL15]{KargerL15}
David~R. Karger and Matthew~S. Levine.
\newblock Fast augmenting paths by random sampling from residual graphs.
\newblock {\em {SIAM} J. Comput.}, 44(2):320--339, 2015.

\bibitem[KLS20]{KathuriaLS20}
Tarun Kathuria, Yang~P. Liu, and Aaron Sidford.
\newblock Unit capacity maxflow in almost $m^{4/3}$ time.
\newblock In {\em 61st {IEEE} Annual Symposium on Foundations of Computer
  Science, {FOCS} 2020}, pages 119--130. {IEEE}, 2020.

\bibitem[KP15]{KangP15}
Donggu Kang and James Payor.
\newblock Flow rounding.
\newblock {\em CoRR}, abs/1507.08139, 2015.

\bibitem[LNP{\etalchar{+}}21]{LiNPSY21}
Jason Li, Danupon Nanongkai, Debmalya Panigrahi, Thatchaphol Saranurak, and
  Sorrachai Yingchareonthawornchai.
\newblock Vertex connectivity in poly-logarithmic max-flows.
\newblock In {\em Proceedings of the 53rd Annual {ACM} {SIGACT} Symposium on
  Theory of Computing, STOC 2021}, pages 317--329. {ACM}, 2021.

\bibitem[Lou10]{Louis10}
Anand Louis.
\newblock Cut-matching games on directed graphs.
\newblock {\em CoRR}, abs/1010.1047, 2010.

\bibitem[LP20]{LiP20}
Jason Li and Debmalya Panigrahi.
\newblock Deterministic min-cut in poly-logarithmic max-flows.
\newblock In {\em 61st {IEEE} Annual Symposium on Foundations of Computer
  Science, {FOCS} 2020}, pages 85--92. {IEEE}, 2020.

\bibitem[LRW25]{LiRW25}
Jason Li, Satish Rao, and Di~Wang.
\newblock Congestion-approximators from the bottom up.
\newblock In {\em Proceedings of the 2025 Annual {ACM-SIAM} Symposium on
  Discrete Algorithms, {SODA} 2025}, pages 2111--2131. {SIAM}, 2025.

\bibitem[LS14]{LeeS14}
Yin~Tat Lee and Aaron Sidford.
\newblock Path finding methods for linear programming: Solving linear programs
  in $\widetilde{O}(\sqrt{rank})$ iterations and faster algorithms for maximum
  flow.
\newblock In {\em 55th {IEEE} Annual Symposium on Foundations of Computer
  Science, {FOCS} 2014}, pages 424--433. {IEEE} Computer Society, 2014.

\bibitem[LS20a]{LiuS20Divergence}
Yang~P. Liu and Aaron Sidford.
\newblock Faster divergence maximization for faster maximum flow.
\newblock {\em CoRR}, abs/2003.08929, 2020.

\bibitem[LS20b]{LiuS20Energy}
Yang~P. Liu and Aaron Sidford.
\newblock Faster energy maximization for faster maximum flow.
\newblock In {\em Proceedings of the 52nd Annual {ACM} {SIGACT} Symposium on
  Theory of Computing, {STOC} 2020}, pages 803--814. {ACM}, 2020.

\bibitem[Mad13]{Madry13}
Aleksander Madry.
\newblock Navigating central path with electrical flows: From flows to
  matchings, and back.
\newblock In {\em 54th Annual {IEEE} Symposium on Foundations of Computer
  Science, {FOCS} 2013}, pages 253--262. {IEEE} Computer Society, 2013.

\bibitem[Mad16]{Madry16}
Aleksander Madry.
\newblock Computing maximum flow with augmenting electrical flows.
\newblock In {\em {IEEE} 57th Annual Symposium on Foundations of Computer
  Science, {FOCS} 2016}, pages 593--602. {IEEE} Computer Society, 2016.

\bibitem[OZ14]{OrecchiaZ14}
Lorenzo Orecchia and Zeyuan~Allen Zhu.
\newblock Flow-based algorithms for local graph clustering.
\newblock In {\em Proceedings of the Twenty-Fifth Annual {ACM-SIAM} Symposium
  on Discrete Algorithms, {SODA} 2014}, pages 1267--1286. {SIAM}, 2014.

\bibitem[Qua24]{Quanrud24}
Kent Quanrud.
\newblock Faster exact and approximation algorithms for packing and covering
  matroids via push-relabel.
\newblock In {\em Proceedings of the 2024 {ACM-SIAM} Symposium on Discrete
  Algorithms, {SODA} 2024}, pages 2305--2336. {SIAM}, 2024.

\bibitem[RST14]{RackeST14}
Harald R{\"{a}}cke, Chintan Shah, and Hanjo T{\"{a}}ubig.
\newblock Computing cut-based hierarchical decompositions in almost linear
  time.
\newblock In {\em Proceedings of the Twenty-Fifth Annual {ACM-SIAM} Symposium
  on Discrete Algorithms, {SODA} 2014}, pages 227--238. {SIAM}, 2014.

\bibitem[SP25]{SulserP25}
Aurelio~L. Sulser and Maximilian {Probst Gutenberg}.
\newblock Near-optimal algorithm for directed expander decompositions.
\newblock In {\em 52nd International Colloquium on Automata, Languages, and
  Programming, {ICALP} 2025}, volume 334 of {\em LIPIcs}, pages 132:1--132:20.
  Schloss Dagstuhl - Leibniz-Zentrum f{\"{u}}r Informatik, 2025.

\bibitem[ST83]{SleatorT83}
Daniel~Dominic Sleator and Robert~Endre Tarjan.
\newblock A data structure for dynamic trees.
\newblock {\em J. Comput. Syst. Sci.}, 26(3):362--391, 1983.

\bibitem[ST04]{SpielmanT04}
Daniel~A. Spielman and Shang{-}Hua Teng.
\newblock Nearly-linear time algorithms for graph partitioning, graph
  sparsification, and solving linear systems.
\newblock In {\em Proceedings of the 36th Annual {ACM} Symposium on Theory of
  Computing}, pages 81--90. {ACM}, 2004.

\bibitem[SW19]{SaranurakW19}
Thatchaphol Saranurak and Di~Wang.
\newblock Expander decomposition and pruning: Faster, stronger, and simpler.
\newblock In {\em Proceedings of the Thirtieth Annual {ACM-SIAM} Symposium on
  Discrete Algorithms, {SODA} 2019}, pages 2616--2635. {SIAM}, 2019.

\bibitem[vdBCK{\etalchar{+}}24]{vdBCKLMPS24}
Jan van~den Brand, Li~Chen, Rasmus Kyng, Yang~P. Liu, Simon Meierhans,
  Maximilian~Probst Gutenberg, and Sushant Sachdeva.
\newblock Almost-linear time algorithms for decremental graphs: Min-cost flow
  and more via duality.
\newblock In {\em 65th {IEEE} Annual Symposium on Foundations of Computer
  Science, {FOCS} 2024}, pages 2010--2032. {IEEE}, 2024.

\bibitem[vdBCP{\etalchar{+}}23]{Brand0PKLGSS23}
Jan van~den Brand, Li~Chen, Richard Peng, Rasmus Kyng, Yang~P. Liu, Maximilian
  {Probst Gutenberg}, Sushant Sachdeva, and Aaron Sidford.
\newblock A deterministic almost-linear time algorithm for minimum-cost flow.
\newblock In {\em 64th {IEEE} Annual Symposium on Foundations of Computer
  Science, {FOCS} 2023}, pages 503--514. {IEEE}, 2023.

\bibitem[vdBGJ{\etalchar{+}}22]{BrandGJLLPS22}
Jan van~den Brand, Yu~Gao, Arun Jambulapati, Yin~Tat Lee, Yang~P. Liu, Richard
  Peng, and Aaron Sidford.
\newblock Faster maxflow via improved dynamic spectral vertex sparsifiers.
\newblock In {\em Proceedings of the 54th Annual {ACM} {SIGACT} Symposium on
  Theory of Computing, STOC 2022}, pages 543--556. {ACM}, 2022.

\bibitem[vdBLL{\etalchar{+}}21]{BrandLLSS0W21}
Jan van~den Brand, Yin~Tat Lee, Yang~P. Liu, Thatchaphol Saranurak, Aaron
  Sidford, Zhao Song, and Di~Wang.
\newblock Minimum cost flows, mdps, and $\ell_1$-regression in nearly linear
  time for dense instances.
\newblock In {\em Proceedings of the 53rd Annual {ACM} {SIGACT} Symposium on
  Theory of Computing, STOC 2021}, pages 859--869. {ACM}, 2021.

\bibitem[vdBLN{\etalchar{+}}20]{BrandLNPSS0W20}
Jan van~den Brand, Yin~Tat Lee, Danupon Nanongkai, Richard Peng, Thatchaphol
  Saranurak, Aaron Sidford, Zhao Song, and Di~Wang.
\newblock Bipartite matching in nearly-linear time on moderately dense graphs.
\newblock In {\em 61st {IEEE} Annual Symposium on Foundations of Computer
  Science, {FOCS} 2020}, pages 919--930. {IEEE}, 2020.

\bibitem[Vis92]{Vishkin92}
Uzi Vishkin.
\newblock A parallel blocking flow algorithm for acyclic networks.
\newblock {\em J. Algorithms}, 13(3):489--501, 1992.

\end{thebibliography}

\appendix

\section{Deterministic Vertex-Capacitated Flow}
\label{sec:deterministic-b-matching}

In this section we show that how our algorithm can be derandomized in the special case of vertex-capacitated flows, proving \cref{thm:vertex-flow}.
Note that the only randomized component of our algorithm is the non-stop cut-matching game of \cref{lemma:non-stop-cut-matching} that runs in near-linear time with respect to the number of edges returned by the matching player which can be as large as $\Omega(n^2)$ in the edge-capacitated maximum flow instances.
However, when the graph is vertex-capacitated, we can substitute the randomized cut-matching game with a deterministic one (see, e.g.,~\cite{ChuzhoyGLNPS20}) that runs in near-quadratic time instead of near-linear time, leveraging the fact that the matching player always outputs a matching of size $\tilde{O}(n)$.\footnote{Recall that the size of the matching is the number of edges it is supported by, which in \cref{sec:decomp} was bounded only by $\tilde{O}(m)$.}

Consider a graph $G = (V,E)$ with vertex capacities $\Bc_G\in \N^V$.
We do the standard transformation of splitting each vertex $v$ into an in-vertex $v_{\mathrm{in}}$ and an out-vertex $v_{\mathrm{out}}$ with a directed edge $(v_{\mathrm{in}}, v_{\mathrm{out}})$ of capacity $\Bc_G(v)$ between them.
Edges in $G$ are directed from the out-vertex of its tail to the in-vertex of its head, with infinite edge capacity.
Let $\tilde{G}$ be this edge-capacitated graph.
It is easy to see that there is a bijection between $(s_{\mathrm{out}}, t_{\mathrm{in}})$-flows in $\tilde{G}$ respecting edge capacities and $(s, t)$-flow in $G$ respecting vertex capacities of the same value.
Thus we will run our maximum flow algorithm on $\tilde{G}$.
Note that $\tilde{G}$ contains $n$ finite-capacity edges which we exploit to argue that the matching player always outputs a sparse matching.
Recall that our algorithm works by iteratively finding $O(1)$-approximate flows in residual graphs, and while this property holds for the initial input graph, it may not hold for an arbitrary residual graph. 
Therefore, we first show that we can ``round'' any flow in $\tilde{G}$ to only be supported on few, $O(n)$ edges.
This ensures that the residual graph with respect to this rounded flow will only contain $O(n)$ finite-capacity edges as well.

\begin{lemma}
  Given an integral $(s, t)$-flow $\Bf$ in $\tilde{G}$, there is a \textbf{deterministic} $O(m\log n)$ time algorithm that computes another integral $(s, t)$-flow $\Bf^\prime$, of value $|\Bf^\prime| = |\Bf|$, that contains $3n$ non-zero entries.
  \label{lemma:round-to-saturate}
\end{lemma}

\begin{proof}
  We consider a correspondence between vertex-capacitated flows and bipartite $\Bb$-matchings.
  Consider a bipartite graph where the left bipartition consists of all out-vertices $v_{\mathrm{in}}$ and the right bipartition all in-vertices $v_{\mathrm{in}}$.
  Each edge $(u_{\mathrm{out}}, v_{\mathrm{in}})$ is included to the bipartite graph as an undirected edge.
  We can interpret the flow $\Bf$ naturally as a $\Bb$-matching in the bipartite graph where $\Bb(v_{\mathrm{in}})=\Bb(v_{\mathrm{out}})=\Bfout(v)$ for $v\not\in\{s,t\}$, $\Bb(s_{\mathrm{out}})=\Bfout(s)$, and $\Bb(t_{\mathrm{in}})=\Bfin(t)$.
  Conversely, any $\Bb$-matching $\Bx$ in the bipartite graph corresponds to an $(s, t)$-flow in $G$ with the same value as $\Bf$ and the same support (modulo the edges between in- and out-vertices).

  Suppose there is a cycle $C = (e_1,\ldots,e_2,\ldots,e_t)$ (as a sequence of edges) in the support of $\Bx$.
  We can observe that if we add $\delta$ to $\Bx(e_i)$ for odd $i$ and $-\delta$ to $\Bx(e_i)$ for even $i$, where $\delta \defeq \min_{j}\Bx(e_{2j})$, then this is still a $\Bb$-matching and with one less non-zero edges.
  Repeating this process until there is no cycles in the support of $\Bx$, we see that now $\Bx$ contains at most $2n$ non-zero entries (and thus the corresponding flow contains $3n$ non-zero entries).
  To implement the algorithm, we can use a standard link-cut tree data structure to efficiently discover cycles and alter the values of $\Bx$ along the cycle.
  This will take $O(m\log n)$ time.
\end{proof}

In \cref{sec:deterministic-apx-flow} we prove the following lemma that computes an $O(1)$-approximate maximum flow \emph{deterministically} when there are few finite-capacity edges.

\begin{lemma}
  There is a \textbf{deterministic} algorithm that, given a graph $G=(V,E)$ with edge capacities $\Bc \in (\N \cup \{\infty\})^{E}$ such that $E_{\infty} \defeq \{e: \Bc(e)=\infty\}$ forms a DAG and $|E\setminus E_{\infty}|=O(n)$, computes an $O(1)$-approximate integral $(s, t)$-flow in $O(n^2\log^{45} n)$ time.
  \label{lemma:determinisitc-apx-flow}
\end{lemma}

Now \cref{thm:vertex-flow} follows from \cref{lemma:determinisitc-apx-flow}.

\VertexFlow*

\begin{proof}
  We start with the empty flow $\Bf$ and iteratively apply \cref{lemma:determinisitc-apx-flow}.
  We maintain the invariant that $\Bf$ always contains at most $3n$ non-zero entries.
  In each iteration, in $O(n^2\log^{45} n)$ time we can find an $O(1)$-approximate integral flow $\Bf^\prime$ in $G_{\Bf}$, noting that (1) $G_{\Bf}$ contains $4n$ edges with finite capacities (including the initial $n$ edges between in- and out-vertices and the additional $3n$ ones from the residual edges) and (2) the infinite-capacity edges form a DAG (since they all go from out-vertices to in-vertices).
  We update our $\Bf$ by combining it with $\Bf^\prime$, and then apply \cref{lemma:round-to-saturate} to reduce the number of non-zero edges back to $3n$.
  This takes $O(\log n)$ iterations until we find a maximum flow, proving the theorem.
\end{proof}

\subsection{Deterministic Approximate Flow Algorithm}\label{sec:deterministic-apx-flow}

Throughout this section we consider the input graph $G$ to \cref{lemma:determinisitc-apx-flow}.
In particular we let $E_{\infty}$ denote the infinite-capacity edges and assume that they form a DAG and contain the majority of edges.

\begin{assumption}
  The set of infinite-capacity edges $E_{\infty} \defeq \{e: \Bc(e)=\infty\}$ forms a DAG and the number of finite-capacity edges is $|E\setminus E_{\infty}|=O(n)$.
  \label{assumption:DAG}
\end{assumption}

Recall that the randomized analogue of \cref{lemma:determinisitc-apx-flow}, i.e., \cref{lem:apx flow}, is proven via \cref{alg:apx flow}.
The first step of \cref{alg:apx flow} is to construct an expander hierarchy $\mathcal{H}$, its shortcut, and an unfolding data structure.
We provide a determinstic algorithm analogous to \cref{lem:hierarchy and unfolding structuring}, where $\phi_{\det} = O(1/\log^{24} n)$ is universal and defined in \cref{lem:deterministic exp decomp}---the deterministic analogue of \cref{lem:exp decomp}.

\begin{lem}
\label{lem:deterministic hierarchy and unfolding structuring}There is a \textbf{deterministic} algorithm that, given an $n$-vertex graph $G$ satisfying \cref{assumption:DAG}, in $O(n^{2}\log^{36}n)$ time computes 
\begin{itemize}
\item a hierarchy ${\cal H}$ of $G$ with $L = O(\log n)$ levels,
\item a shortcut $A$ induced by ${\cal H}$ with capacity scale $\psi_{\det}=\Theta(\phi_{\det}/\log n)$, and 
\item a ``flow-unfolding'' data structure ${\cal D}_{\unfold}$ such that,  given any $\psi$-integral feasible flow $\Bf_{A}$ in $G_{A}\defeq G\cup A$, ${\cal D}_{\unfold}$ returns a flow $\Bf\defeq{\cal D}_{\unfold}(\Bf_{A})$ in $G$ routing the same vertex-demand with congestion $3$ in $O(n^2\log^{45} n)$ time.
\end{itemize}
\end{lem}

Using \cref{lem:deterministic hierarchy and unfolding structuring}, \cref{lemma:determinisitc-apx-flow} follows in exactly the same way as \cref{lem:apx flow} in \cref{sec:maxflow algo}.
The running time is $O(n^2\log^{45} n)$.
To prove \cref{lem:deterministic hierarchy and unfolding structuring}, we use a deterministic analog of \cref{lem:exp decomp}.

\begin{lemma}
\label{lem:deterministic exp decomp}
There is a universal $\phi_{\det} = \Theta(1/\log^{24} n)$ where $\phi_{\det} \in 1/\N$ such that there exists a \textbf{deterministic} algorithm that,
given the round-$r$ level function $\level^{(r)}$ such that $E_r^{(r)} \cap E_{\infty} = \emptyset$ and the induced shortcut graph $G_{A^{(r)}}$ with capacity scale $\psi$ of an $n$-vertex graph $G$, computes in $O(n^2\log^{10} n \cdot (\log^7 n + 1/\psi))$ time the \emph{round-$(r+1)$ top-level edge set} $E_{r+1}^{(r+1)}$ such that
\begin{itemize}
\item $\Bc(E_{r+1}^{(r+1)})\le 0.9 \cdot \Bc(E_{r}^{(r)})$, and 
\item  $E_{r}^{(r+1)}$ is $[G\setminus E_{r+1}^{(r+1)}]$-component-constrained
  $\Omega(1/\log^{23} n)$-expanding in $G_{A^{(r)}}$.
\end{itemize}
Furthermore, given a $[G\setminus E_{r+1}^{(r+1)}]$-component-constrained $(\Bvol_{E_r^{(r+1)}}\cdot \kappa/z)$-respecting $1/z$-integral demand $(\Bsource,\Bsink)$ for some $\kappa,z\in \N$, there is an $O(n^2\log^{19} n \cdot (\log^5 n + 1/\psi))$ time algorithm that returns a flow $\Bf$ routing $(\Bsource,\Bsink)$ in $G_{A^{(r)}}$ such that $\Bf(e) \leq \frac{\lceil \Bc(e)\cdot\kappa/\phi_{\det}\rceil}{z}$.
\end{lemma}

\begin{proof}[Proof of \cref{lem:deterministic hierarchy and unfolding structuring}]
We first put $E_{\infty}$ into the first level of the hierarchy and then build a hierarchy on the rest of the edges, i.e., we set $\level^{(1)}(e) \defeq 1$ for all $e \in E$, $\level^{(2)}(e) \defeq 1$ for $e \in E_{\infty}$ and $\level^{(2)}(e) \defeq 2$ otherwise.
Note that this is vacuously valid because $E_{\infty}$ forms a DAG by \cref{assumption:DAG}, so everything is $(G \setminus E_{2}^{(2)})$-component-constrained expanding (the shortcut $A^{(1)}$ is trivial since we only consider intra-component edges in the construction of stars; see \cref{dfn:star}).
We now build the rest of the hierarchy via \cref{lem:deterministic exp decomp}.
Note that infinite-capacity edges will never appear in upper levels since $\Bc(E_{r+1}^{(r+1)})\leq \Bc(E_{r}^{(r)})$ where the right-hand side is finite for $r\geq 2$.
This and the flow unfolding data structure are the same as \cref{lem:hierarchy and unfolding structuring}. 
\end{proof}

Now, we prove \cref{lem:deterministic exp decomp}.
Recall that \cref{lem:exp decomp} is obtained by running the randomized non-stop cut-matching game where the matching player (\cref{lemma:matching-player}) runs the weighted push-relabel algorithm on the shortcut graphs.
Note that the matching player obtains the matching by decomposing the flow into paths, and standard path-decomposition algorithm (e.g., \cref{lemma:path-decomposition}) guarantees that each path it found either saturates a source/sink vertex or saturates an edge.
Since we put all infinite-capacity edges in the bottom level of the hierarchy, our algorithms are dealing with finite demand, and therefore we only saturate finite-capacity edges.
This implies that the size of the matching is in fact bounded by $n + |E(G_A)\setminus E_{\infty}|$ (which in our case is $O(n\log n)$; $O(n)$ from the original graph, and $O(n\log n)$ finite-capacity star edges) instead of the cruder bound of $m^\prime = |E(G_A)|$ in \cref{lemma:matching-player}. 

\begin{lemma}[Matching Player for Vertex-Capacitated Graphs]
  Let $F \subseteq E$ and $\cH$ be a hierarchy of $G \setminus F$.
  Let $A$ be the shortcut induced by $\cH$ with capacity scale $\psi \in 1/\N$ and $G_A \defeq G \cup A$ be the shortcut graph that contains $m^\prime$ edges.
  Then, for any $\phi < 1$ and $\delta < 1$, given a partition $(P, Q)$ of $V$ and $\psi$-integral $\Bd^\prime \leq \Bd$ for $\Bd \defeq \Bvol_F$, there is a deterministic algorithm that runs in $O(n^2\log^4 n \cdot (1/\phi+1/\psi))$ time and computes either
  \begin{itemize}
    \item a balanced $\phi$-sparse cut $S$ such that $\vol_F(S) \geq \delta/2 \cdot \vol_F(V)$, or
    \item a $\psi$-integral $(P,Q,\Bd^\prime,\Delta,n+|E(G_A)\setminus E_{\infty}|)$-bidirectional-matching $(\overrightarrow{M},\overleftarrow{M})$ where $\Delta \defeq \delta \cdot \vol_F(V)$.
  \end{itemize}

  In the second case, the matchings are such that, given any $\psi$-integral flow $\Bf_M$ on $M \defeq \overrightarrow{M} \cup \overleftarrow{M}$ of congestion $\kappa$,
  there exists a $\psi$-integral flow $\Bf_{G_A}$ routing the same demand as $\Bf$ does in $G_A$ with congestion $O(\kappa \cdot \log n/\phi)$ such that $\frac{\Bw_{\cH}(\Bf_{G_A})}{|\Bf_{G_A}|} \leq h$ for some $h = O(n\log n\cdot (1/\phi+1/\psi))$.
  
  \label{lemma:sparse-matching-player}
\end{lemma}

Our derandomization comes from substituting the randomized non-stop KRV-style cut-matching game with a deterministic KKOV-style variant.
Note that we do not prove a non-stop version as in \cref{sec:decomp}; Instead, we require the matching player to have $\delta = 0$, or in other words that it always returns a perfect matching.

\begin{lemma}[Deterministic Cut-Matching Game~\cite{BernsteinGS20}]
  There exists a universal constant $c_{\ref{lemma:deterministic-cmg}}$ such that the following holds.
  
  Consider an integral vertex measure $\Bd$ on $n$ vertices $V$ where $\|\Bd\|_{\infty}\leq\poly(n)$.
  Let $z\in \N$.
  Suppose there is an oracle that on input partition $(P,Q)$ and $1/z$-integral $\Bd^\prime\leq \Bd$ outputs a $1/z$-integral $(P,Q,\Bd^\prime,\Delta,m)$-bidirectional-matching for $\Delta \defeq c_{\ref{lemma:deterministic-cmg}}/\log^4 n \cdot \Bd(V)$.
  
  Then, there is a \textbf{deterministic} algorithm that invokes the oracle $T_{\KKOV} = O(\log n)$ times and takes $O(m^2\log^9 n)$ additional running time to compute a submeasure $\tilde{\Bd} \leq \Bd$ with $\|\tilde{\Bd}\|_1 \geq \frac{7}{8}\|\Bd\|_1$ such that $\tilde{\Bd}_1$ is $O(\log^{15} n)$-hop $\Omega(1/\log^{16} n)$-expanding in $W \defeq \bigcup_{i \in [k]} \forward{M}_i \cup \backward{M}_i$, where $\forward{M}_i \cup \backward{M}_i$ is the output of the oracle in the $i$-th invocation.
  \label{lemma:deterministic-cmg}
\end{lemma}

\begin{proof}[Proof of \cref{lem:deterministic exp decomp}]
  The proof remains largely the same as that of \cref{lem:exp decomp} and thus we only describe the difference.
  Recall that the strategy is to use cut-matching games to either certify that a large fraction of the terminal edges are expanding (which we add to $E_{r}^{(r+1)}$), or find a balanced sparse cut for us to recurse on.

  We replace the randomized non-stop cut-matching game \cref{lemma:non-stop-cut-matching} with \cref{lemma:deterministic-cmg}, using \cref{lemma:sparse-matching-player} with $\phi_{\mathrm{exp}} = O(1/\log^5 n)$ and $\delta = O(1/\log^4 n)$, both sufficiently small, as the oracle.
  If \cref{lemma:sparse-matching-player} ever returns a balanced sparse cut, then we recurse on both sides of the cut.
  Otherwise, we use the output of \cref{lemma:sparse-matching-player} as the matching that is fed to \cref{lemma:deterministic-cmg}. 
  In $T_{\KKOV} = O(\log n)$ rounds, the cut-matching game \cref{lemma:deterministic-cmg} terminates with a $\tilde{\Bd}$ that is $O(\log^{15} n)$-hop $\Omega(1/\log^{16} n)$-expanding in $W$.
  Using the same reasoning as in the proof of \cref{lem:exp decomp}, we can extract a large edge set that is $(\Bw_{\cH}, O(h\log^{15} n))$-length $\Omega(\phi_{\mathrm{exp}}/\log^{18} n)$-expanding in $G_A$ for $h = O(n\log n \cdot (\log^5 n + 1/\psi))$.

  The correctness of the algorithm is the same as in \cref{lem:exp decomp}.
  Since the size of the subgraph decreases by a $(1-\Omega(1/\log^4 n))$ fraction in each recursion, the recursion depth is bounded by $O(\log^5 n)$ and for $\phi_{\mathrm{exp}}$ sufficiently small the capacity of $E_{r+1}^{(r+1)}$ is bounded by $0.9 \cdot \Bc(E_r^{(r)})$ as desired.
  The construction takes $O(\log^5 n) \cdot O(\log n) \cdot (O(n^2\log^{11} n) + O(n^2\log^4 n \cdot (\log^5 n + 1/\psi))) = O(n\log^{10} n \cdot (\log^7 n + 1/\psi))$ time.
  Given any $(\Bvol_{E_r^{(r+1)}}\cdot \kappa/z)$-respecting $1/z$-integral demand, we can route it using \cref{lemma:routing-data-structure} with a flow $\Bf$ with $\Bf(e) \leq \frac{\lceil \Bc_{G_A}(e)\cdot\kappa/\phi_{\det} \rceil}{z}$, for some $\phi_{\det} \defeq \Theta(1/\log^{24} n)$, in $O(n^2\log^{19} n\cdot(\log^5 n + 1/\psi))$ time.
  This completes the proof.
\end{proof}

This finishes the description of our deterministic vertex-capacitated max-flow algorithm, modulo the proof of \cref{lemma:deterministic-cmg} which we do in the next section.

\subsection{Deterministic Cut-Matching Game}\label{sec:deterministic-cmg}

In this section we show how the directed version of the deterministic cut-matching game (i.e., \cref{lemma:deterministic-cmg}) can be proved.
The cut-matching game works as follows:
Let $W_{\mathrm{all}}$ initially contain only self loops on the vertices so that $\deg_{W_{\mathrm{all}}}(v) = 2\Bd(v)$ for all $v$.
We maintain the invariant that $\deg_{W_{\mathrm{all}}}(v) = 2k \cdot \Bd(v)$ for all $v$ at the start of the $k$-th iteration, and that $W_{\mathrm{all}} = W \cup F$ for some \emph{real} edges $W$ and \emph{fake} edges $F$, with $\Bc(F) \leq c \cdot \Bc(W_{\mathrm{all}})/\log^4 n$ for some sufficiently small constant $c$ (depending on the value of $c_{\ref{lemma:weighted-balanced-sparse-cut}}$ below in \cref{lemma:weighted-balanced-sparse-cut}).
In each iteration, we run the following \cref{lemma:weighted-balanced-sparse-cut} on $W_{\mathrm{all}}$ (with $F$ being the set of fake edges).
If for some iteration, \cref{lemma:weighted-balanced-sparse-cut} concludes that for some $X\subseteq V$ with $\vol_{W_{\mathrm{all}}}(X) \geq 0.88 \cdot \vol_{W_{\mathrm{all}}}(V)$ we have that $W_{\mathrm{all}}[X] \setminus F = W[X]$ is $\Omega(1/\log^{14} n)$-expander, then we show that $\tilde{\Bd} \defeq \Bd|_X$ satisfies the output requirement of \cref{lemma:deterministic-cmg}.
Indeed, we have that $\|\tilde{\Bd}\|_1 = \vol_{W_{\mathrm{all}}}(X) / 2k \geq 0.88 \cdot \vol_{W_{\mathrm{all}}}(V) / 2k \geq \frac{7}{8} \|\Bd\|_1$.
To see the expansion guarantee of $\tilde{\Bd}$, note that $\deg_{W[X]}(v) \geq \tilde{\Bd}(v)$ for all $v$.
Combining the following claim and that $W[X]$ is an $\Omega(1/\log^{14} n)$-expander, we get that $\tilde{\Bd}$ is $O(\log^{15} n)$-hop $\Omega(1/\log^{16} n)$-expanding in $W$.

\begin{fact}[{see, e.g.,~\cite[Lemma 5.10]{BernsteinBST24}}]
  Let $W$ be an $n$-vertex $\phi$-expander.
  Then, $\Bvol_W$ is $O(\log n/\phi)$-hop $\Omega(\phi/\log^2 n)$-expanding in $W$.
\end{fact}

If, on the other hand, \cref{lemma:weighted-balanced-sparse-cut} finds a cut $(S, \overline{S})$, we use the matching player to compute a matching between $S$ and $\overline{S}$ which is then added to $W_{\mathrm{all}}$.\footnote{Technically, the bipartition is made perfectly balanced by moving vertices in the larger side to the smaller side arbitrarily. See \cite[Section 7.1]{BernsteinGS20} for more details.}
Note crucially that while the matchings the oracle outputs are not necessarily perfect, we add a (fairly arbitrary) set of fake edges with capacity bounded by $2\Delta = 2c_{\ref{lemma:deterministic-cmg}}/\log^4 n \cdot \Bd(V)$ to make them perfect before including them in $W_{\mathrm{all}}$.
With $c_{\ref{lemma:deterministic-cmg}}$ sufficiently small, the edge set $F$ always has capacity bounded by $c_{\ref{lemma:weighted-balanced-sparse-cut}} \cdot U/\log^4 n$, where $U \defeq \sum_{e \in W_{\mathrm{all}}}\Bc(e)$, which ensures that \cref{lemma:weighted-balanced-sparse-cut} is always applicable in the next iteration.

It then remains to bound the number of iterations this process takes.
Note that in $W_{\mathrm{all}}$, we always add a perfect matching between a balanced sparse cut $(S,\overline{S})$.
This was first analyzed by \cite{KhandekarKOV07} for unweighted undirected graphs, and \cite{BernsteinGS20} used the same entropy analysis to show that this works for unweighted directed graphs.
The same proof strategy generalizes to the weighted case as well.
The argument is essentially that whenever there is a cut as in \cref{lemma:weighted-balanced-sparse-cut}(\labelcref{item:cmg-cut}), then some entropic measure of the flow matrix increases by $\Omega(n)$ after adding the matchings to the witness.
Using the fact that the measure is always upper-bounded by $O(n\log n)$, this shows that we must arrive at \cref{lemma:weighted-balanced-sparse-cut}(\labelcref{item:cmg-expanding}) in $T_{\KKOV} = O(\log n)$ rounds.

Overall, the running time of \cref{lemma:deterministic-cmg} is $O(m^2\log^9 n)$ where $m$ is the size of the matchings, since the witness contains at most $O(m\log n)$ edges.
We omit the details on the entropy analysis and focus on proving the balanced sparse cut algorithm.

\begin{lemma}
  There exists a universal constant $c_{\ref{lemma:weighted-balanced-sparse-cut}} > 0$ such that,
  given an $n$-vertex $m$-edge capacitated graph $G$ and an edge set $F\subseteq E(G)$ with $\Bc_G(F) \leq c_{\ref{lemma:weighted-balanced-sparse-cut}} \cdot U/\log^4 n$ where $U \defeq \sum_{e\in E}\Bc_G(e)$,
  there is a \textbf{deterministic} $O((m+n)^2\log^7 n)$ time algorithm that either
  \begin{enumerate}
    \item\label{item:cmg-expanding} outputs a vertex set $X$ of $\vol_G(X) \geq 0.88 \cdot \vol_G(V)$ such that $G[X] \setminus F$ is an $\Omega(1/\log^{14} n)$-expander, or
  \item\label{item:cmg-cut} a cut $S \subseteq V(G)$ with $\vol_G(S),\vol_G(\overline{S}) \geq \vol_G(V)/50$ and $\min\{\Bc_G(E_G(S,\overline{S})), \Bc_G(E_G(\overline{S},S))\} \leq \vol_G(V) / 400$.
  \end{enumerate}
  \label{lemma:weighted-balanced-sparse-cut}
\end{lemma}

To prove \cref{lemma:weighted-balanced-sparse-cut}, we are going to reduce to the case where the graph is unweighted and has constant degree and then apply the following algorithm from \cite{ChuzhoyGLNPS20}.
We would like to point out that the following \cref{lemma:unweighted-balanced-sparse-cut-constant} does not follow exactly from \cite{ChuzhoyGLNPS20} because their algorithm is for undirected.
We give a simplified proof in \cref{sec:directed-bal-cut} that is specialized to the specific form that we are using.

\begin{lemma}
  There is a universal constant $c_{\ref{lemma:unweighted-balanced-sparse-cut-constant}} > 0$ such that, given an $n$-vertex unit-capacitated graph $G$ with maximum degree $\Delta_G = O(1)$, an edge set $F\subseteq E(G)$ with $|F| < c_{\ref{lemma:unweighted-balanced-sparse-cut-constant}}\cdot n / \log^4 n$, and a constant $\psi < 1$, there is an $O(n^2\log^7 n)$ time algorithm that either
  \begin{enumerate}
    \item\label{item:fake-edges} outputs a vertex set $X\subseteq V$ of size $|X| \geq 0.89 \cdot n$ such that $G[X] \setminus F$ is an $\Omega(1/\log^{14} n)$-expander, or
    \item\label{item:cut} outputs a cut $S \subseteq V(G)$ with $|S|,|\overline{S}| \geq n/20$ such that $\min\{|E_G(S,\overline{S})|, |E_G(\overline{S},S)|\} \leq \psi \cdot \vol_G(V)$.
  \end{enumerate}
  \label{lemma:unweighted-balanced-sparse-cut-constant}
\end{lemma}

\paragraph{Reducing to Unweighted Graphs.}
We first transform $G$ into a unit-capacitated graph $G^\prime$ by throwing away low-capacity edges and replacing each remaining edge with a proportional number of parallel unit-capacitated edges.
This uses ideas from \cite{vdBCKLMPS24}.
Let $U\defeq\sum_{e\in E}\Bc_G(e)$ be the sum of edge capacities, and let $\tau \defeq U/(1000m)$ be a threshold.
For each edge $e$ with $\Bc_G(e)\geq \tau$, we add $\lceil \Bc_G(e)/\tau \rceil$ parallel copies of $e$ into $G^\prime$.
For each vertex $v$, we also add $\lceil \deg_G(v)/(20\tau) \rceil$ self-loops on $v$ to $G^\prime$ (each contributing $2$ units of degree).
Edges with small capacities are discarded.
Let $m^\prime$ be the number of edges in $G^\prime$, and let $F^\prime$ be edges in $G^\prime$ that correspond to $F$.

\begin{claim}
  It holds that $1099m \leq m^\prime \leq 1102m$.
\end{claim}

\begin{proof}
We first lower bound $m^\prime$.
Edges with $\Bc_G(e)<\tau$ have a total capacities of at most $U/1000$, and for edges with $\Bc_G(e)\geq \tau$ we have $\lceil \Bc_G(e)/\tau \rceil \geq \Bc_G(e)/\tau$ copies of edges in $G^\prime$.
This contributes $\frac{U-U/1000}{\tau} = 999m$ edges.
Likewise, we have $\lceil \deg_G(v)/(20\tau) \rceil \geq \deg_G(v)/(20\tau)$, and since the degrees sum up to $2U$, this contributes at least $100m$ edges.

For the upper bound, we have $\lceil \Bc_G(e)/\tau\rceil \leq \Bc_G(e)/\tau + 1$, and thus this contributes at most $U/\tau + m = 1000m + m = 1001m$ edges.
Likewise, the degrees contribute at most $2U/(20\tau) + n \leq 101m$ edges.
\end{proof}

\begin{claim}
  Let $X\subseteq V$.
  If $G^\prime[X] \setminus F^\prime$ is a $\phi$-expander, then $G[X] \setminus F$ is also a $\phi/20$-expander.
  \label{claim:weighted-to-unweighted}
\end{claim}

\begin{proof}
  Note that the self loops in $G^\prime$ contributes a degree of at least $2 \cdot \deg_G(v)/(20\tau) = \deg_G(v)/(10\tau)$ to each $v$, which is not affected by the removal of $F^\prime$ and the restriction to $X$.
  In particular, this suggests that $\vol_{G^\prime[X] \setminus F^\prime}(S) \geq \vol_{G[X] \setminus F}(S)/(10\tau)$ for all $S\subseteq X$.
  Let $e_{(u,v)}$ denote the number of parallel edges $(u, v)$ in $G^\prime\setminus F^\prime$.
  We know that if $e_{(u,v)}>0$, then $\Bc_G(u,v) \geq e_{(u,v)} \cdot \tau/2$.
  Overall, this suggests that
  \[ \frac{\Bc(E_{G[X]\setminus F}(A,B))}{\min\{\vol_{G[X]\setminus F}(A),\vol_{G[X]\setminus F}(B)\}} \geq \frac{1}{20} \cdot \frac{|E_{G^\prime[X]\setminus F^\prime}(A,B)|}{\min\{\vol_{G^\prime[X] \setminus F^\prime}(A),\vol_{G^\prime[X] \setminus F^\prime}(B)\}} \]
  for all $A, B\subseteq X$,
  and thus if $G^\prime[X] \setminus F^\prime$ is a $\phi$-expander then $G[X] \setminus F$ is a $\phi/20$-expander.
\end{proof}

\begin{claim}
  For any cut $S\subseteq V(G)$ we have $\vol_G(S)\geq \frac{10}{11}\tau \cdot \vol_{G^\prime}(S) - U/250$ and $\Bc_G(E_G(S,\overline{S})) \leq \tau \cdot |E_{G^\prime}(S,\overline{S})| + U/1000$.
  \label{claim:cut-weighted-to-unweighted}
\end{claim}

\begin{proof}
  We first bound $\vol_G(S)$.
  For a vertex $v$ we have $\deg_{G^\prime}(v) \leq 2 \cdot \lceil \deg_G(v)/(20\tau)\rceil + \deg_G(v)/\tau + e_v \leq 11/10 \cdot \deg_G(v)/\tau + e_v + 2$, where $e_v$ is the number of edges incident to $v$ (the first term in the inequality comes from the self-loops we added; the second comes from edges incident to $v$).
  The bound follows since $e_v$ sums up to at most $2m$.
  Likewise, for each edge $e$ we have $\lceil \Bc_G(e)/\tau\rceil \leq \Bc_G(e)/\tau+1$ edges in $G^\prime$.
  The $+1$ terms amount to $\tau \cdot m = U/1000$ units of capacities.
\end{proof}

\paragraph{Reducing to Low-Degree Graphs.}
We apply a second transformation that reduces the maximum degree of $G^\prime$ to a constant and obtain $G^{\prime\prime}$, following the ideas in \cite{ChuzhoyGLNPS20}.
For each vertex $v$ we use \cite[Theorem 2.4]{ChuzhoyGLNPS20} to deterministically construct an $\Omega(1)$-expander $H_v$ with maximum degree $9$ containing $\deg_{G^\prime}(v)$ vertices and replace $v$ by $H_v$ in $G^{\prime\prime}$.
Let us number these vertices by $U_v \defeq \{v_1,\ldots,v_{\deg_{G^\prime}(v)}\}$.
We also number the edges incident to $v$ arbitrarily by $\{e_1,\ldots,e_{\deg_{G^\prime}(v)}\}$.
Now, for each edge $e = (u,v)$, if $e$ is the $i$-edge incident to $u$ and the $j$-th edge incident to $v$, then we add $(u_i, v_j)$ to $G^{\prime\prime}$.
Note that since $G^\prime$ might contain self-loops from the first transformation, we treat each self-loop on $v$ as two edges in the list $\{e_i\}$ of $v$ (since they contribute two units of degree each) and thus map it to an edge $(v_i, v_j)$ for some $i$ and $j$.
Now $G^{\prime\prime}$ has $n^{\prime\prime}$ vertices and $m^{\prime\prime}$ edges where $n^{\prime\prime} = 2m^\prime$ and $m^{\prime\prime} \leq 9n^{\prime\prime}+m^\prime\leq 20m^\prime$, and has maximum degree $10$.
Again, let $F^{\prime\prime}$ denote edges in $G^{\prime\prime}$ that correspond to $F^\prime$ (and thus $F$).

\begin{proof}[Proof of \cref{lemma:weighted-balanced-sparse-cut}]
  We do the two transformations and apply \cref{lemma:unweighted-balanced-sparse-cut-constant} with the constant
  $\psi$ sufficiently small
  on $G^{\prime\prime}$ (with the edge set $F^{\prime\prime}$) that runs in $O({n^{\prime\prime}}^2\log^7 n^{\prime\prime}) = O((m+n)^2\log^7 n)$ time.
  Indeed, note that if $\Bc_G(F) \leq c_{\ref{lemma:weighted-balanced-sparse-cut}} \cdot U/\log^4 n$ for some constant $c > 0$ (as required by \cref{lemma:weighted-balanced-sparse-cut}), then $|F^{\prime\prime}| \leq 2\Bc_G(F)/\tau \leq 2000c \cdot m/\log^4 n$ which is less than $c_{\ref{lemma:unweighted-balanced-sparse-cut-constant}} \cdot n^{\prime\prime}/\log^4 n^{\prime\prime}$ for $c_{\ref{lemma:weighted-balanced-sparse-cut}}$ being sufficiently small, depending on $c_{\ref{lemma:unweighted-balanced-sparse-cut-constant}}$.
  Consider the following two cases that \cref{lemma:unweighted-balanced-sparse-cut-constant} returns.
  \begin{itemize}
    \item If \cref{lemma:unweighted-balanced-sparse-cut-constant} terminates in Case~\ref{item:fake-edges}, i.e., it outputs 
    a vertex set $X^{\prime\prime}$ of size $|X^{\prime\prime}| \geq 0.89 n^{\prime\prime}$ so that $G^{\prime\prime}\{X^{\prime\prime}\} \setminus F^{\prime\prime}$ is an $\Omega(1/\log^{14} n)$-expander.
      Then, we first map $X^{\prime\prime}$ to an $X^\prime\subseteq V(G)$ by including all $v \in V$ with $|U_v \cap X^{\prime\prime}| \geq |U_v|/100$.
      We claim that this new set $X^\prime$ in $G^\prime$ is an expander.

      \begin{claim}
        It holds that $G^\prime\{X^\prime\} \setminus F^\prime$ is an $\Omega(1/\log^{14} n)$-expander.
      \end{claim}
      \begin{proof}
      Consider any $\Bvol_{G^\prime\{X^\prime\}\setminus F^\prime}$-respecting demand $(\Bsource,\Bsink)$.
      We may treat $(\Bsource,\Bsink)$ as a demand on $G^{\prime\prime}$ by distributing the demand on $v$ evenly to $U_v$.
      Since $H_v$ is an $\Omega(1)$-expander, we can route the demand on $U_v\setminus X^{\prime\prime}$ to $U_v \cap X^{\prime\prime}$ by a flow in $H_v$ with congestion $O(1)$.
      We then route the resulting demand, which is $\Bvol_{G^{\prime\prime}\{X^{\prime\prime}\}\setminus F^{\prime\prime}} \cdot O(1)$-respecting, in $G^{\prime\prime}\{X^{\prime\prime}\}\setminus F^{\prime\prime}$ with congestion $O(\log^{14} n)$.
      Such a flow naturally corresponds to a flow in $G^\prime\{X^\prime\} \setminus F^\prime$ with congestion $O(\log^{14} n)$.
      Combining this with the flow in $H_v$ proves the claim.
      \end{proof}

      Now it follows from \cref{claim:weighted-to-unweighted} that $G\{X^\prime\} \setminus F$ is an $\Omega(1/\log^{14} n)$-expander.
      The volume of $X^\prime$ in $G$ can also be bounded as $\vol_G(X^\prime) \geq  \frac{10}{11}\tau \cdot \vol_{G^\prime}(X^\prime)-U/250$ by \cref{claim:cut-weighted-to-unweighted}, where
      \[ \vol_{G^\prime}(X^\prime)=2m^\prime-\sum_{v\in V: |U_v\cap X^{\prime\prime}|<|U_v|/100}|U_v|\geq 2m^\prime-n^{\prime\prime}/9 = 16m^\prime/9. \]
      The inequalities above follow since $\sum_{|U_v \cap X^{\prime\prime}| < |U_v|/100}|U_v| \leq \frac{100}{99} \cdot \sum|U_v\setminus X^{\prime\prime}| \leq \frac{100}{99} \cdot 0.11n^{\prime\prime}$.
      This suggests that $\vol_G(X^\prime)\geq \frac{10}{11}\tau \cdot (\frac{16}{9} \cdot 1099m) - U/250\geq 0.88 \vol_G(V)$ (note that $\vol_G(V) = 2U$), as desired.

    \item On the other hand, if \cref{lemma:unweighted-balanced-sparse-cut-constant} terminates in Case~\ref{item:cut}, i.e., it outputs a cut $S^{\prime\prime}\subseteq V(G^{\prime\prime})$ where $|S^{\prime\prime}|\geq n^{\prime\prime}/20$ and $\min\{|E_{G^{\prime\prime}}(S^{\prime\prime},\overline{S^{\prime\prime}})|, |E_{G^{\prime\prime}}(\overline{S^{\prime\prime}}, S^{\prime\prime})|\} \leq \psi \cdot n^{\prime\prime}$.
    Suppose that our constant-expanders $H_v$ have expansion $\alpha_0$ (as in \cite[Theorem 2.4]{ChuzhoyGLNPS20}).
    If $\psi \leq \alpha_0/2$ (this gives the bound on how small $\psi$ should be), then by \cite[Lemma 5.4]{ChuzhoyGLNPS20} we can compute another cut $T^{\prime\prime} \subseteq V(G^{\prime\prime})$ such that $|T^{\prime\prime}| \geq |S^{\prime\prime}|/2$, $|\overline{T^{\prime\prime}}| \geq |\overline{S^{\prime\prime}}|/2$, and $\min\{|E_{G^{\prime\prime}}(T^{\prime\prime}, \overline{T^{\prime\prime}})|, |E_{G^{\prime\prime}}(\overline{T^{\prime\prime}}, T^{\prime\prime})|\}\leq O(\min\{|E_{G^{\prime\prime}}(S^{\prime\prime}, \overline{S^{\prime\prime}})|, |E_{G^{\prime\prime}}(\overline{S^{\prime\prime}}, S^{\prime\prime})|\})$.
    More importantly, for each $v \in V(G)$, either $U_v \subseteq T^{\prime\prime}$ or $U_v \cap T^{\prime\prime} = \emptyset$, i.e., the cut $T^{\prime\prime}$ is consistent with the original vertices.
    We can thus naturally map $T^{\prime\prime}$ to a cut $S^\prime \subseteq V(G^\prime)$ such that $\vol_{G^\prime}(S^\prime) = |T^{\prime\prime}|$ and $\min\{|E_{G^{\prime}}(S^{\prime},\overline{S^{\prime}})|, |E_{G^{\prime}}(\overline{S^{\prime}}, S^{\prime})|\} \leq O(\psi) \cdot n^{\prime\prime} = O(\psi) \cdot m^{\prime\prime}$.
    By \cref{claim:cut-weighted-to-unweighted},  the cut $S^\prime$ in $G$ has $\vol_G(S^\prime)\geq \frac{10}{11}\tau \cdot \vol_{G^\prime}(S^\prime) - U/250 \geq U/25 = \vol_G(V)/50$.
    The same bound holds for $\vol_{G}(\overline{S})$ as well.
    Finally, we have $\min\{\Bc_G(E_G(S^\prime,\overline{S^\prime})), \Bc_G(E_G(\overline{S^\prime}, S^\prime))\} \leq O(\psi) \cdot U + U/1000 \leq \vol_G(V)/400$ for $\psi$ sufficiently small, as desired.
  \end{itemize}
  This concludes the proof.
\end{proof}

\subsection{Proving \texorpdfstring{\cref{lemma:unweighted-balanced-sparse-cut-constant}}{Lemma A.10}}\label{sec:directed-bal-cut}

We will reduce \cref{lemma:unweighted-balanced-sparse-cut-constant} to the following similar lemma.

\begin{lemma}
  There exists a universal constant $c_{\ref{lemma:unweighted-balanced-sparse-cut}} > 0$ such that, given an $n$-vertex unit-capacitated graph $G$ with maximum degree $\Delta_G = O(1)$, an edge set $F \subseteq E(G)$ with $|F| < c_{\ref{lemma:unweighted-balanced-sparse-cut}} \cdot n/\log^4 n$, and a constant $\psi < 1$, there is a deterministic $O(n^2 \log^2 n)$ time algorithm that either
  \begin{enumerate}
    \item output $X \subseteq V(G)$ such that $G[X]\setminus F$ is an $\Omega(1/\log^{14} n)$-expander, or
    \item outputs a cut $S \subseteq V(G)$ with $|\overline{S}|\geq |S| \geq \Omega(n/\log^4 n)$ such that $\min\{|E_G(S,\overline{S})|,|E_G(\overline{S},S)|\}<\psi \cdot |S|$.
  \end{enumerate}
  \label{lemma:unweighted-balanced-sparse-cut}
\end{lemma}

We first prove \cref{lemma:unweighted-balanced-sparse-cut-constant} assuming \cref{lemma:unweighted-balanced-sparse-cut}.

\begin{proof}[Proof of \cref{lemma:unweighted-balanced-sparse-cut-constant}]
  We start with $S \defeq \emptyset$ being initially empty.
  While $|S| < n/10$, we invoke \cref{lemma:unweighted-balanced-sparse-cut} on $G[V\setminus S]$ with the given value of $\psi$.
  If we get an $X\subseteq V\setminus S$ for which $G[X]\setminus F$ is an $\Omega(1/\log^{14} n)$-expander, then we return the same $X$.
  The size of $X$ is at least $0.99 \cdot 0.9 n \geq 0.89n$.
  Otherwise, we get a cut $S^\prime\subseteq V\setminus S$, and we include $S^\prime$ into $S$, i.e., set $S \gets S \cup S^\prime$.

  Now, suppose that we have reached a state where $|S|\geq n/10$ but none of the invocations of \cref{lemma:unweighted-balanced-sparse-cut} returns an expander subgraph $X$.
  Then, since the set $S^\prime$ returned by \cref{lemma:unweighted-balanced-sparse-cut} has $|S^\prime|\leq n/2$, we have $|S| \leq 3n/4$.
   Now we show that we can extract a sparse cut within $S$.
   Indeed, this follows from the following observation from prior work (e.g., \cite{BernsteinGS20,BernsteinBST24}) that a union of sparse cuts contains a large sparse cut.\footnote{This is in contrast to the undirecte case where the union of sparse cuts \emph{is} a sparse cut itself.}

   \begin{claim}
     Consider a graph $W$ and suppose there is a sequence of cuts $S_1,\ldots, S_k$ where $S_i$ is in $W_i \defeq W\setminus (S_1 \cup \cdots \cup S_{i-1})$ and $|S_i| \leq |V(W_i)|/2$.
     Then, there is an $O(|E(W)|)$ time algorithm that computes an $S\subseteq S_1\cup \cdots \cup S_k$ with $|S|\geq (|S_1|+\cdots+|S_k|)/2$ and $\min\{|E_W(S,\overline{S})|, |E_W(\overline{S}, S)|\} < \sum_{i}\min\{|E_{W_i}(S_i,\overline{S_i})|, |E_{W_i}(S_i,\overline{S_i})|\}$.
     \label{claim:extract-sparse-cut}
   \end{claim}

   By \cref{claim:extract-sparse-cut}, we can compute a sparse cut $S^\prime\subseteq S$ of size $|S^\prime|\geq |S|/2\geq n/20$ such that $\min\{|E_{G}(S^\prime,\overline{S^\prime})|, |E_{G}(\overline{S^\prime}, S^\prime)| < 2\psi \cdot |S^\prime| \leq \psi \cdot n$.
   This $S^\prime$ can thus be used as the output.
   The running time is $O(\log^5 n) \cdot O(n^2\log^2n)$ since the size of the graph decreases by a $1-\Omega(1/\log^4n)$ factor for each invocation of \cref{lemma:unweighted-balanced-sparse-cut}.
   This concludes the proof.
\end{proof}

We now show \cref{lemma:unweighted-balanced-sparse-cut} whose algorithm works as follows.
We deterministically construct an $\Omega(1)$-expander $H$ on $n$ vertices via \cite[Theorem 2.4]{ChuzhoyGLNPS20}, and then attempt to embed $H$ into $G$.
Note that $H$ can be decomposed into $k = O(1)$ disjoint matchings and we embed each of them separately.
This is done via the following algorithm in \cref{lemma:route-or-cut} which is a simplified version of \cite[Theorem 3.2]{ChuzhoyGLNPS20} specialized to our use case.

\begin{lemma}
  Given an $n$-vertex graph $G$ with maximum degree $\Delta_G = O(1)$ and a matching $M$ where $V(M) \subseteq V(G)$, a $z \in \N$ where $z < n/4$, and a constant $\psi < 1$, there is a deterministic $O(n^2\log^2 n)$ time algorithm that either
  \begin{itemize}
    \item outputs a subset of $Z \subseteq M$ of size $|Z| \leq z$ such that $M$ can be embedded into $G \cup Z$ with congestion $O(\log^2 n)$, or
    \item outputs a cut $S\subseteq V(G)$ with $|S|,|\overline{S}|\geq z/6$ such that $\Psi_G(S)\leq \psi$, where $\Psi_G(S)\defeq \frac{\min\{|E_G(S,\overline{S})|,|E_G(\overline{S}, S)|\}}{\min\{|S|,|\overline{S}|\}}$.
  \end{itemize}
  \label{lemma:route-or-cut}
\end{lemma}

If \cref{lemma:route-or-cut} with $z\defeq \frac{dn/\log^4 n}{k}$ with a small constant $d > 0$ returns a cut when embedding one of the matchings, then \cref{lemma:unweighted-balanced-sparse-cut} returns the same cut.
Otherwise, overall, that the $\Omega(1)$-expander $H$ is embeddable into $G \cup Z$ with congestion $O(\log^2 n)$ means that $G \cup Z$ is an $\Omega(1/\log^2 n)$-expander.
Here $Z \defeq \bigcup_{i \in [k]}Z_i$ where $Z_i$ is the set of fake edges outputted by \cref{lemma:route-or-cut} when embedding the $i$-th matching.
Since $|F \cup Z| \leq (c+d)n/\log^4 n$, we can apply the following expander pruning algorithm to extract an $X\subseteq V(G)$ with $|X| \geq n - |V(G)\setminus X| \geq n - \vol_{G\cup Z}(V(G)\setminus X) \geq 0.99 n$ for $c$ and $d$ sufficiently small such that $G[X]\setminus F$ is an $\Omega(\log^{14} n)$-expander.

\begin{lemma}[{\cite[Theorem 4.1]{SulserP25}}]
  Given an $m$-edge $\phi$-expander $G$ and $D\subseteq E(G)$ with $|D| < \frac{\phi^2}{200} m$, there is an $O(|D|\log m/\phi^4)$ time algorithm that outputs an $X\subseteq V(G)$ with $\vol_G(V(G)\setminus X) \leq \frac{8}{\psi} |D|$ such that $(G\setminus D)[X]$ is an $\Omega(\phi^7)$-expander.
\end{lemma}

It thus remains to prove \cref{lemma:route-or-cut}.

\begin{proof}[Proof of \cref{lemma:route-or-cut}]
  The algorithm works in phases, where in each phase we try to embed a subgraph $M^\prime\subseteq M$ into $G$.
  We start with $M^\prime \defeq M$ and stop running the phases when $|M^\prime| < z$.
  For each phase, we let $G^\prime$ be a copy of $G$ and go through the edges $(u, v) \in M^\prime$ in an arbitrary order.
  Let $M_{\mathrm{embed}}$ initially empty be the set of edges in $M^\prime$ that we embed in this phase.
  We check in $O(m)$ time if there exists an $(u, v)$-path $P$ in $G^\prime$ with length at most $\ell \defeq \frac{24\Delta_G\log n}{\psi} = O(\log n)$.
  If there is, then we embed $(u, v) \in M^\prime$ into the path $P$, add $(u, v)$ to $M_{\mathrm{embed}}$, and remove $P$ from $G^\prime$.
  If we have successfully embedded at least $\frac{\psi}{12\ell}$ fraction of the edges in $M^\prime$, i.e., $|M_{\mathrm{embed}}| \geq \frac{\psi}{12\ell} \cdot |M^\prime|$, then we proceed to the next phase by setting $M^\prime \gets M^\prime \setminus M_{\mathrm{embed}}$.
  Once we reach the state where $|M^\prime| < z$, we know that $M \setminus M^\prime$ can be embedded into $G$ with congestion $O(\ell \log n) = O(\log^2 n)$ (since there are $O(\ell \log n)$ phases), and thus we set $F \defeq E(M^\prime)$ which makes $M$ trivially embeddable into $G\cup F$ with congestion $O(\log^2 n)$.

  Otherwise, if for one of the phases we did not successfully embed $\frac{\psi}{12\ell}$ fraction of the edges in the current $M^\prime$, then we will find a sparse cut in $G$.
  We first prove the following claim.
  \begin{claim}
    Given $u$ and $v$ in $\tilde{G}\subseteq G$ such that $\dist_{\tilde{G}}(u,v)>\ell$, there is an $O(m)$ time algorithm that finds a cut $S\subseteq V(\tilde{G})$ such that $|S| \leq n/2$ and $\min\{|E_{\tilde{G}}(S,\overline{S})|, |E_{\tilde{G}}(\overline{S}, S)|\} < \psi/12 \cdot |S|$.
    \label{claim:ball-growing}
  \end{claim}

  \begin{proof}
    This follows from a standard ball-growing argument.
    Let $S_i$ denote the set of vertices reachable from $u$ in $\tilde{G}$ by a path of length at most $i$.
    Assume without loss of generality that $|S_{\ell/2}|\leq n/2$, otherwise we swap $u$ and $v$ and work in the reversed $G^\prime$ instead.
    If none of the cut $S_i$ for $i < \ell/2$ is sparse, then we have $|S_{i+1}|\geq |S_i| \cdot (1+\psi/(12\Delta_G))$.
    This leads to a contradiction since $(1+\psi/(12\Delta_G))^{\ell/2} > n/2$.
  \end{proof}

   Now, we maintain a set $U \subseteq V$ which is initially empty.
   While $|U| < z/2$, since $|M^\prime \setminus M_{\mathrm{embed}}| \geq z/2$, an edge $(u, v) \in M^\prime\setminus M_{\mathrm{embed}}$ where $u, v \not\in U$ exists (since $M$ is a matching).
   We apply \cref{claim:ball-growing} and $u$ and $v$ to find a cut $S$ in $\tilde{G}\defeq G^\prime[V\setminus U]$ and set $U \gets U \cup S$.
   When this loop terminates, we know that $z/2  \leq |U| \leq z +n/2 \leq 3n/4$.
   We can then extract a balanced sparse cut from $U$ via \cref{claim:extract-sparse-cut}.
   Applying \cref{claim:extract-sparse-cut}, we get an $S\subseteq V$ with $z/2 \leq |S|\leq 3n/4$ such that $\min\{|E_{G^\prime}(S,\overline{S})|,|E_{G^\prime}(\overline{S},S)|\} < \psi/6 \cdot |S|$.
   We also know that $V\setminus S \geq n/4 \geq |S|/3$ and therefore $\Psi_{G^\prime}(S)\leq \psi/2$.
   Note that $G^\prime$ is obtained from $G$ by removing $|M^\prime \setminus M_{\mathrm{embed}}| \leq \frac{\psi}{12\ell} \cdot z$ paths of length $\ell$, and thus $\Psi_G(S) \leq \Psi_{G^\prime}(S) + \psi/12 \cdot \frac{z}{\min\{|S|,|V\setminus S|\}} \leq \psi/2+\psi/2=\psi$.
   Finally, note that the running time is $O(\ell\log n) \times O(n^2) = O(n^2\log^2 n)$.
   This completes the proof.
\end{proof}

\section{Discussion of Non-Stop Cut-Matching Game}\label{appendix:cmg}

In this section we discuss how our \cref{lemma:non-stop-cut-matching} is implied by the work of \cite{FleischmannLL25}.
Its main result, \cite[Theorem 5.1]{FleischmannLL25}, is a combination of the abstract non-stop cut-matching game we described together with a specific matching player tailored to their use case, i.e., computing directed expander decomposition of an input graph.

\paragraph{How the Cut-Matching Game Works.}
To describe the differences in presentation between \cref{lemma:non-stop-cut-matching} and \cite{FleischmannLL25}, let us first outline how this non-stop cut-matching game works.
The algorithm runs in $O(\log^2 n)$ iterations.
In each iteration $t$, it maintains for each vertex $v$ the ``fraction'' $\Bd_t(v)\leq\Bd(v)$ that is ``alive'', and let $A_t \defeq \supp(\Bd_t)$ be the set of ``alive'' vertices at iteration $t$.\footnote{This reflects the fact that in the weighted setting, a vertex might get partially matched, which is in contrast to the unweighted case where vertices are simply either matched or unmatched.}
The algorithm also maintains a multi-commodity flow matrix $\BF_t$, where $\BF_t(u,v)$ denotes the unit of $u$-commodity that is sent to $v$.
Initially, we have $A_0 = \supp(\Bd)$ and $\BF_0(u,u) = \Bd(u)$, indicating that we have $\Bd(u)$ units of $u$-commodity at $u$ at the beginning of the cut-matching game.
The goal of the cut-matching game is to have the commodities spread out among the vertices to prove expansion.

To do so, in each iteration $t$, based on $\BF_t$, the algorithm computes a random $\Bd_t$-balanced bipartition $(L,R)$ of $A_t$, i.e., $\Bd_t(L) = \Bd_t(R)$, and attempts to mix the commodities between $L$ and $R$.
The mixing is done through two ``matchings'', one from $L$ to $R$ and the other from $R$ to $L$ that indicate how much commodity currently at $\ell \in L$ is to be sent to $r \in R$ and vice versa.
Ideally, the two matchings should be such that each $v \in A_t$ receives and sends out exactly $\Bd_t(v)$ units of commodities.
However, this is not always possible in directed graphs.
Instead, the algorithm computes $\Bd_{t+1}(v) \leq \Bd_t(v)$ as the minimum between the units $v$ receives and sends out, and once $\Bd_{t+1}(v) < \Bd(v)/2$, $v$ is completely ``removed'', i.e., we set $\Bd_{t+1}(v) \gets 0$.
After this, the flow matrix is updated by sending flows between $L$ and $R$ through the matchings (see \cite[Section 5.1]{FleischmannLL25}).
Finally, at the end of $O(\log^2 n)$ iterations, it is shown that $\BF$ ``mixes'' well enough in the graph, proving expansion of $\Bd$ on the alive vertices $A_t$.

For more details, see \cite[Section 5]{FleischmannLL25} and in particular \cite[Algorithm 1]{FleischmannLL25}.

\paragraph{Comparison Between the Two Presentations.}
We now compare the presentations of the cut-matching game in our paper and \cite{FleischmannLL25}.
To start, \cite[Theorem 5.1]{FleischmannLL25} is described as a combination of the above abstract game with a specific matching player implemented by running the standard unweighted push-relabel algorithm (in particular \cite[Lemma 5.19]{FleischmannLL25}) on the input graph to \cite[Theorem 5.1]{FleischmannLL25}.
This matching player works by attempting to route from each $\ell \in L$ (with $\Bd_t(\ell)$ units of source demand) to each $r \in R$ (with $\Bd_t(r)$ units of sink demand), and it will compute a flow $\Bf$ partially routing the demand and a sparse cut $S$ such that $\Bf$ rotues the demand in $V\setminus S$ almost perfectly, up to $O(\Bd_t(S))$ units of source/sink.
A matching can be extracted from this flow through standard path decomposition, similar to what we did in \cref{lemma:matching-player}.

The main difference is that for their end result, \cite[Theorem 5.1]{FleischmannLL25} wants to find a sequence of \emph{non-overlapping} cuts in the input graph such that the uncut part is expanding (in the whole graph).
To do so, even though a vertex $v \in S$ might still be partially matched by $\Bf$, the algorithm will completely remove it (i.e., set $\Bd_{t+1}(v)\gets 0$ for $v \in S$).
For other vertices, some of them might not fully send/receive $\Bd_t(v)$ units of commodities due to the flow $\Bf$ not routing the demand perfectly.
The algorithm of \cite{FleischmannLL25} proceeds to set $\Bd_{t+1}(v)$ to be the minimum between the amount $v$ sends and receives, and removes those with $\Bd_{t+1}(v) < \Bd(v)/2$, as described earlier.
Having the vertices in the cut completely removed allows them to add the cut $S$ to their sequence, and then solve the flow problems of future iterations \emph{only on the uncut part of the graph} to make the cuts non-overlapping.

Our formulation differs from \cite{FleischmannLL25} mainly in that our cut-matching game is not concerned about the input graph, hence there is no notion of sparse cuts for \cref{lemma:non-stop-cut-matching}.
In particular, given $\Bd$, the $\Bd^\prime$ vector that we use (in \cref{def:matching}) in each iteration will simply be the $\Bd_t$ vector in the cut-matching game, and the bidirectional matchings will naturally define the matching matrix $\BM_t$ of \cite{FleischmannLL25} to update $\BF_t$.
Crucially, since there is no notion of cuts in the abstract form of the cut-matching game that we use, we \emph{do not} remove vertices in the cut like \cite{FleischmannLL25}.
Instead, all vertices $v$ just set their $\Bd_{t+1}(v)$ based on the amount they receive and send, and only those with $\Bd_{t+1}(v)<\Bd(v)/2$ are removed.
It is easy to see this does not invalidate the proofs in \cite{FleischmannLL25}.
The final measure $\tilde{\Bd}$ in \cref{lemma:non-stop-cut-matching} will simply be the final alive demand $\Bd_t$.
Since we ensure that each bidirectional matching is perfect up to some $O(1/\log^2 n) \cdot \Bd(V)$ units of demand, over $O(\log^2 n)$ iterations we only remove a small constant fraction from $\Bd_t(V)$.
The removal of $\Bd_t(V)$ caused by completely setting $\Bd_{t}(v)$ to $0$ when it reaches below $\Bd(v)/2$ can also be charged to the former quantity, and thus we can guarantee a bound of $\tilde{\Bd}(V)\geq \frac{7}{8}\Bd(V)$ in \cref{lemma:non-stop-cut-matching}.

Our formulation also differs from \cite{FleischmannLL25} in where and how the final subset of demand is expanding.
In \cite[Theorem 5.1]{FleischmannLL25}, since it is explicit in the theorem setup that the matchings are embeddable with low congestion in the input graph $G$, the final demand is $O(\phi/\log^2 n)$-expanding in $G$, where $1/\phi$ is the congestion of each embedding.
Inspecting the proof of this fact (see \cite[Section 5.4]{FleischmannLL25}) and as standard in the cut-matching game literature, that $\Bd_t$ is $O(\phi/\log^2 n)$-expanding in $G$ is because each $\Bd_t$-respecting demand can be routed by (1) following the flow matrix $\BF_t$ (which ``mixes'' well at the end of the game)---indeed, \cite[Page 26, Proof of Lemma 5.26]{FleischmannLL25} defined an explicit flow routing the demand---and (2) using the fact that $\BF_t$ can be routed in $G$ through the embedding with low congestion $O(\log^2 n /\phi)$.

Since our formulation does not involve the underlying graph $G$, we simply say that $\BF_t$ can be routed with congestion $O(1)$ in the witness $W \defeq M_1 \cup M_2 \cup \cdots \cup M_k$.
Indeed, based on how $\BF_t$ is updated (see \cite[Section 5.1]{FleischmannLL25}), the graph $W$ is defined so that $\BF_t$ can be trivially routed with congestion $O(1)$---the routing just follows how $\BF_t$ is ``mixed'' between $L$ and $R$ in each iteration.
This establishes the expansion guarantee listed in \cref{lemma:non-stop-cut-matching}.
Finally, we need a low-hop guarantee in this routing on $W$.
This also follows fairly simply from the definition of $W$.
The routing through $W$ is obtained as $\BF_t$ evolves as new matchings are added: In each iteration we simply send flow through the newly added matching $M_t$ to update $\BF_t$.
Since there are only $T_{\KRV} = O(\log^2 n)$ matchings, the hop of the routing is naturally $T_{\KRV}$.

\section{Omitted Proofs} \label{appendix:omitted}

In this section we provide proofs that are omitted from the main body of the paper.

\FlowOnShortcut*

We will use the following lemma that finds a path decomposition of a flow and the weights of the paths.
Computing a path decomposition in capacitated graphs is fairly standard using a dynamic tree, and by simply augmenting the dynamic tree to maintain the weights we can obtain the desired property.

\begin{lemma}[\cite{SleatorT83}]
  Given a flow $\Bf$ on an $m$-edge graph $G$ with edge weights $\Bw$, there is an $O(m\log m)$ time algorithm that computes the representation $\{(\lambda_i, s_i, t_i)\}_{i\in[k]}$ together with $\{W_i\}_{i\in [k]}$ for some path decomposition $(C_{\Bf}, \{(\lambda_i, P_i)\}_{i\in[k]})$ of $\Bf$, where $\Bw(P_i) \leq W_i$.
  \label{lemma:path-decomposition}
\end{lemma}

\begin{proof}[Proof of \cref{lemma:flow-on-shortcut-short-path}]
  We run \cref{lem:push-relabel-on-shortcut} multiple iterations, each time routing a large fraction of demand or return the desired cut.
  Let $(\Bsource^{(1)}, \Bsink^{(1)}) \defeq (\Bsource, \Bsink)$ be the initial demand.
  For each iteration $i$, we run \cref{lem:push-relabel-on-shortcut} to route $(\Bsource^{(i)}, \Bsink^{(i)})$ in $G_A$, getting a flow $\Bf^{(i)}$.
  If $|\Bf^{(i)}| \geq \|\Bsource^{(i)}\|_1/2$,
  then we run \cref{lemma:path-decomposition} on $\Bf^{(i)}$ with the weight function $\Bw_\cH$.
  Let $h_{\ref{lem:push-relabel-on-shortcut}} = O(n\cdot(\kappa+1/\psi)\cdot\log n)$ be the bound on the average weight of flow paths returned by \cref{lem:push-relabel-on-shortcut}.
  Let $\tilde{\Bf}^{(i)} \defeq \sum_{i \in [k]: W_i \leq 2h}\lambda_i \Bf_{P_i}$ consists of only paths with weight bounded by $2h$.
  By Markov's inequality, $\tilde{\Bf}^{(i)}$ routes at least half of the demand that $\Bf^{(i)}$ does, i.e., $|\tilde{\Bf}^{(i)}| \geq \|\Bsource^{(i)}\|_1/4$.
  We set $(\Bsource^{(i+1)}, \Bsink^{(i+1)}) \defeq (\Bsource^{(i)}_{\tilde{\Bf}^{(i)}}, \Bsink^{(i)}_{\tilde{\Bf}^{(i)}})$ be the residual demand (which is $\psi$-integral since $\tilde{\Bf}^{(i)}$ is) and continue to the next iteration.
  If all the $O(\log n)$ executions of \cref{lem:push-relabel-on-shortcut} route at least half of the demand (and thus, the corresponding $\tilde{\Bf^{(i)}}$ routes at least a quarter), then $\Bf \defeq \tilde{\Bf}^{(1)} + \cdots + \tilde{\Bf}^{(O(\log n))}$ routes $(\Bsource, \Bsink)$ with congestion $O(\kappa \log n)$.
  Additionally, $\Bf$ trivially admits a $(\Bw_{\cH}, 2h_{\ref{lem:push-relabel-on-shortcut}})$-short path decomposition since each $\tilde{\Bf}^{(i)}$ by construction does.
  The representation of this decomposition can also be easily computed from that of each $\tilde{\Bf}^{(i)}$ (which we obtained from \cref{lemma:path-decomposition}).
  
  Otherwise, we have $|\Bf^{(i)}| < \|\Bsource^{(i)}\|/2$ and get a cut $S^{(i)}$ from \cref{lem:push-relabel-on-shortcut}.
  We have
  \begin{align*}
    \Bc(E_{G_{A}}(S^{(i)},\overline{S^{(i)}}))
    &\leq \frac{41|\Bf^{(i)}|+\min\{\vol_{F}(S^{(i)}),\vol_{F}(\overline{S^{(i)}})\}}{\kappa} \\
    &\leq \frac{41\min\{\Bsource^{(i)}(S^{(i)}), \Bsink^{(i)}(\overline{S^{(i)}})\}+\min\{\vol_{F}(S^{(i)}),\vol_{F}(\overline{S^{(i)}})\}}{\kappa} \\
    &\leq \frac{41\min\{\Bsource(S^{(i)}), \Bsink(\overline{S^{(i)}})\}+\min\{\vol_{F}(S^{(i)}),\vol_{F}(\overline{S^{(i)}})\}}{\kappa},
  \end{align*}
  where we used that $|\Bf^{(i)}| \leq \Bsource^{(i)}(S^{(i)})$ since $\Bsource^{(i)}(S) \geq \Bsource^{(i+1)}(S^{(i)}) = \|\Bsource^{(i+1)}\|_1$ while $|\Bf^{(i)}| < \|\Bsource^{(i)}\|_1/2 < \|\Bsource^{(i+1)}\|_1$, and that $|\Bf^{(i)}| \leq \Bsink^{(i)}(S^{(i)})$ since $\Bsink^{(i+1)}(S^{(i+1)}) = 0$ and therefore $\Bsink^{(i)}(\overline{S^{(i+1)}}) \geq \Bsink^{(i+1)}(\overline{S^{(i+1)}}) \geq \|\Bsource^{(i+1)}\|_1 > |\Bf^{(i)}|$.
  Moreover, we have $(\Bsource^{(i+1)}, \Bsink^{(i+1)}) = (\Bsource_{\Bf}, \Bsink_{\Bf})$ for $\Bf \defeq \tilde{\Bf}^{(1)} + \cdots + \tilde{\Bf}^{(i)}$, and thus $\Bf$ routes all sources except than those from $S^{(i)}$, and $S^{(i)}$ is fully saturated.
  Again, $\Bf$ admits a $(\Bw_{\cH}, 2h_{\ref{lem:push-relabel-on-shortcut}})$-short path decomposition which we can compute from that of $\tilde{\Bf}^{(i)}$.
  The congestion of $\Bf$ in this case is, similar to above, $O(\kappa \log n)$.
\end{proof}

\end{document}